\documentclass[12pt]{article}
{\begin{Sbox}\begin{minipage}}%
{\end{minipage}\end{Sbox}\fbox{\TheSbox}}%
\usepackage{bbm}
\usepackage{ulem}
\usepackage{color}
\usepackage{ulem}
\usepackage{adjustbox}
\usepackage[psamsfonts]{amssymb}
\usepackage{amsmath, amsthm, anysize, enumerate, color, url}
\usepackage{graphicx}
\usepackage{graphicx}
\usepackage{epstopdf}
\usepackage{epsfig}
\usepackage{subfigure}
\usepackage[ruled,linesnumbered]{algorithm2e}
\usepackage[sectionbib]{natbib}
\usepackage{authblk}
\usepackage{threeparttable}

\usepackage[colorlinks, breaklinks,
                   linkcolor=red,
                  citecolor=blue
                  ]{hyperref}

\usepackage{breakurl}
\renewcommand{\hat}{\widehat}
\newtheorem{theorem}{Theorem} %%%[section]
\newtheorem{proposition}{Proposition} %%%[theorem]

\theoremstyle{definition}
\newtheorem{remark}{Remark}  %%%[theorem]

\newtheorem{condition}{Condition}

\usepackage{setspace}

\begin{document}

\title{Time-varying $\beta$-model for dynamic directed networks}

\author[1]{Yuqing Du}
\author[1]{Lianqiang Qu  \thanks{Correspondence: {\tt qulianq@ccnu.edu.cn}}}
\author[1]{Ting Yan}
\author[2]{Yuan Zhang}
% \contrib{All authors contributed equally.}

% Include full affiliation details for all authors
\affil[1]{School of Mathematics and Statistics, Central China Normal University, Wuhan, Hubei, 430079, China}
\affil[2]{Department of Statistics, The Ohio State University, Columbus, Ohio, 43210, USA}

%\corremail{qulianq@ccnu.edu.cn}

%\runningauthor{Du, Qu, Yan and Zhang}

\date{}

\maketitle

\begin{abstract}
\begin{spacing}{1.2}
We extend the well-known $\beta$-model for directed graphs to dynamic network setting, where we observe snapshots of adjacency matrices at different time points.  We propose a kernel-smoothed likelihood approach for estimating $2n$ time-varying parameters in a network with $n$ nodes, from $N$ snapshots.
We establish consistency and asymptotic normality properties of our kernel-smoothed estimators as either $n$ or $N$ diverges.
Our results contrast their counterparts in single-network analyses, where $n\to\infty$ is invariantly required in asymptotic studies.
We conduct comprehensive simulation studies that confirm our theory's prediction and illustrate the performance of our method from various angles.
We apply our method to an email data set and obtain meaningful results.
\end{spacing}

\vskip 5 pt \noindent
\begin{spacing}{1.4}
\textbf{Key words}:   $\beta$-model; Directed networks; Dynamic networks; Kernel smoothing
\end{spacing}

\end{abstract}

\vskip 5pt

\section{Introduction}
\label{section::introduction}
In statistical network analysis, node degrees play a fundamentally important role, both for its meaningfulness in modeling and high computational efficiency in parametric and nonparametric inference \citep{zhang2022edgeworth}.
Among existing literature, a family of degree-driven models, namely, the $\beta$-model and its variants, received extensive attention and research interest \citep{chatterjee2011random,hillar2013maximum,olhede2012degree}.  Compared to other popular network models such as stochastic block model, the $\beta$-model provides a simple yet elegant way to characterize networks with potentially high heterogeneity.
A particularly attractive advantage of the $\beta$-model is its convenient and highly efficient parameter estimation \citep{chen2019analysis}, see also \citet{hillar2013maximum,zhang2021L2}.
An incomplete list of notable results also includes:
MLE existence \citep{ rinaldo2013maximum},
central limit theorems \citep{yan2013central},
directed and bipartite $\beta$-model \citep{yan2016asymptotics, fan2022asymptotic},
weighted edges \citep{hillar2013maximum}, incorporating covariates \citep{graham2017econometric},
regularization \citep{chen2019analysis, stein2021sparse, zhang2021L2} and so on.
But existing literature on the $\beta$-model almost exclusively focused on static networks;
whereas time-varying data appear in many applications, such as dynamic email networks and online social networks \citep{cui-chen-2023}.
Therefore, it is certainly of interest to build a time-varying $\beta$-model for dynamic networks.

In this paper, we fill in the significant blank in the modeling, methodology and theory of the $\beta$-model for addressing dynamic networks.
Our contributions are three-folds.
First, we propose a time-varying directed $\beta$-model
for characterizing the evolution of dynamic bi-degrees in directed dynamic networks by extending the static $\beta$-model.
The model contains $2n$ time-varying parameters, for which each node is attached to an out-degree varying parameter and
an in-degree varying parameter.
Second, we propose a kernel-smoothed likelihood approach to estimate $2n$ unknown parameter functions by
borrowing information from nearby time points.  This is inspired by the classical $\beta$-model estimation methods for static networks and kernel methods in classical nonparametric statistics.
Third, we establish consistency and asymptotic normality of our proposed estimator.
In contrast to the results for static directed networks \citep{yan2016asymptotics}, our analysis reveals that borrowing information from nearby time points can significantly improve estimation accuracy and relax the assumptions for consistency and asymptotic normality. The results hold even when $n$ stays fixed, as long as the number of observed times goes to infinity.
For a sparse static network with zero in-degrees or zero out-degrees, the MLE in the directed $\beta$-model does not exist. Our kernel smoothed
estimator eliminates this case as long as those nodes are connected to other nodes in the full observation period.
Numerical studies and a real data application demonstrate our theoretical findings.

The rest of this paper is organized as follows.  We first set up some notation;
Section \ref{section::our-method} describes our proposed model, devises estimation methods and inference procedures, and establishes theoretical justifications;
in Section \ref{section::simulations}, we conduct comprehensive simulation studies to address the question of selecting the tuning parameter, assess the performance of our method from different aspects, and validate our theoretical prediction;
in Section \ref{section::data-example}, we apply our method to a real-world data set and interpret the results;
we conclude our paper with some discussion in Section \ref{section::discussion}.

\subsection{Notation}
\label{subsec::notation}
\begin{sloppypar}

We define the notation used in this paper.
For a vector $\mathbf{x}=(x_1, x_2,\ldots x_n)^\top$, define $\| \mathbf{x} \|_{\infty} =\max_{1 \leq i \leq n} |x_i|$.
For a matrix $V:=(v_{i,j})$,
define $\|V\|_\infty$ to be
$
    \|V\|_{\infty}
    = \sup_{\mathbf{x}\neq0}
    \| V\mathbf x \|_{\infty}/
    \| \mathbf x \|_{\infty}= \max_{1\le i\le n}\sum_{j=1}^{n}| V_{i,j} |.
$
Define $\|V\|_{\max}:=\max_{1\leq \{i,j\}\leq n}|V_{i,j}|$.
We will need to approximate the inverse of some $V$ matrix of certain structures.  Specifically, following \citet{yan2016asymptotics}, we call a $(2n-1) \times (2n-1)$ matrix $V = (v_{i,j})$ to belong to the matrix class ${\mathcal L}_n(m,M)$, if $V$ satisfies:
$m \le v_{i,j} - \sum_{j=n+1}^{2n-1}v_{i,j} \le M,\quad i=1, \dots , n-1$;
$v_{n,n} = \sum_{j=n+1}^{2n-1}v_{nj}$;
$v_{i,j} = 0,\quad i,j=1, \dots , n, i \neq j$;
$v_{i,j} = 0,\quad i,j = n+1, \dots, 2n-1,i \neq j$;
$m \le v_{i,j} = v_{ji} \le M,\quad i = 1, \dots, n,j = n+1, \dots, 2n-1, j \neq n+i$;
$v_{i,n+i} = v_{n+i,i} = 0,\quad i = 1, \cdots, n-1$;
and
$v_{i,i} = \sum_{k=1}^{n}v_{ki} = \sum_{k=1}^{n} v_{ik},\quad i = n+1, \cdots, 2n-1$.
If $V \in {\mathcal L}_n(m,M)$, then $V$ is symmetric, element-wise non-negative and diagonally dominant, thus invertible.
Moreover, by \citet{yan2016asymptotics}, $V^{-1}$ can be well-approximated by $S=\mathcal{S}(V)$, as follows
\begin{equation}\label{definition-s}
    s_{i,j}=\left\{
    \begin{aligned}
        &\frac{\delta_{i,j}}{v_{i,i}}+ \frac{1}{v_{2n,2n}}, &i&, j = 1, \dots, n,\\
        &-\frac{1}{v_{2n,2n}}, &i&= 1, \dots, n,~j = n+1, \dots, 2n-1,\\
        &-\frac{1}{v_{2n,2n}}, &i&= n+1, \dots, 2n-1,~j = 1, \dots, n,\\
        &\frac{\delta_{i,j}}{v_{i,i}} + \frac{1}{v_{2n,2n}}, &i&, j = n+1, \dots, 2n-1,
    \end{aligned}
    \right.
\end{equation}
where $\delta_{i,j}:=\mathbbm{1}_{[i=j]}$ and
$v_{2n,i} = v_{i,2n} := v_{ii} - \sum_{j=1,j \neq i}^{2n-1} v_{ij}$ for $i = 1, \dots, 2n-1$ and $v_{2n,2n} = \sum_{i=1}^{2n-1} v_{2n,i}$.
\end{sloppypar}
\section{Our method}
\label{section::our-method}
\subsection{Data structure and our model}

We start with presenting the data structure.
We observe network snapshots at $N$ randomly selected time points $T_1,\ldots,T_N$ selected by a probability density function $f_T(t)$ with sample space $[a,b]$ for some constants $a<b$.
All snapshots share a common node set $[n]=\{1,\ldots,n\}$ across different time points.
At each time point $t\in {\cal T}:=\{T_1,\ldots,T_N\}$, we observe a directed binary network, represented by its adjacency matrix $A(t):=\{A_{i,j}(t)\}_{1\leq \{i, j\}\leq n}\in \{0,1\}^{n\times n}$, where $A_{i,j}(t)=1$ if there is a directed edge from $i$ to $j$ at time $t$, and $A_{i,j}(t)=0$ otherwise.
For simplicity, we assume all edges are generated independently from each other.
We assume no self-loops: $A_{i,i}(t)\equiv0$.

To model this data structure, we propose a time-varying $\beta$-model for directed networks.
At any time $t\in {\cal T}$, each node $i$ is associated with two parameters $\alpha_i(t)$ and $\beta_i(t)$, encoding the strengths of the sender's effect and receiver's effect, respectively.  The edge probability from $i$ to $j$ is
\begin{equation}
    \label{def::directed-beta-model}
    P\big(A_{i,j}(t)=1\big)
    =
    \dfrac{e^{\alpha_i(t)+\beta_j(t)}}{1+e^{\alpha_i(t)+\beta_j(t)}}
    \quad
    \textrm{and}
    \quad
    P\big(A_{i,j}(t)=0\big)
    =
    1-P\big(A_{i,j}(t)=1\big).
\end{equation}
Similar to directed and bipartite $\beta$-models for static networks \citep{yan2016asymptotics, fan2022asymptotic},
here, we need to enforce an additional regularity conditions to ensure parameter identifiability.
For simplicity, we set $\beta_n(t)=0$ \citep{yan2016asymptotics}.
This will also guarantee the uniqueness of the solution to our estimation equation set \eqref{eqn::our-estimation-equation-1} and \eqref{eqn::our-estimation-equation-2},
which we shall present in Section \ref{subsec::parameter-estimation}.
Let
\begin{equation*}
\theta(t):=\Big( \alpha_1(t),\ldots,\alpha_n(t),\beta_1,\ldots,\beta_{n-1}(t) \Big)^\top,
\end{equation*}
to denote all free parameters at $t$.
% Observe a simple fact that at each time $t$, the sum of all in-degrees equals the sum of all out-degrees.
% Naturally, we impose the constraint that for each $t$,
% \begin{equation}
%     \label{eqn::parameter-constraint-beta-n}
%     \sum_{i=1}^n \alpha_i(t) = \sum_{j=1}^n \beta_j(t).
% \end{equation}
% It is not difficult to prove that \eqref{eqn::parameter-constraint-beta-n} is sufficient to guarantee parameter identifiability.
% This will also guarantee the uniqueness of the solution to our estimation equation set \eqref{eqn::our-estimation-equation-1} and \eqref{eqn::our-estimation-equation-2}, which we shall present in Section \ref{subsec::parameter-estimation}.
% For notation simplicity, set
% \begin{equation*}
% \theta(t):=\Big( \alpha_1(t),\ldots,\alpha_n(t),\beta_1,\ldots,\beta_{n-1}(t) \Big)^\top
% \end{equation*}
% to denote all free parameters at $t$.
% Notice that $\theta(t)$ does not include $\beta_n(t)$ due to \eqref{eqn::parameter-constraint-beta-n}.

Next, we describe our assumption on the relationship between the model parameters at different time points.
We naturally anticipate that $\theta(t)$ and $\theta(t')$ should be similar for close-by time points $t\approx t'$, unless there exists a change point in between;
while those parameters at distant time points might possibly be very different (or they might not).
For simplicity, we stick to the continuous $\theta(t)$ setting in our method development and analysis.
We will discuss how our method can be slightly tweaked to handle change points at the end of Section \ref{section::discussion}.

\subsection{Parameter estimation}
\label{subsec::parameter-estimation}

Our estimation method is semi-parametric.
The parametric flavor of our approach is reflected by
\eqref{def::directed-beta-model};
whereas the nonparametric flavor lies in that we utilize the smoothness of $\theta(t)$ over $t$ to enhance estimation accuracy.
Given the observed time points ${\cal T}$, the log-likelihood function is
\small{
\begin{align}
    \label{def::data-partial-likelihood}
    {\cal L}_0(\{(t,\theta(t))\}_{t\in {\cal T}})
    = &~
    \log\Bigg[
    \prod_{t\in {\cal T}}
    \prod_{(i,j): 1\leq \{i\neq j\}\leq n}
    \dfrac
        {e^{A_{i,j}(t)\{ \alpha_i(t) + \beta_j(t) \}}}
        {1+e^{\alpha_i(t) + \beta_j(t)}}
    \Bigg]\notag\\
    = &~
    \sum_{t\in {\cal T}}
    \bigg\{
        \sum_{i=1}^n  \Big(
            d_i(t) \alpha_i(t)
            +
            b_i(t) \beta_i(t)
        \Big)
        -
        \sum_{(i,j):1\leq \{i\neq j\}\leq n}
        \log\Big(
            1+e^{\alpha_i(t)+\beta_j(t)}
        \Big)
    \bigg\},
\end{align}}
where $d_i(t)=\sum_{j:1\leq j\leq n, j\neq i}A_{i,j}(t)$ and $b_j(t)=\sum_{i:1\leq i\leq n, i\neq j}A_{i,j}(t)$ denote the out-degree of node $i$ and the in-degree of node $j$, respectively.
Straightly applying a maximum-likelihood method on \eqref{def::data-partial-likelihood}, separately for each time point, would produce an estimation of $\theta(t)$ using only the $A(t)$ at time $t$.  That is, for every $i\in[n]$ and $j\in[n-1]$, define
\begin{align}
    F_{i;0}(t,\theta(t))
    := &~
    d_i(t) -
    \sum_{\substack{j:1\leq j\leq n\\j\neq i}}\dfrac{e^{\alpha_i(t)+\beta_j(t)}}{1+e^{\alpha_i(t)+\beta_j(t)}},
    \quad
    t\in \{T_1,\ldots,T_N\}
    \label{eqn::original-estimation-equation-1}
    \\
    F_{n+j;0}(t,\theta(t))
    := &~
    b_j(t) -
    \sum_{\substack{i:1\leq i\leq n\\i\neq j}}\dfrac{e^{\alpha_i(t)+\beta_j(t)}}{1+e^{\alpha_i(t)+\beta_j(t)}},
    \quad
    t\in \{T_1,\ldots,T_N\}
    \label{eqn::original-estimation-equation-2}
    \\
    \nonumber
    F_0(t, \theta(t)):=&~(F_1(t, \theta(t)), \ldots, F_{2n-1}(t, \theta(t)) )^\top.
\end{align}
This approach estimates $\theta(t)$ by the solution to $F_0=0$.
Then one can directly apply the algorithm and theory of \citet{yan2016asymptotics}.
We call this method ``point-wise estimation'' throughout this paper and will use it as a benchmark for comparison in our simulation studies.

Our goal is to estimate $\theta(t)$ for all $t\in [a,b]$, not just at those observed time points $t\in {\cal T}$.
This could not be achieved by the point-wise estimator \eqref{eqn::original-estimation-equation-1} and \eqref{eqn::original-estimation-equation-2}.
Our main idea to estimate $\theta(t)$ for all $t\in[a,b]$ is to borrow information from $T_\ell$'s close to $t$ -- in fact, this will also improve the estimation at time points $t\in {\cal T}$.
Bearing this in mind, we generalize \eqref{eqn::original-estimation-equation-1} and \eqref{eqn::original-estimation-equation-2}, replacing the observed out- and in-degrees $d_i(t)$ and $b_i(t)$ at time $t$ by their kernel-smoothed versions incorporating information from nearby time points.
More precisely speaking, define
\begin{align}
    F_i(t,\theta(t))
    := &~
    \sum_{\ell\in [N]} K_h(t-T_\ell)
    \Bigg\{
        d_i(T_\ell)
        -
        \sum_{\substack{j:1\leq j\leq n\\j\neq i}}\dfrac{e^{\alpha_i(t)+\beta_j(t)}}{1+e^{\alpha_i(t)+\beta_j(t)}}
    \Bigg\}
    ,
    \quad
    t\in [a,b],
    1\leq i\leq n
    \label{eqn::our-estimation-equation-1}
    \\
    F_{n+j}(t,\theta(t))
    := &~
    \sum_{\ell\in [N]} K_h(t-T_\ell)
    \Bigg\{
        b_j(T_\ell)
        -
        \sum_{\substack{i:1\leq i\leq n\\i\neq j}}\dfrac{e^{\alpha_i(t)+\beta_j(t)}}{1+e^{\alpha_i(t)+\beta_j(t)}}
    \Bigg\}
    ,
    \quad
    t\in [a,b],
    1\leq j\leq n-1
    \label{eqn::our-estimation-equation-2}
\end{align}
where the kernel function
$
    K_h(u) := (1/h) \cdot K\big( u/h \big)
$
is even, supported on $[-1,1]$ and satisfies $\int_{\mathbb{R}}K(u)du =1$.

Our estimator, denoted by $\hat\theta(t)$, is the solution to the estimation equation set
$
    F(t,\theta(t)) = 0.
$
For a fluent narration, we relegate the choice of the bandwidth $h$ to Section \ref{section::simulation-1} and focus on solving \eqref{eqn::our-estimation-equation-1} and \eqref{eqn::our-estimation-equation-2}, for which we shall employ Newton's method.
The Jacobian matrix of $F(t,\theta(t))/(Nn)$, denoted by $V(t,\theta(t))$, is
\begin{align*}
    V(t,\theta(t))
    := &
    \frac{1}{N}\sum_{\ell\in [N]} K_h(t-T_\ell)
    V_0(t,\theta(t)),
\end{align*}
where $V_0(t,\theta(t)) = V_0(t,(\alpha(t),\beta(t)))$ is further defined as follows
\begin{align*}
    \{V_0(t,\theta(t))\}_{i, n+j}
    = &
    \{V_0(t,\theta(t))\}_{n+i, j}\notag
    :=
    \frac{1}{n}\dfrac{e^{\alpha_i(t)+\beta_j(t)}}{\big\{1+e^{\alpha_i(t)+\beta_j(t)}\big\}^2},
    & i\in [n],~ j\in [n-1],~ i\neq j,
    \\
    \{V_0(t,\theta(t))\}_{i, i}
    := &
    \sum_{\substack{i:1\leq i\leq n\\j\neq i}}
    \{V_0(t,\theta(t))\}_{i, n+j},
    & i\in [n],
    \\
    \{V_0(t,\theta(t))\}_{n+j, n+j}
    := &
    \sum_{\substack{j:1\leq j\leq n-1\\i\neq j}}
    \{V_0(t,\theta(t))\}_{i, n+j},
    & j\in [n-1],
\end{align*}
and $\{V_0(t,\theta(t))\}_{i, j}=0$ otherwise.
Since $V(t,\theta(t))$ is element-wise non-negative and strictly diagonally dominant, it is positive definite.  Therefore, when the solution to the estimation equations \eqref{eqn::our-estimation-equation-1} and \eqref{eqn::our-estimation-equation-2} exists, it must be unique and can be found by a gradient descent or Newton's method \citep{MAL-050}.
We will address the existence of $\hat\theta(t)$ and formally characterize its accuracy in Theorem \ref{theorem-1}.

\subsection{Theoretical properties of the estimator}
Recall that our estimator $\hat\theta(t)$ is obtained by solving \eqref{eqn::our-estimation-equation-1} and \eqref{eqn::our-estimation-equation-2}.
In this section, we establish three aspects of theoretical guarantee for $\hat\theta(t)$:
existence, error rate and asymptotic normality.
Let $\theta^*(t):=\big( \alpha_1^*(t),\ldots,\alpha_n^*(t), \beta_1^*(t),\ldots,\beta_{n-1}^*(t) \big)^\top$ be the true parameters.
Define
\begin{equation*}
    Q_{Nnh}
    =
    \sup_{t\in[a,b]}
    \max_{i,j:1\leq \{i\neq j\}\leq n}
    \dfrac{(1+e^{\alpha^*_i(t)+\beta^*_j(t)})^2}{e^{\alpha^*_i(t)+\beta^*_j(t)}}.
\end{equation*}
Readers familiar with the $\beta$-model literature may immediately notice the analogy between
our $Q_{Nnh}$ and its counterpart, usually denoted by $b_n$ \citep{fan2022asymptotic}, in the static network setting.
Now we describe some mild regularity conditions that we would need.
\begin{condition}
    \label{condition-1}
    The parameters $\theta^*(t)$ is element-wise twice continuously differentiable.
    Moreover, there exists a constant or diverging deterministic series $C_{Nn}$ such that
    \begin{equation*}
        \max_{i\in[2n-1]}\left|\dfrac{d \theta_i^*(t)}{d t}\right| \leq C_{Nn}
        \quad \textrm{and} \quad
        \max_{i\in[2n-1]}\left|\dfrac{d^2 \theta_i^*(t)}{d t^2}\right| \leq C_{Nn},
        \label{condition::theta-derivative-bound}
    \end{equation*}
    where $C_{Nn}$ satisfies
    \begin{equation*}
        C_{Nn} h \to 0~~\text{as}~~Nn\to \infty.
        \label{condition-3-b}
    \end{equation*}
\end{condition}

\begin{condition}
    \label{condition-2}
    Assume $f(t)>0$ for all $t \in [a,b],$ and is twice continuously differentiable.
\end{condition}
\begin{remark}
Condition \ref{condition-1} allows $\theta^*(t)$ to change increasingly rapidly as $Nn\to\infty$.  This is not surprising since a growing $Nn$ brings an increasing amount of information.
Condition \ref{condition-2} is a very mild condition that we choose observation times from $[a,b]$ in a balanced fashion, so no part of the entire time period $[a,b]$ is ignored.
Let $k_{21}:=\int v^2 K(v)dv<\infty.$
Now we are ready to present the main theorem on the existence and uniform consistency of our estimator.
\end{remark}

\begin{theorem}
    \label{theorem-1}
    Suppose Conditions \ref{condition-1} and \ref{condition-2} hold.
    Assume $h \rightarrow 0, Nh\rightarrow \infty$ and either $n$ remains constant or $n\to\infty$.  If
    \begin{equation}
        Q_{Nnh}=o\Bigg\{\frac{1}{\big(\sqrt{C_1\log (Nnh)/(Nnh)}+C_2C_{Nn}h^2\big)^{1/6}}\Bigg\},
        \label{eqn::condition-Q-theorem-1}
    \end{equation}
    then the estimator $\hat\theta(t)$ exists and satisfies
    \begin{align}
        \sup_{t \in [a,b]}\| \hat \theta(t)-\theta^*(t) \|_{\infty}
        =
        O_p\Bigg\{Q_{Nnh}^3\Bigg(\sqrt{C_1\frac{\log (Nnh)}{Nnh}}+C_2C_{Nn}h^2\Bigg)\Bigg\},
        \label{eqn::theorem-1-result}
    \end{align}
    where $C_1>6$ is a constant and $C_2>k_{21}\sup_{t\in[a,b]}(f'(t)+f(t)/2)$.
\end{theorem}
The condition \eqref{eqn::condition-Q-theorem-1} in Theorem \ref{theorem-1} appears stronger than what is needed to send the right hand side of \eqref{eqn::theorem-1-result} to zero.
One may naturally wonder if it can be substantively improved.  The answer is ``not easily".
The precise reason is technical, but a quick explanation is that \eqref{eqn::condition-Q-theorem-1} not only contributes to uniform consistency, but also the existence of the solution $\hat\theta(t)$.
It turns out that the existence of solution demands stronger assumptions on $Q_{Nnh}$ in the analysis, but this point is not explicitly reflected in \eqref{eqn::theorem-1-result}.
The phenomenon that the assumption on $Q_{Nnh}$ is stronger than the apparent need (for the error bound to diminish) is frequently reported in many related works on the $\beta$-model, including \citep{chatterjee2011random,yan2016asymptotics, chen2019analysis, stein2021sparse,zhang2021L2}.

The uniform convergence rate in Theorem \ref{theorem-1} has the familiar bias-variance trade-off form, thus the bandwidth should be carefully selected to properly balance bias and variance.
%By Theorem \ref{theorem-1}, we may choose $h$ to minimize the right hand side of \eqref{eqn::theorem-1-result}.
% In practice, we can use the rule of thumb bandwidth.
Theoretically, if we know $C_{Nn}$, then we can optimize the choice of $h$ analytically.
For example, if $C_{Nn}$ is constant, then we can choose $h\asymp (Nn)^{-1/5},$
which ensures that the right hand side of \eqref{eqn::theorem-1-result} is $o_p(1).$
% Here $c$ is some prespecified constant.
% However, the rule of thumb bandwidth may not be adaptive to the data.
But in practice, it might not always be easy to know $C_{Nn}$.
Therefore, in Section \ref{section::simulations}, we develop a leave-one-out cross-validation procedure to select $h$.
We find it to perform well in our numerical studies.

% the cross-validation method may be time-consuming.

%{\color{cyan}[Yuan says: Lianqiang, please rewrite the sentence on selecting $h$.]}
% In particular, when $(C^2_2C^2_{Nn}Nn h^5)/\log(Nnh)\rightarrow 0, h \to 0$ and $Nnh \rightarrow \infty,$ the bias in Theorem \ref{theorem-1} is bounded uniformly by $O_p\Big(Q_{Nnh}^3\sqrt{\log (Nnh)/Nnh}\Big).$
% {\color{red}
% Notice that Theorem \ref{theorem-1} can hold for a fixed $n$.}

\begin{remark}
The classical literature on the $\beta$-model typically assumes very dense networks, see the overviews of this issue in \citet{chen2019analysis,zhang2021L2} and the references therein.  In our work, to guarantee consistency, we need
$Q_{Nnh}=o\Big\{1/(\sqrt{C_1\log (Nnh)/(Nnh)}+C_2C_{Nn}h^2)^{1/6}\Big\}$.
This yields the following lower bound on network sparsity
\begin{equation}
    \min_{1\le i < j \le n}P\big(A_{i,j}(t)=1\big) \ge \frac{1}{Q_{Nnh}}
    =
    \omega
    \Bigg(\Big\{\sqrt{C_1\frac{\log (Nnh)}{Nnh}}+C_2C_{Nn}h^2\Big\}^{1/6}
    \Bigg).
    % =
    % \omega\Bigg(\frac{\log(Nn)^{1/12}C_{Nn}^{1/30}}{(Nn)^{1/15}}\Bigg).
    \label{remark::density-1}
\end{equation}
To better decipher \eqref{remark::density-1}, let $N = O(n^\psi)$ and notice that
when $\psi>1/4$, we can set $h=O\{(NnC_{Nn})^{-1/5}\}$ and further simplify \eqref{remark::density-1} into
\begin{equation*}
    \Bigg\{\sqrt{C_1\frac{\log (Nnh)}{Nnh}}+C_2C_{Nn}h^2\Bigg\}^{1/6}
    =
    \frac{\log(Nn)^{1/12}C_{Nn}^{1/30}}{(Nn)^{1/15}}.
    \label{discussion::density-2}
\end{equation*}
Consider a simple case where $C_{Nn}=O(1)$.
Compared to the requirements of $\rho_n\gtrsim \log^{-C}n$ for some $C$ in \citet{yan2013central} and $\rho_n\gtrsim n^{-1/12}$ in \citet{yan2016asymptotics}, we accommodate sparser networks (recall $N\gg n^{1/4}$).
Notice that in this simple illustration, our choice of $h$ is still tailored to minimize the right hand side of \eqref{eqn::theorem-1-result}.
If the network becomes sparser, we would need to choose larger $h$ values -- this is understandable, since if we observe too few edges from network snapshots near one time point, the natural thing to do is to expand $h$ to incorporate data from a wider time window; but this of course may inflate bias.
An interesting future work is to handle very sparse networks, for which purpose, some regularization might be necessary \citep{chen2019analysis, stein2021sparse, zhang2021L2}.
% We recall that this choice of $h$ is optimal in the sense of optimizing estimation error, not aimed at minimizing network sparsity requirement.
% Specifically, when $\psi<1/4$, we need to choose an $h$ that diminishes at a slower rate in order to satisfy the assumption $Nh\to\infty$.
% For example, if $N=n^{1/6}$, setting $h=n^{-1/12}$ is valid, and this yields a density requirement of $n^{-13/144}$, still better than $n^{-1/12}$ in \citet{yan2016asymptotics}, but the error rate with $h=n^{-1/12}$ would be improved by setting $h=n^{-1/10}$ instead, while the latter requires a stronger density assumption than \citet{yan2016asymptotics}.
% We clearly see a trade-off between estimation accuracy and strength of assumption on network density.
% Overall, much needs to be done in the future to push towards sparser networks.  Along this direction, some regularity may be necessary \citep{chen2019analysis, stein2021sparse, zhang2021L2}.
\end{remark}

Next, we establish the asymptotic normality of our estimator under mild conditions.
As a preparation, let us set up some shorthand.
Define
$k_{02}:=\int K^2(v)dv$,
%$k_{21}:=\int v^2 K(v)dv$,
$u_{ij}(t):=\big\{e^{\alpha^*_i(t)+\beta^*_j(t)}\big\}/\big\{1+e^{\alpha^*_i(t)+\beta^*_j(t)}\big\}$
and  $\mu_{ij}(t):=u_{ij}'(t)f'(t)+\frac{1}{2}u_{ij}''(t)f(t)$.
To explicitly express the asymptotic distribution formula, define $\bar S(t,\theta(t))=\mathcal{S}(\bar V(t,\theta(t)))$, where recall $\mathcal{S}(\cdot)$ from \eqref{definition-s} and define $\bar V(t,\theta(t))=f(t)V_0(t,\theta(t))$.
Now we are ready to state the result.

\begin{theorem}
\label{theorem-2}
Suppose Conditions \ref{condition-1} and \ref{condition-2} hold.
Assume $Nnh^5C^2_{Nn}Q^2_{Nnh}=o(1)$,
$h \rightarrow 0, Nh \to \infty,$ $n$ either stays constant or diverges, and
\begin{equation*}
Q_{Nnh}=o\Bigg\{\frac{1}{(Nnh)^{1/18}\big(\sqrt{C_1\log (Nnh)/(Nnh)}+C_2C_{Nn}h^2\big)^2}\Bigg\},
\end{equation*}
where $C_1$ and $C_2$ are defined in Theorem \ref{theorem-1}.  Then
for any fixed $1 \le p \le 2n-1$ and $t\in[a,b],$
$\sqrt{Nnh}\Big[\{\hat{\theta}_1(t) - \theta_1^*(t),
\ldots,\hat{\theta}_{p}(t) - \theta_{p}^*(t)\}^\top-k_{21}h^2\{\mu_1(t),\dots,\mu_p(t)\}^\top\Big]$
converges in distribution to a $p$-dimensional multivariate normal, with mean zero and covariance matrix given by the upper left $p\times p$ block of $k_{02}\bar S(t,\theta^*(t))\bar V(t,\theta^*(t))\bar S(t,\theta^*(t))^\top,$
where
\begin{align*}
\mu_q(t)
=&
\Bigg\{\frac{\sum_{j\in [n],j\neq q}\mu_{qj}(t)/n}{f(t)\big\{V_0(t,\theta^*(t))\big\}_{q,q}}\Bigg\}\{1-\mathbbm{1}_{n+1 \le q \le 2n-1}\}
+\Bigg\{\frac{\sum_{j\in [n],j\neq q-n}\mu_{j,q-n}(t)/n}{f(t)\big\{V_0(t,\theta^*(t))\big\}_{q,q}}\Bigg\}{\mathbbm{1}_{n+1 \le q \le 2n-1}}\\
&+(-1)^{n+1 \le q \le 2n-1}\frac{\sum_{j\in[n-1] }\mu_{jn}(t)/n}{f(t)\big\{V_0(t,\theta^*(t))\big\}_{2n,2n}},
\end{align*}
for all $q=1,\ldots,p$,
$\{V_0(t,\theta^*(t))\}_{2n,2n}
=2\sum_{i\in[n]}\{V_0(t,\theta^*(t))\}_{i,i}-\sum_{i\in[n]}\sum_{j\in[2n-1]}\{V_0(t,\theta^*(t))\}_{i,j}.
$
\end{theorem}

\begin{remark}
We present two quick understandings of the result of Theorem \ref{theorem-2}.
First, this theorem allows $n$ to be fixed or divergent.  If $n\to\infty$, then $p\leq n$ always holds after a certain point.
Second, the bias term $\{\sqrt{Nnh}k_{21}h^2\mu_{p}(t)\}$
in Theorem \ref{theorem-2} is uniformly bounded by $O(\sqrt{Nnh}C_{Nn}Q_{Nnh}h^2)$, thus asymptotically diminishing.

\end{remark}

To put the result of Theorem \ref{theorem-2} to practice, we estimate the asymptotic variance.
Note that $k_{02}$ is known and
\begin{equation*}
 \Sigma(t,\theta(t)):= \bar{S}(t,\theta(t)) \bar{V}(t,\theta(t)) \bar{S}(t,\theta(t))^\top
\end{equation*}
is indeed determined by $\bar{V}(t,\theta(t))$
where recall the definition of $\mathcal{S}(\cdot)$ from Section \ref{subsec::notation}.  Setting
\begin{align}
\hat V(t,\hat\theta(t))
= &
\frac{1}{N}\sum_{\ell \in [N]}K_h(t-T_{\ell})V_0(t,\hat\theta(t))
\label{eqn::estimate-V}
\end{align}
and $S(t,\hat\theta(t))=\mathcal{S}(\hat V(t,\hat\theta(t))),$
the estimated asymptotic variance is then $k_{02}\cdot\hat\Sigma(t,\widehat\theta(t))$, where
\begin{equation*}
\hat\Sigma(t,\hat\theta(t))= S(t,\hat\theta(t))\hat V(t,\hat\theta(t)) S(t,\hat\theta(t))^\top.
\notag
\end{equation*}
Now we confirm that the variance estimator is indeed consistent.

\begin{proposition}
\label{proposition-1}
Suppose that Conditions \ref{condition-1} and \ref{condition-2} hold.
If $Nh \to \infty$ as $N\rightarrow \infty$, then we have
\begin{equation*}
 \|\hat\Sigma(t,\hat\theta(t))-\Sigma(t,\theta^*(t))\|_{\max}=o_p(1).
\end{equation*}
\end{proposition}

In practice, we first estimate $\hat V(t, \hat\theta(t))$ using \eqref{eqn::estimate-V}; then plug it into the approximate inversion formula $\mathcal{S}(\cdot)$ in Section \ref{subsec::notation} to obtain $S(t,\hat\theta(t))$ and $\hat\Sigma(t,\hat\theta(t))$.

\section{Simulations}
\label{section::simulations}
In this section, we conduct numerical experiments to gain some empirical understandings in several aspects of our method, including:
the impact of different $h$ choices, estimation accuracy, computation speed and the match between numerical result and our theory's prediction.
The code is available on the online Supplementary Material.
Throughout all experiments, we set $[a,b] = [0.1,0.9]$ and $f(t)\sim$~Uniform$(a,b)$.
To evaluate our method's effectiveness in estimating different parameter evolution dynamics, we set up $\alpha_i^*(t)$'s and $\beta_j^*(t)$'s each of four different types, as is shown in Table \ref{tab::simulation-set-up1}.
\begin{table}[h!]
    \centering
    \caption{Simulation parameter set up.}
    {\small\def\temptablewidth{0.994\textwidth}
{\rule{\temptablewidth}{1pt}}
	\begin{adjustbox}{max width=\textwidth}
		\begin{tabular}{ccccc}
        Index & $1\leq i\leq n/4$ & $n/4+1\leq i\leq 2n/4$ & $2n/4+1\leq i\leq 3n/4$ & $3n/4+1\leq i\leq n-1$ \\\hline
        $\alpha_i^*(t)$ &
            $-0.8(3t-0.6)$ & $-3.6(t-0.3)^2/(t-3)$ & $-0.4^t$ & $-3.2(t-0.5)^2/(1+t^2)$ \\\hline
        $\beta_i^*(t)$ &
            $2(t-0.5)^2/(t-2)$ & $-1.6t(t-0.3)^3+(t-0.2)^2+0.2t$ & $-1.8(t+0.6)\sin(0.2\pi t)$ & $-1.6(t-0.2)^2\sin(\pi t)$ \\\hline
    \end{tabular}
    \label{tab::simulation-set-up1}
    \end{adjustbox}}
\end{table}

Set $\alpha_n^*(t)=\alpha_{n-1}^*(t)$ and $\beta^*_n(t)\equiv 0$ for model identifiability.  For narration simplicity, we stick to the kernel function $K(x)=0.75(1-x^2)\mathbbm{1}_{[|x|<1]}$ in all simulations.
% Each {\color{red} configuration is tested by $100$ repeated experiments.}

\subsection{Simulation 1: Select the tuning parameter $h$}
\label{section::simulation-1}
In this simulation, we consider two settings:
$(n,N)=(40,100)$ and $(160,200)$.
The purpose of this simulation is to compare the optimal $h$'s between these two settings and check if their ratio matches the theoretical suggestion of our Theorem \ref{theorem-1}.
In each setting,
we select the optimal $h$ from the range $[0.05,0.3]$ by a
leave-one-out
cross-validation with loss function
\begin{equation*}
    L_{\rm CV}(h)
    :=
    \sum_{\ell=1}^N\|A(T_\ell) - \hat W_h(T_\ell)\|_F^2,
\end{equation*}
where $A(T_\ell)$ is the observed adjacency matrix at time $t=T_\ell$, and $\hat W_h(T_\ell)$ is the estimated edge probability matrix based on data at times $\{T_{\ell'}: \ell'=1,\ldots,N, \ell'\neq \ell\}$.

\begin{figure}[h!]
\centering
    \includegraphics[width=0.4\textwidth]{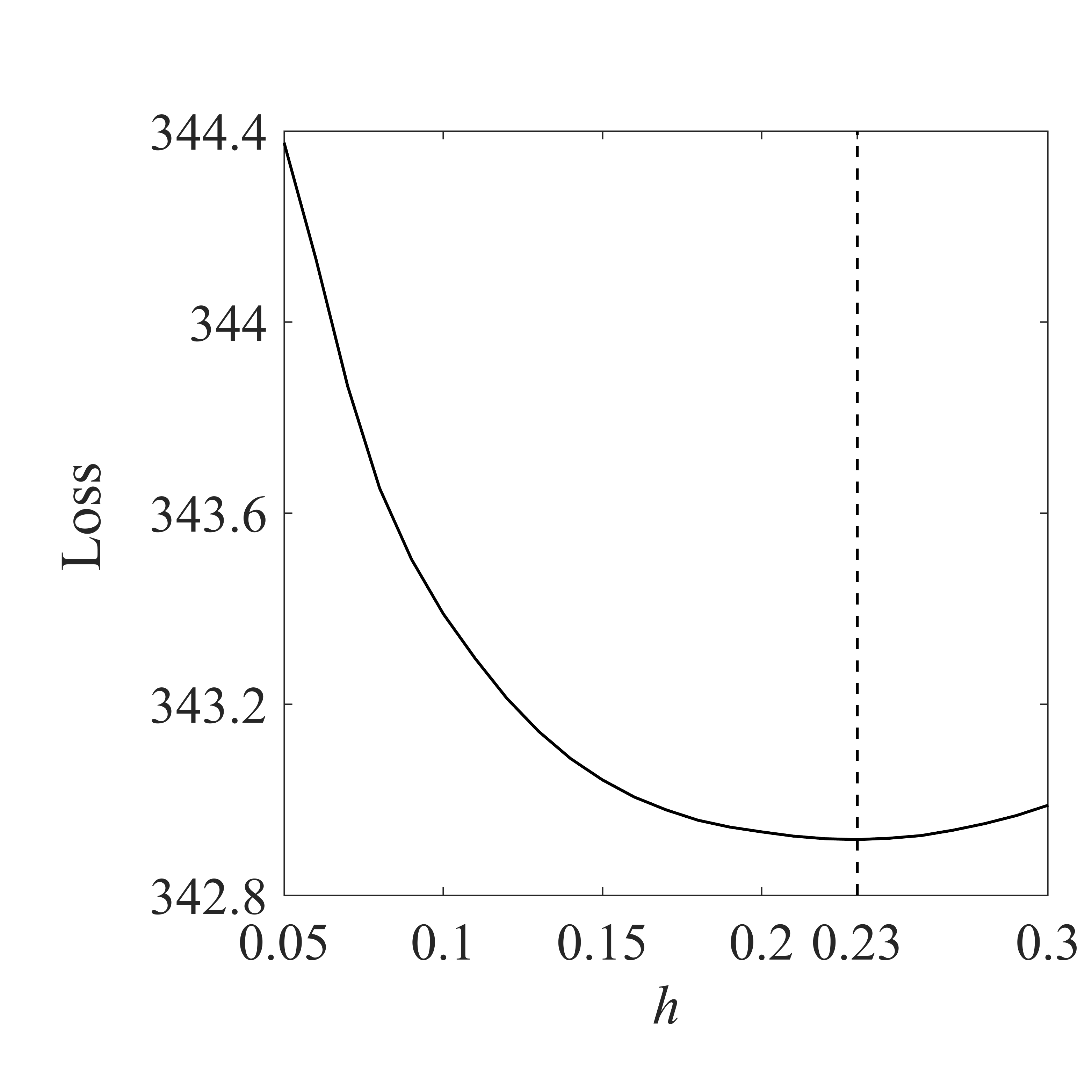}
    \includegraphics[width=0.4\textwidth]{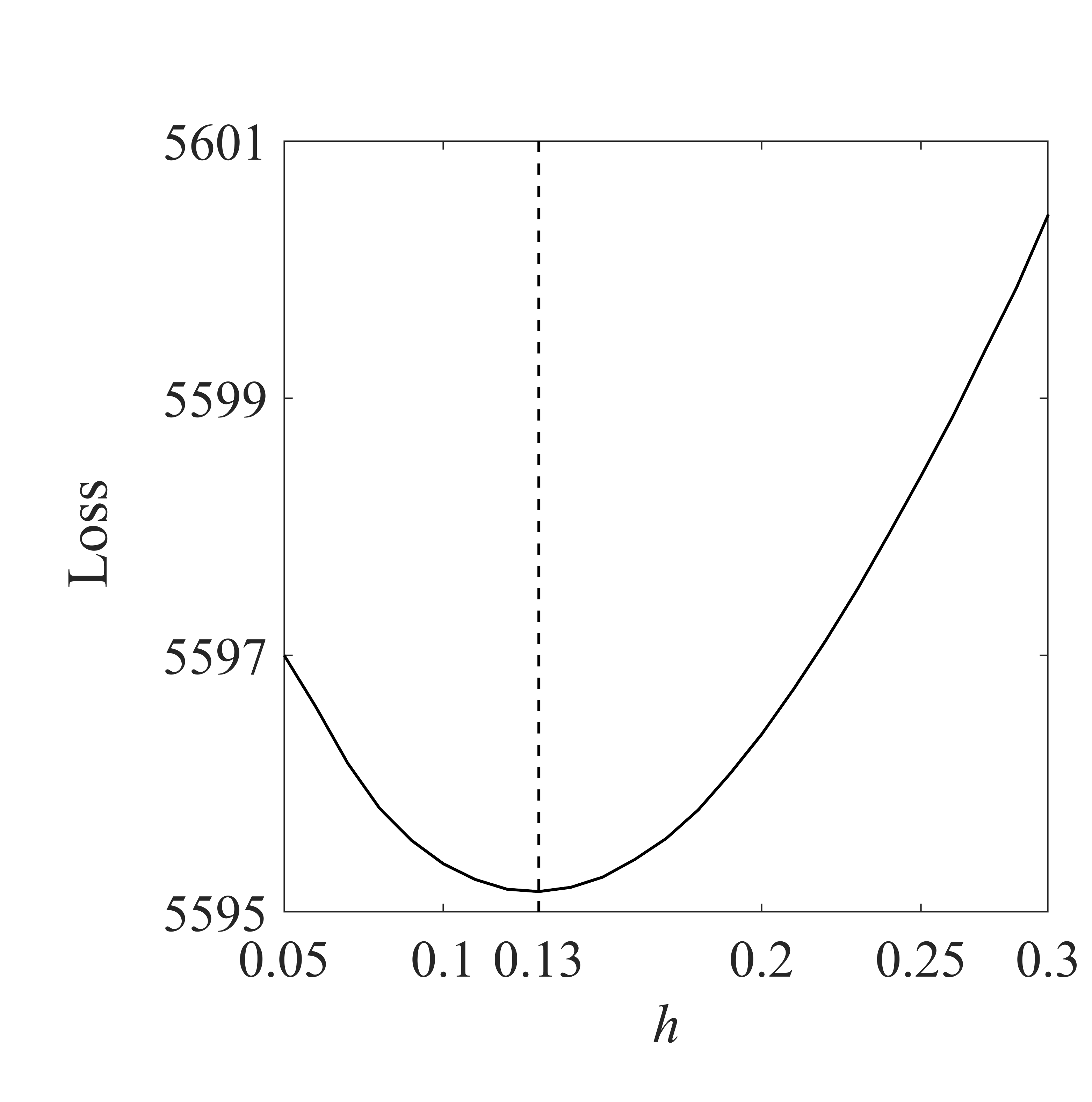}
    \caption{
    Left:
    $n=40,N=100$; Right: $n=160,N=200$.  Both plots show only one experiment.}
    \label{fig::simulation-1}
\end{figure}
We repeated experiments 50 times for both settings.  The average optimal $h_{\rm opt}^{(n,N)}$ selected by cross-validation in the two settings are $\bar h_{\rm opt}^{(40,100)}=0.225$ and $\bar h_{\rm opt}^{(160,200)}=0.140$.  Their ratio is $0.622$, which well-matches our Theorem \ref{theorem-1}'s prediction of $\{(40\times 100)/(160\times 200)\}^{(1/5)}=0.659$.

\subsection{Simulation 2: RMSE diminishing rate and time cost}
\label{section::simulation-2}

In this simulation, we focus on evaluating estimation errors as the network size $n$ and number of time points $N$ vary.  Following the result of Simulation 1 and our theory, we set $h=0.23\big\{ (40\times 100)/(Nn) \big\}^{1/5}$ from this simulation onward.  We vary $n\in\{40,80,160,320\}$ and $N\in\{25,50,100,200\}$.  The estimation error is measured by the following RMSE:
\begin{equation*}
{\rm RMSE}=\sqrt{\frac{1}{N}\sum_{\ell=1}^N(\hat\theta(T_\ell)-\theta^*(T_\ell))^2}
\end{equation*}
We repeat all experiments 100 times and record the average value of RMSE.
The estimation errors and running times are reported in Figure
\ref{fig::simulation-2}.
\begin{figure}[ht!]
    \centering
    \includegraphics[width=0.32\textwidth]{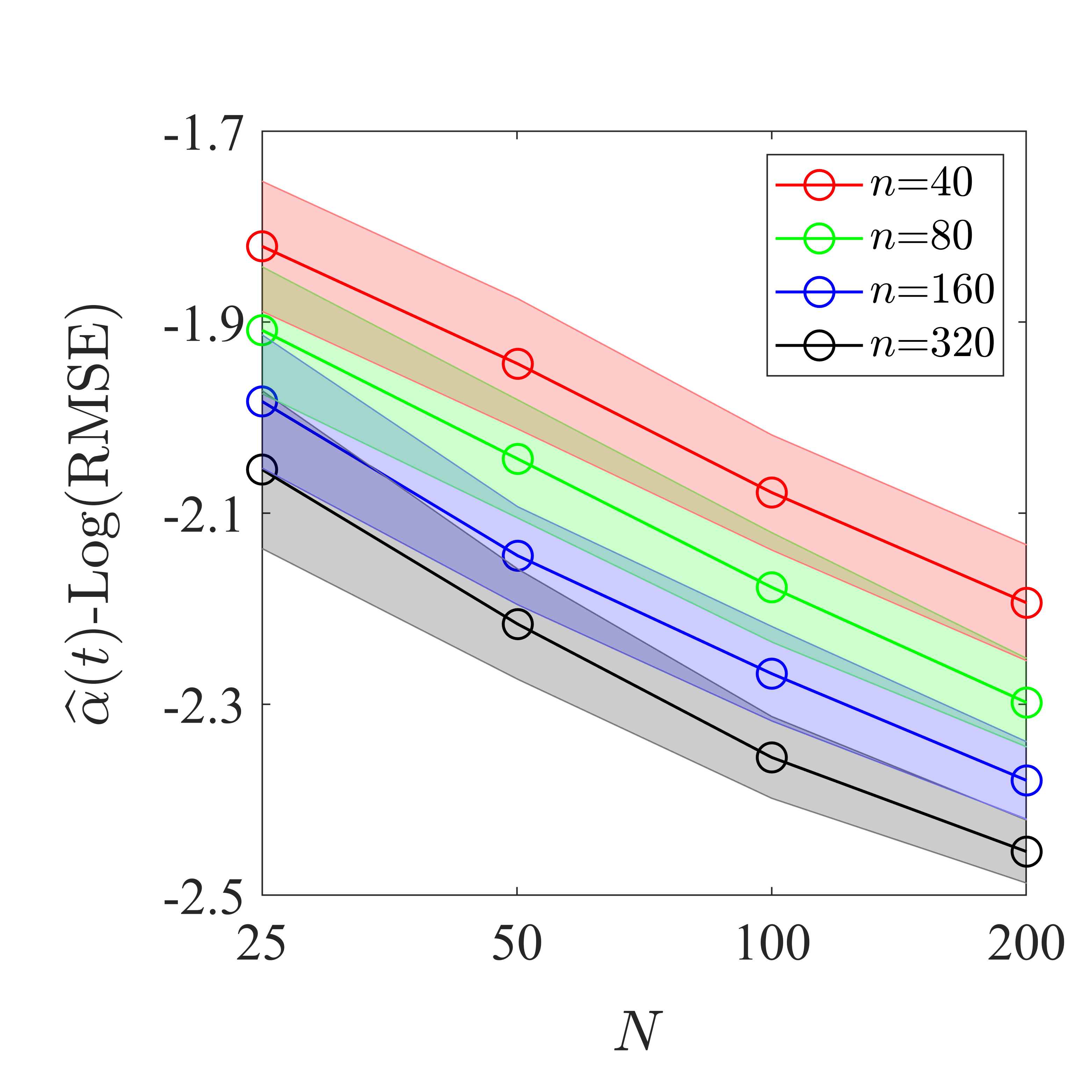}
    \includegraphics[width=0.32\textwidth]{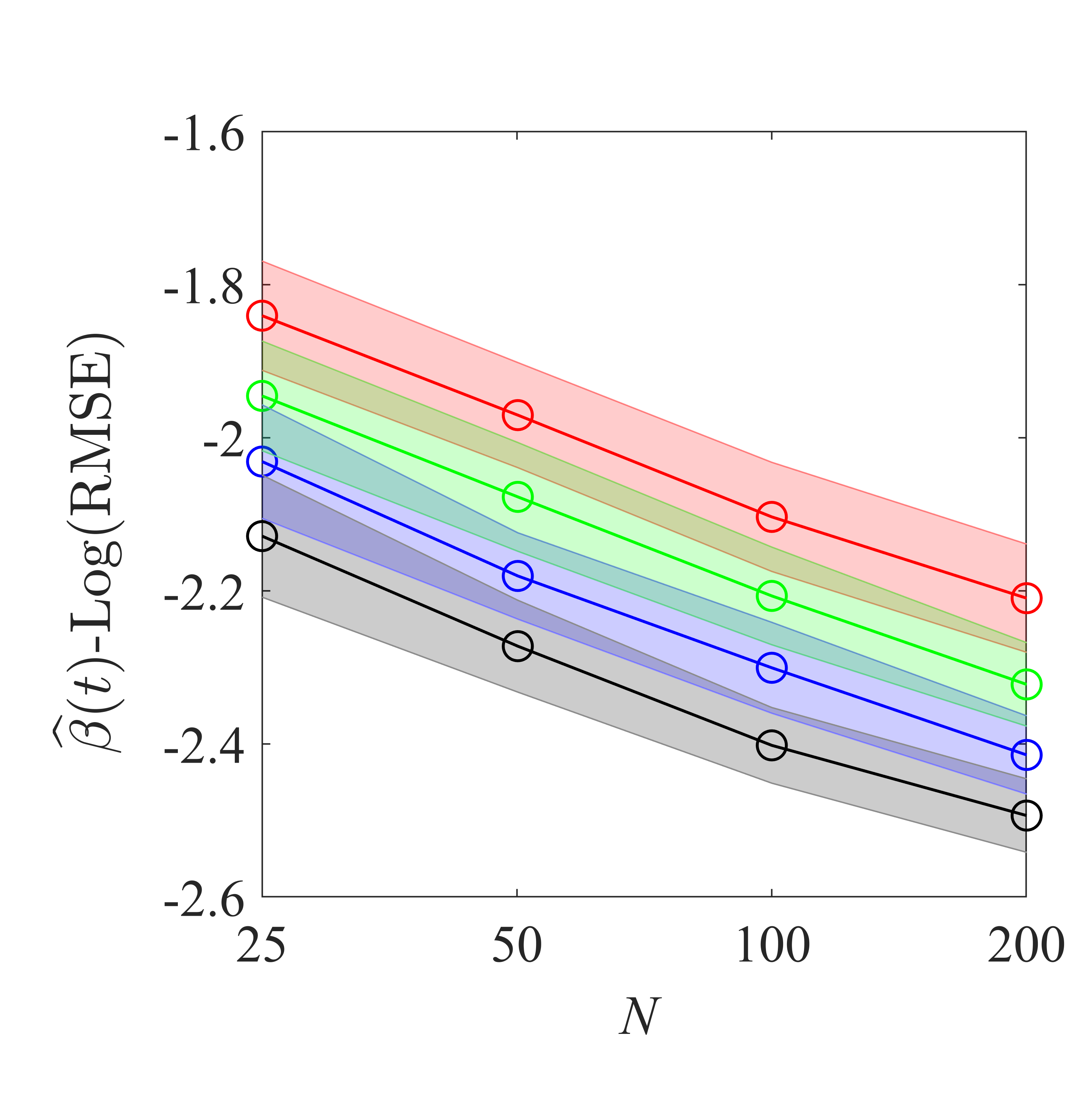}
    \includegraphics[width=0.32\textwidth]{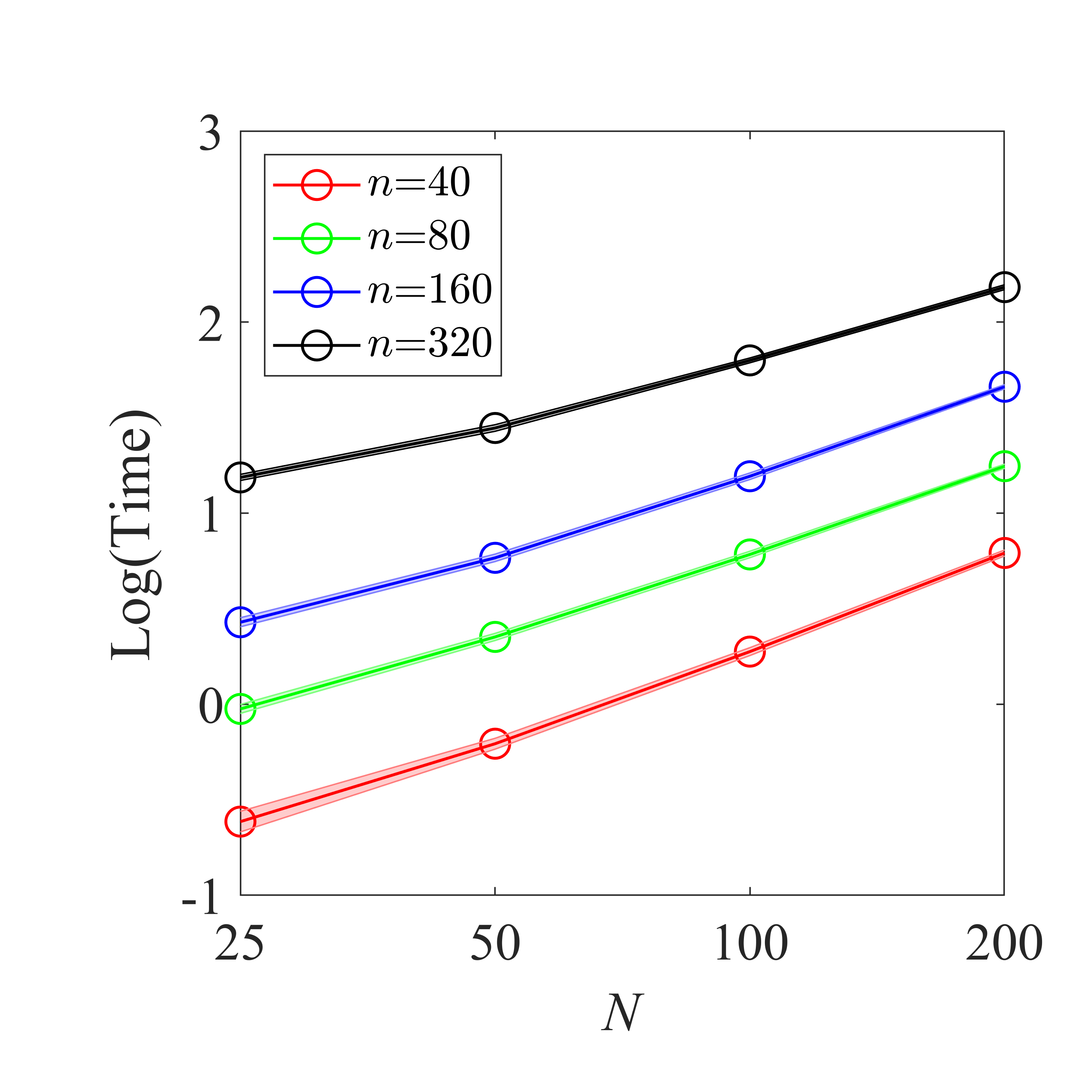}
    \caption{Left and Middle: log estimation errors for $\alpha(t)$ and $\beta(t)$ parameters, respectively.  Right: log time cost.  Shades indicate Monte Carlo standard deviations across repeated experiments.}
    \label{fig::simulation-2}
\end{figure}
From the plots, we see that $\log({\rm RMSE})$ diminishes at the rate roughly $n^{-1/5}$.  Our theory predicts an error rate of $O(Q_{Nnh}^3 \cdot n^{-2/5})$ that agrees with the observed result of this experiment.  Specifically, the time cost seems linear in $N$, which also matches our theoretical prediction; also, the vertical increment in log time cost as $n$ doubles is reasonably close to $\log2 = 0.69$, matching our theoretical understanding that the computation complexity scales quadratically in $n$.

\subsection{Simulation 3: Fitted curves with comparison to point-wise method}
\label{section::simulation-3}
To assess the performance of our method, we plot the true curve $\theta(t)$ and compare it to $\hat\theta(t).$
We compare our method with two alternative approaches:
(1) the point-wise method, which estimates $\theta(t)$ at each $t\in \{T_1,\ldots,T_N\}$;
and (2) the smoothed point-wise method as described in Section \ref{subsec::parameter-estimation},
and then smooth the resulting estimates using a kernel smoothing procedure.
Here, we set $n=160, N=200$ and $h=0.13$ as suggested by the result of Simulation 1.
We repeat the experiment 100 times.
See Table \ref{tab::simulation-set-up1} for the settings of $\alpha^*_i(t)$ and $\beta_{j}^*(t)$.

Figures \ref{fig::simulation-31} and \ref{fig::simulation-32} illustrate the result; and the numerical outputs are summarized in Tables \ref{tab::simulation3:table-3} and \ref{tab::simulation3:table-4}.
Our method shows smaller variance and significantly smaller bias compared to the smoothed point-wise method and a much bigger advantage over the point-wise method.
As aforementioned, our method can borrow data from neighboring time points to improve estimation;
whereas the point-wise method's curve fluctuates across different time points due to a limited $n$ -- this aligns with our theoretical understanding by comparing our Theorem \ref{theorem-1} with classical error rate results at one time point, c.f. \citet{yan2016asymptotics}.

\begin{figure}[ht!]
    \centering
    \includegraphics[width=1\textwidth]{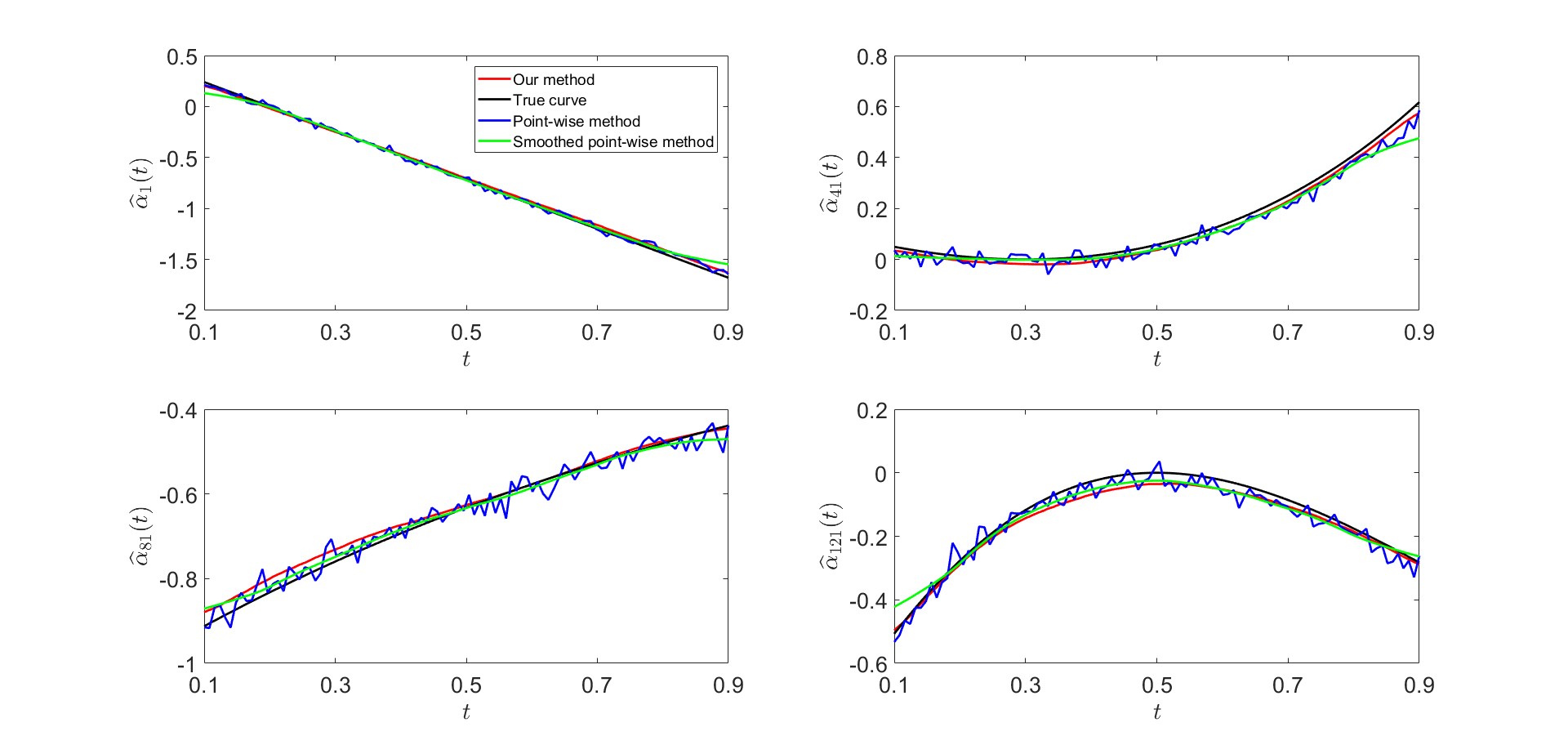}
\includegraphics[width=1\textwidth]{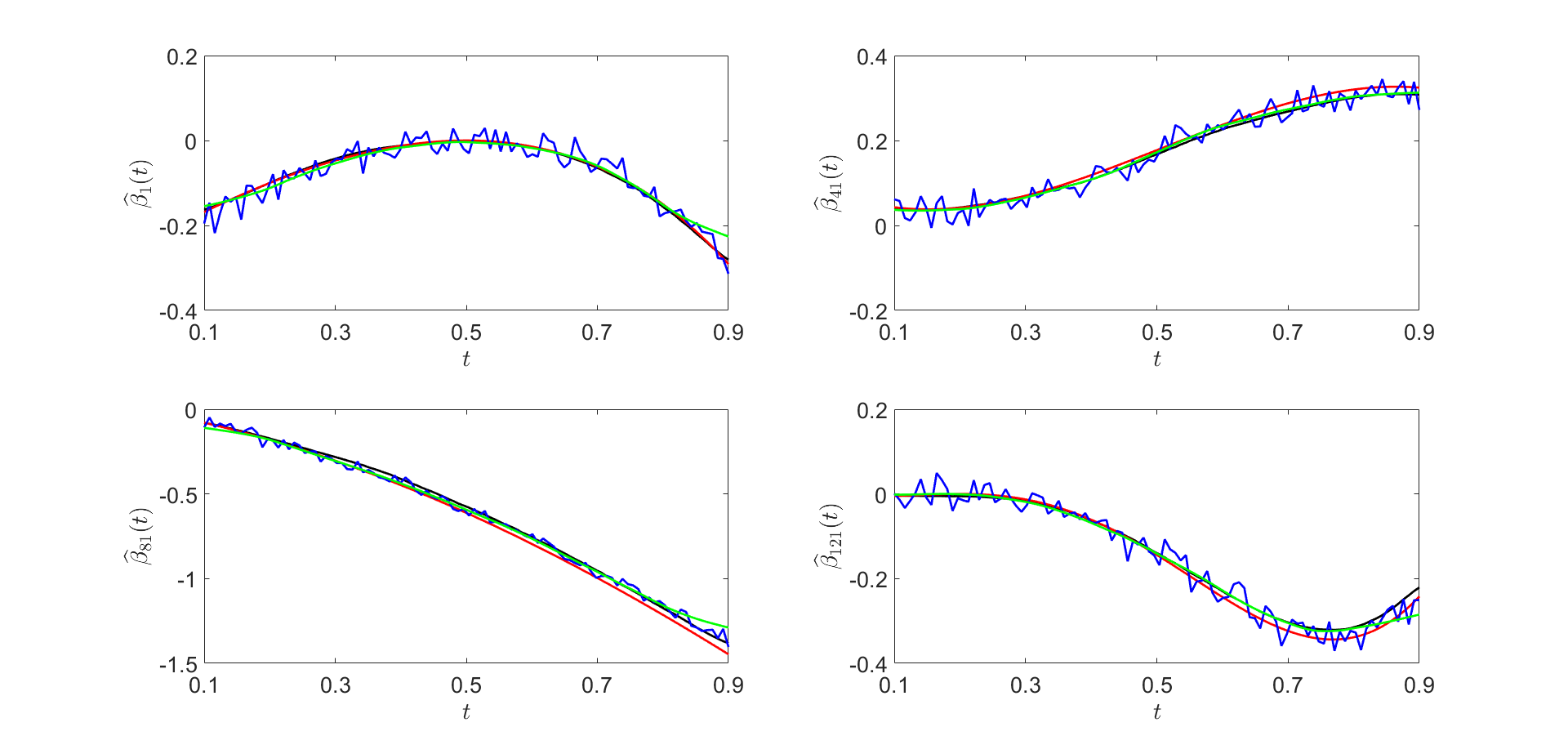}
    \caption{The estimate results of $\alpha_i^*(t)$ and $\beta_{i}^*(t)~(i=1,~ 41,~ 81,~ 121)$ with $n=160,N=200$ and $h=0.13.$}
    \label{fig::simulation-31}
\end{figure}

% \begin{figure}[h!]
%     \centering
%     \includegraphics[width=1.0\textwidth]{R1-figure/R11-figure3_2new.png}
%     \caption{The estimate results of $\beta_1^*(t), \beta_{41}^*(t), \beta_{81}^*(t)$ and $\beta_{121}^*(t)$ with
% $n=160, N=200$ and $h=0.13.$  Black lines: true curves; red lines: our method; blue lines: point-wise method; green lines: smoothing of point-wise method.}
%     \label{fig::simulation-32}
% \end{figure}

\begin{figure}[h!]
    \centering
    \includegraphics[width=1\textwidth]{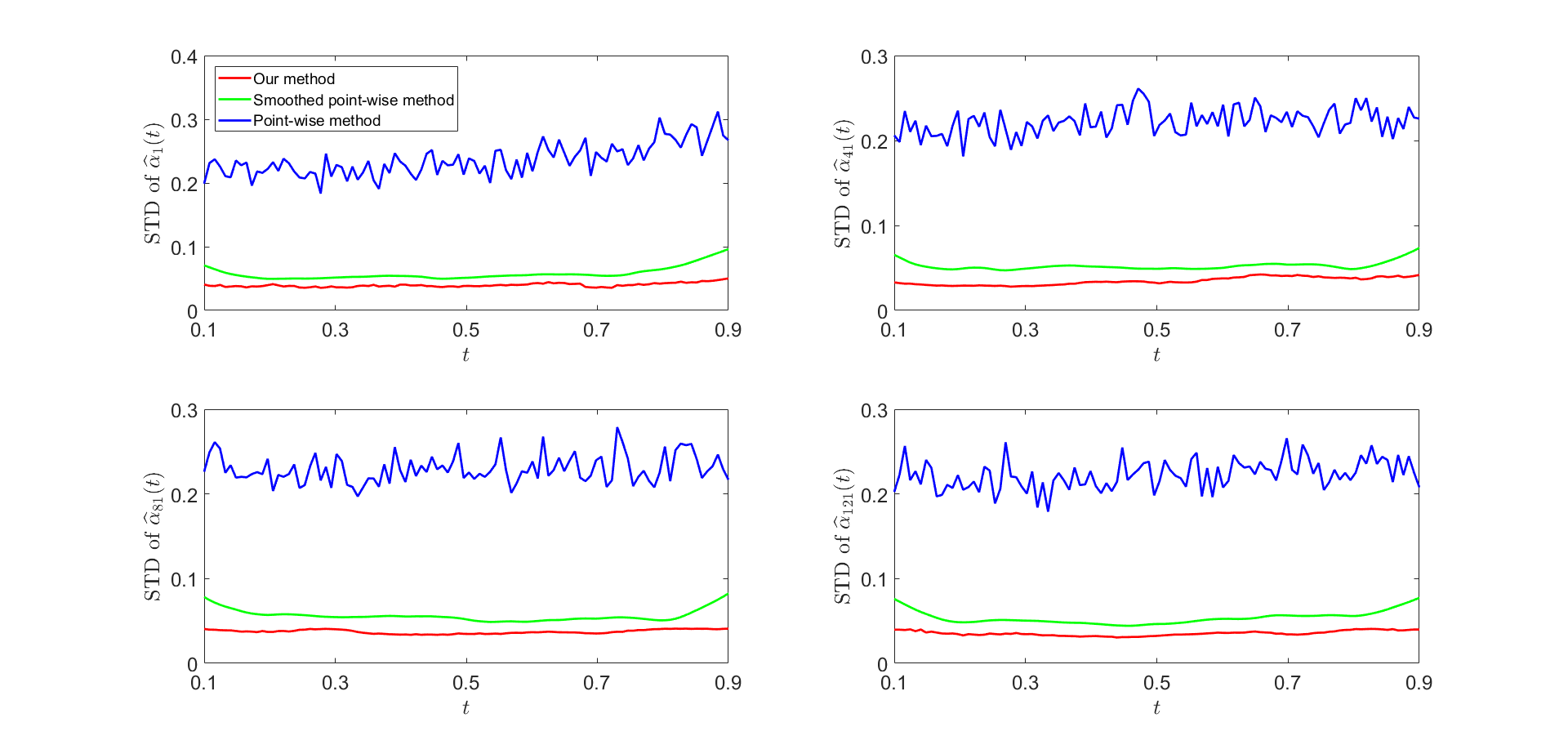}
    \includegraphics[width=1\textwidth]{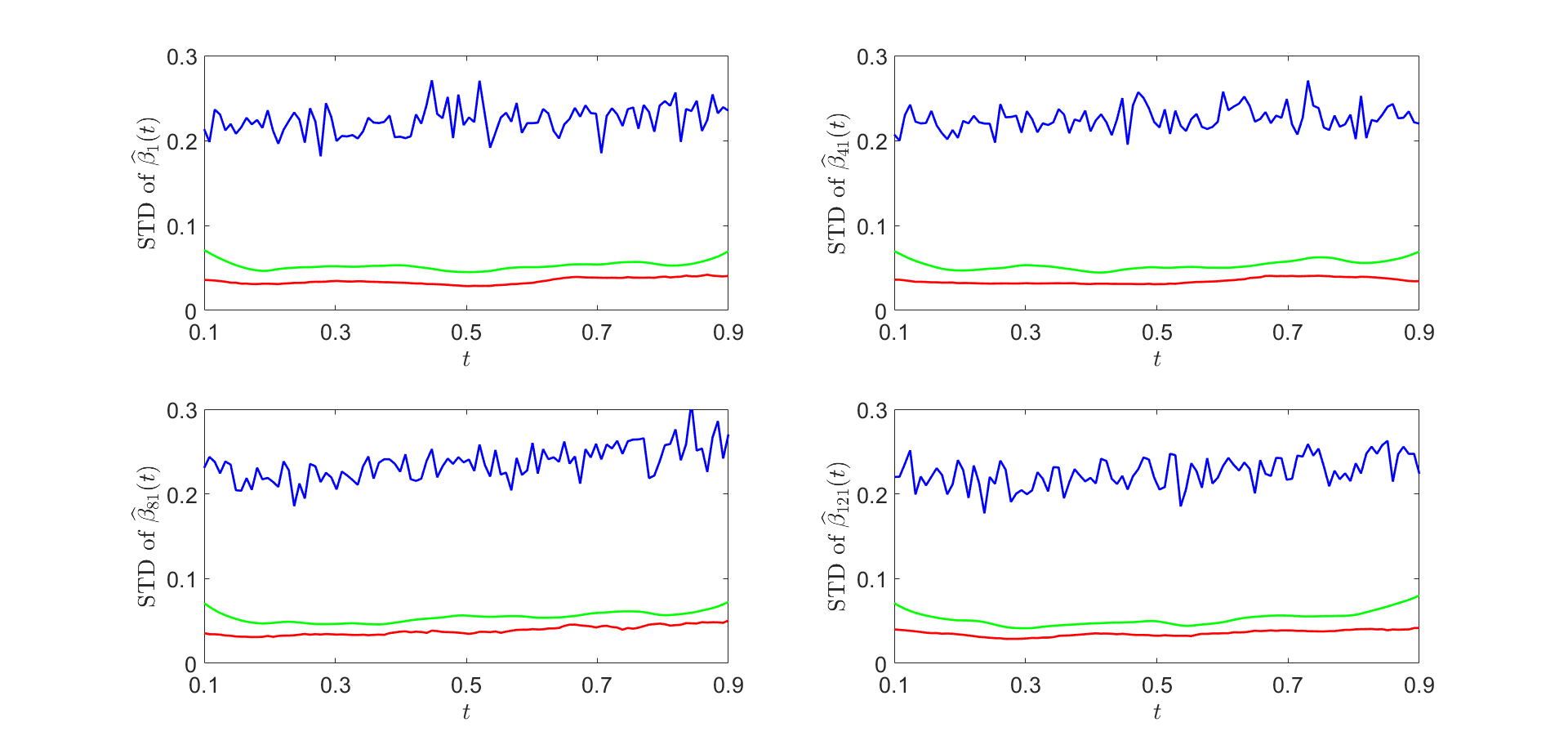}
    \caption{The sample standard deviations (SD)
    of estimators $\widehat\alpha_i(t)$ and $\widehat\beta_{i}(t)~(i=1,~ 41,~ 81,~ 121)$ with $n=160, N=200$ and $h=0.13.$}
    \label{fig::simulation-32}
\end{figure}

Next, we experiment in a sparse network setting.
The settings for parameters are given in Table \ref{tab::simulation-set-up2}. Here we set $n=40, N=100$ and $h=0.23$ as suggested by the result of Simulation 1.
The results are shown in Figure
\ref{fig::simulation-51}.
We see that the estimators of the point-wise method
do not exist at some time points, and this also impacts the smoothed point-wise method because it smooths the result of the point-wise estimation method.
In sharp contrast, our method works stably by incorporating data from neighboring time points at the earliest stage in the estimation.

\begin{figure}[h!]
    \centering
    \includegraphics[width=1\textwidth]{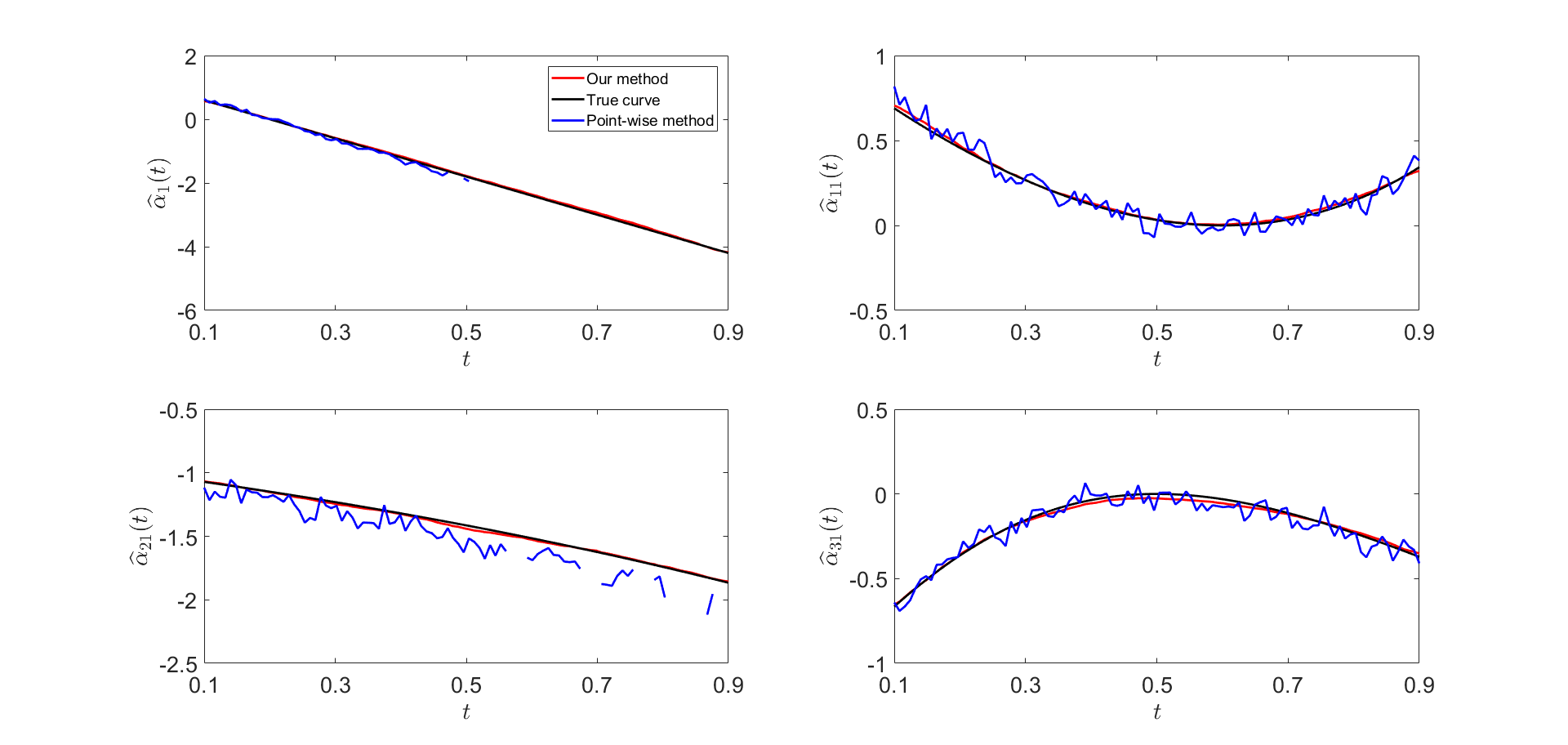}
     \includegraphics[width=1\textwidth]{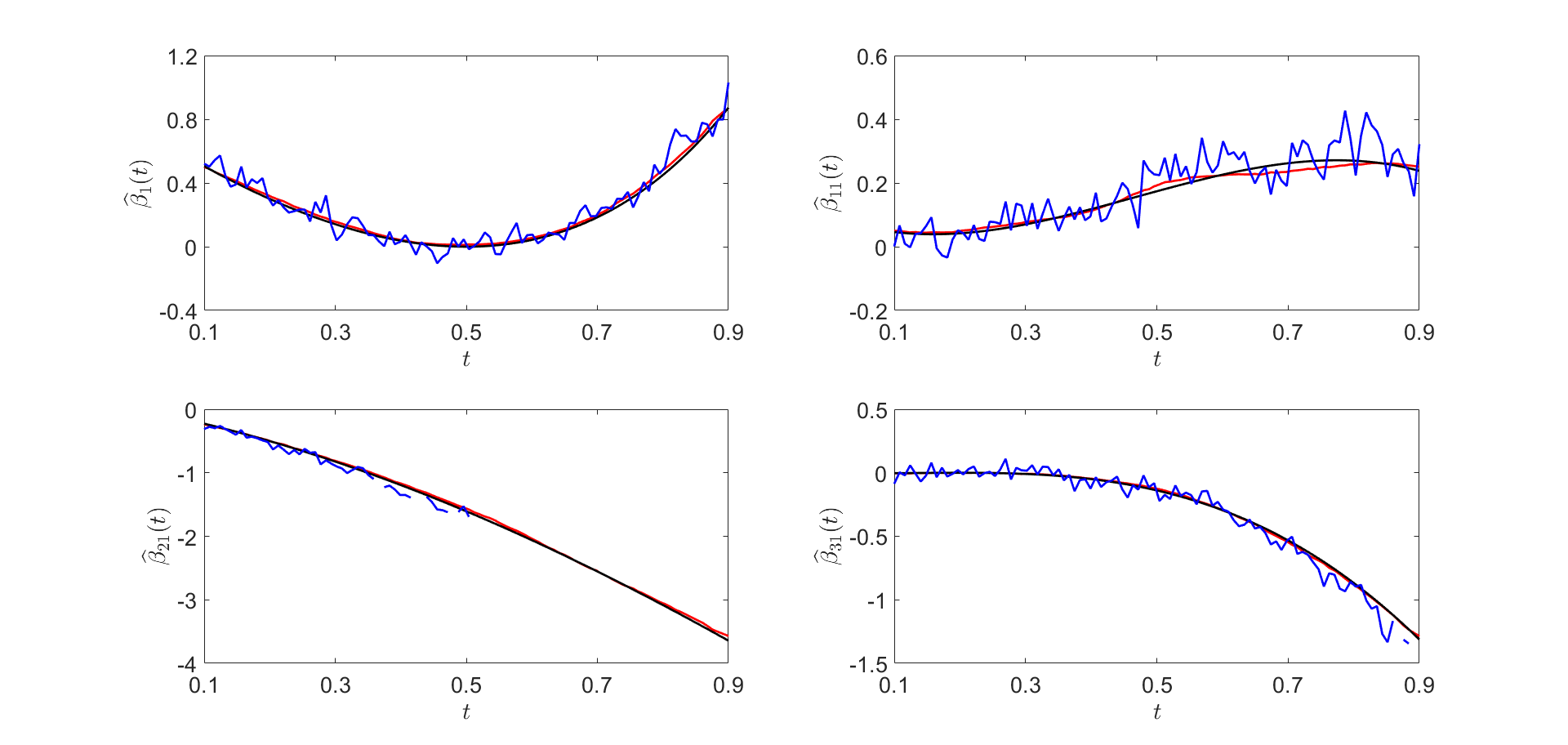}
    \caption{The estimate results of $\alpha_i^*(t)$ and $\beta_i^*(t)~(i=1,~11,~21,~31)$ with $n=40, N=100$ and $h=0.23.$}
    \label{fig::simulation-51}
\end{figure}

% \begin{figure}[h!]
%     \centering
%     \includegraphics[width=1\textwidth]{R1-figure/R11-figure3_8new.png}
%     \caption{The estimate results of $\beta_1^*(t), \beta_{11}^*(t), \beta_{21}^*(t)$ and $\beta_{31}^*(t)$ with
% $n=40, N=100$ and $h=0.23.$ Black lines: true curves; red lines: our method;  blue lines: point-wise method.}
%     \label{fig::simulation-52}
% \end{figure}

% Figure \ref{fig::simulation-3} illustrates the true and estimated curves from one experiment.

\begin{table}[ht!]
\caption{Simulation results of
				the Bias and SD for $\widehat\alpha_i(t)$ when
				$n=160, N=200, h=0.13.$}
{\small\def\temptablewidth{0.995\textwidth}
{\rule{\temptablewidth}{1pt}}
	\begin{adjustbox}{max width=\textwidth}
		\setlength{\tabcolsep}{6.3mm}
		\begin{tabular}{cccccccc}
					&&\multicolumn{2}{c}{Our method} & \multicolumn{2}{c}{Point-wise method} & \multicolumn{2}{c}{Smoothed point-wise method}  \\
					\cline{3-4}   \cline{5-6} \cline{7-8}
t    & $\alpha^*(t)$ & Bias  & SD                     & Bias  & SD                     & Bias  & SD                                   \\
			\hline
			0.2  & $i=1$      & -0.004 & 0.041                 & 0.019  & 0.233                 & -0.019 & 0.050                               \\
			& $i=41$     & -0.008 & 0.029                 & -0.026 & 0.181                 & -0.017 & 0.049                               \\
			& $i=81$     & 0.013  & 0.037                 & 0.016  & 0.204                 & 0.033  & 0.057                               \\
			& $i=121$    & -0.008 & 0.033                 & -0.022 & 0.205                 & -0.011 & 0.048                               \\
			0.4  & $i=1$      & -0.002 & 0.040                 & -0.037 & 0.227                 & 0.012  & 0.054                               \\
			& $i=41$     & -0.011 & 0.034                 & -0.047 & 0.216                 & -0.023 & 0.052                               \\
			& $i=81$     & 0.008  & 0.034                 & -0.016 & 0.214                 & 0.018  & 0.055                               \\
			& $i=121$    & -0.026 & 0.032                 & -0.056 & 0.209                 & -0.041 & 0.047                               \\
			0.6  & $i=1$      & 0.004  & 0.042                 & -0.025 & 0.218                 & 0.024  & 0.056                               \\
			& $i=41$     & -0.019 & 0.038                 & -0.026 & 0.242                 & -0.019 & 0.050                               \\
			& $i=81$     & -0.008 & 0.037                 & -0.018 & 0.238                 & 0.000  & 0.049                               \\
			& $i=121$    & -0.030 & 0.036                 & -0.045 & 0.208                 & -0.029 & 0.053   \\
			\hline
			\end{tabular}
   \label{tab::simulation3:table-3}
    \end{adjustbox}}
\end{table}

\begin{table}[h!]
\caption{Simulation results of
				the Bias and SD for $\widehat\beta_i(t)$ when
				$n=160, N=200, h=0.13.$}
{\small\def\temptablewidth{0.999\textwidth}
{\rule{\temptablewidth}{1pt}}
	\begin{adjustbox}{max width=\textwidth}
		\setlength{\tabcolsep}{6.3mm}
		\begin{tabular}{cccccccc}
					&&\multicolumn{2}{c}{Our method} & \multicolumn{2}{c}{Point-wise method} & \multicolumn{2}{c}{Smoothed point-wise method}  \\
					\cline{3-4}   \cline{5-6} \cline{7-8}
					t    & $\beta^*(t)$ & Bias  & SD                     & Bias  & SD                     & Bias  & SD \\
					\hline
						0.2  & $i=1$      & 0.001  & 0.031                 & 0.018  & 0.212                 & -0.012 & 0.048                               \\
					& $i=41$     & 0.001  & 0.032                 & -0.007 & 0.223                 & -0.003 & 0.047                               \\
					& $i=81$     & 0.008  & 0.031                 & 0.001  & 0.214                 & 0.000  & 0.048                               \\
					& $i=121$    & -0.005 & 0.034                 & -0.015 & 0.228                 & -0.001 & 0.051                               \\
					0.4  & $i=1$      & 0.000  & 0.032                 & 0.031  & 0.203                 & -0.004 & 0.053                               \\
					& $i=41$     & -0.011 & 0.031                 & 0.018  & 0.224                 & -0.011 & 0.045                               \\
					& $i=81$     & 0.037  & 0.038                 & 0.056  & 0.247                 & 0.010  & 0.050                               \\
					& $i=121$    & -0.001 & 0.035                 & -0.003 & 0.212                 & -0.007 & 0.047                               \\
					0.6  & $i=1$      & -0.003 & 0.032                 & -0.018 & 0.220                 & -0.003 & 0.051                               \\
					& $i=41$     & -0.010 & 0.035                 & -0.011 & 0.258                 & -0.002 & 0.050                               \\
					& $i=81$     & 0.040  & 0.040                 & 0.059  & 0.260                 & 0.031  & 0.054                               \\
					& $i=121$    & 0.014  & 0.035                 & -0.001 & 0.233                 & 0.016  & 0.049     \\
					\hline
			\end{tabular}	
    \end{adjustbox}}
    \label{tab::simulation3:table-4}
\end{table}

\begin{table}[h!]
\caption{The second set up of true parameters in Simulation 3.}
{\small\def\temptablewidth{0.994\textwidth}
{\rule{\temptablewidth}{1pt}}
	\begin{adjustbox}{max width=\textwidth}
		\begin{tabular}{ccccc}
        Index & $1\leq i\leq n/4$ & $n/4+1\leq i\leq 2n/4$ & $2n/4+1\leq i\leq 3n/4$ & $3n/4+1\leq i\leq n-1$ \\\hline
        $\alpha_i^*(t)$ &
            $-2(3t-0.6)$ & $-8(t-0.6)^2/(t-3)$ & $-2^t$ & $-4.2(t-0.5)^2/(1+t^2)$ \\\hline
        $\beta_i^*(t)$ &
            $-6(t-0.5)^2/(t-2)$ & $-2(t-0.3)^3+(t-0.2)^2+0.2t$ & $-4(t+0.8)\sin(0.2\pi t)$ & $-5(t-0.2)^2\sin(0.2\pi t)$ \\\hline
    \end{tabular}
    \label{tab::simulation-set-up2}
    \end{adjustbox}}
\end{table}

\subsection{Simulation 4: Asymptotic normality}
\label{section::normality-test}
This simulation focuses on validating our asymptotic normality theory, namely, Theorem \ref{theorem-2}.  For simplicity, we consider one particular setting $(n,N)=(160,200)$. According to the result of Simulation 1, we set $h=0.13$.  We simulate 1000 Monte Carlo repetitions.
\begin{figure}[h!]
    \centering
    \includegraphics[width=0.45\textwidth]{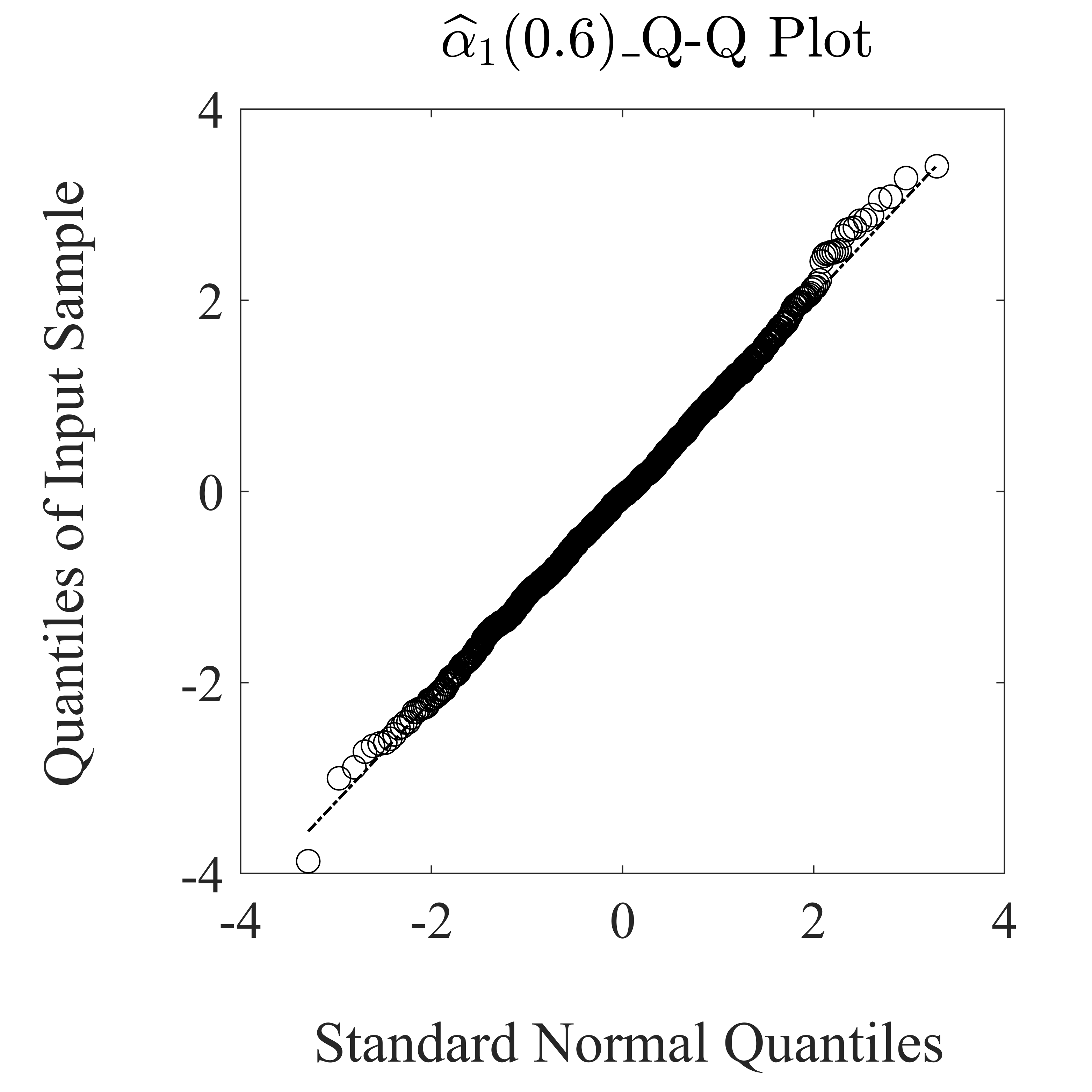}
    \includegraphics[width=0.45\textwidth]{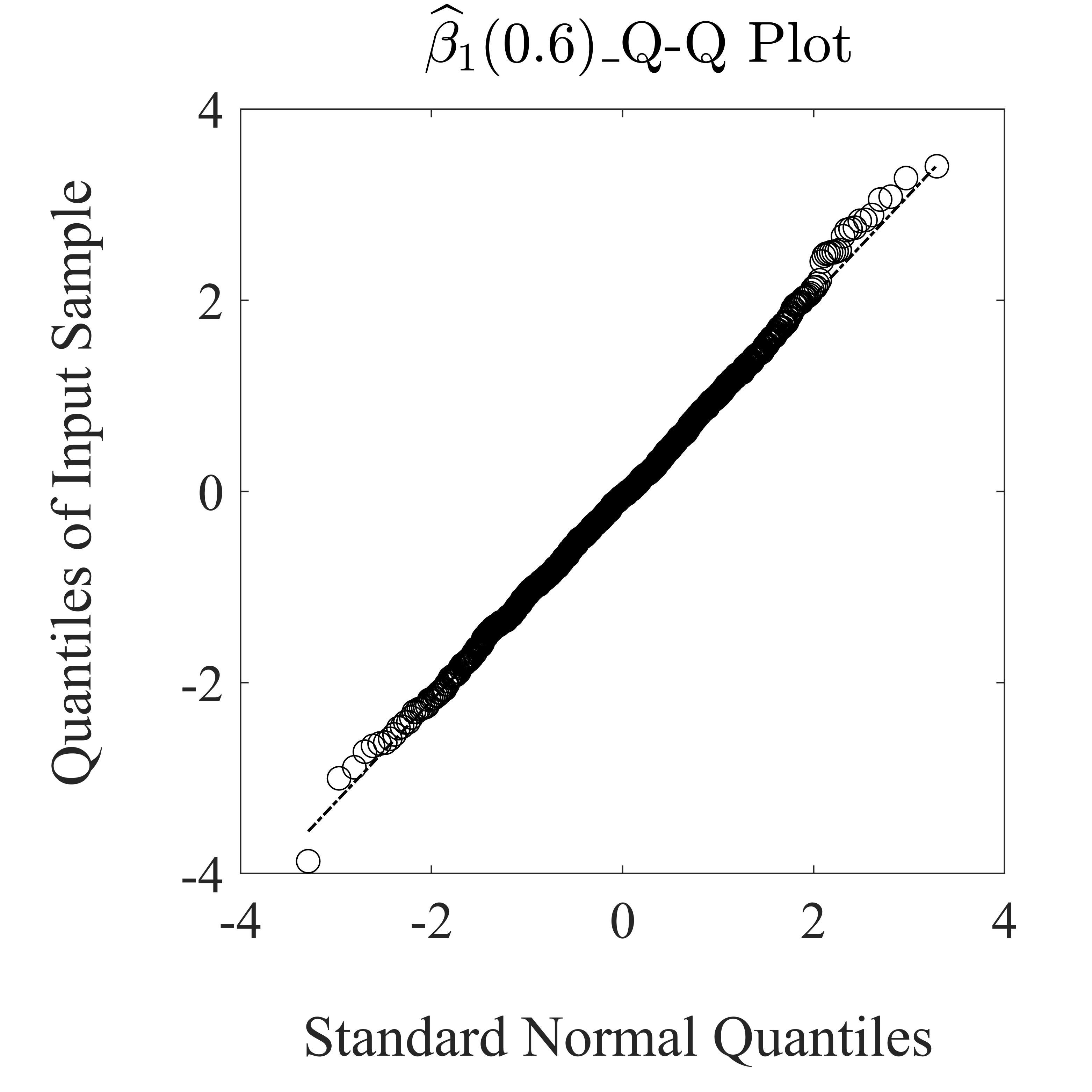}
    \caption{One-dimensional marginal normality.}
    \label{fig::simulation-4-QQ}
\end{figure}
\begin{figure}[h!]
    \centering
    \includegraphics[width=0.45\textwidth]{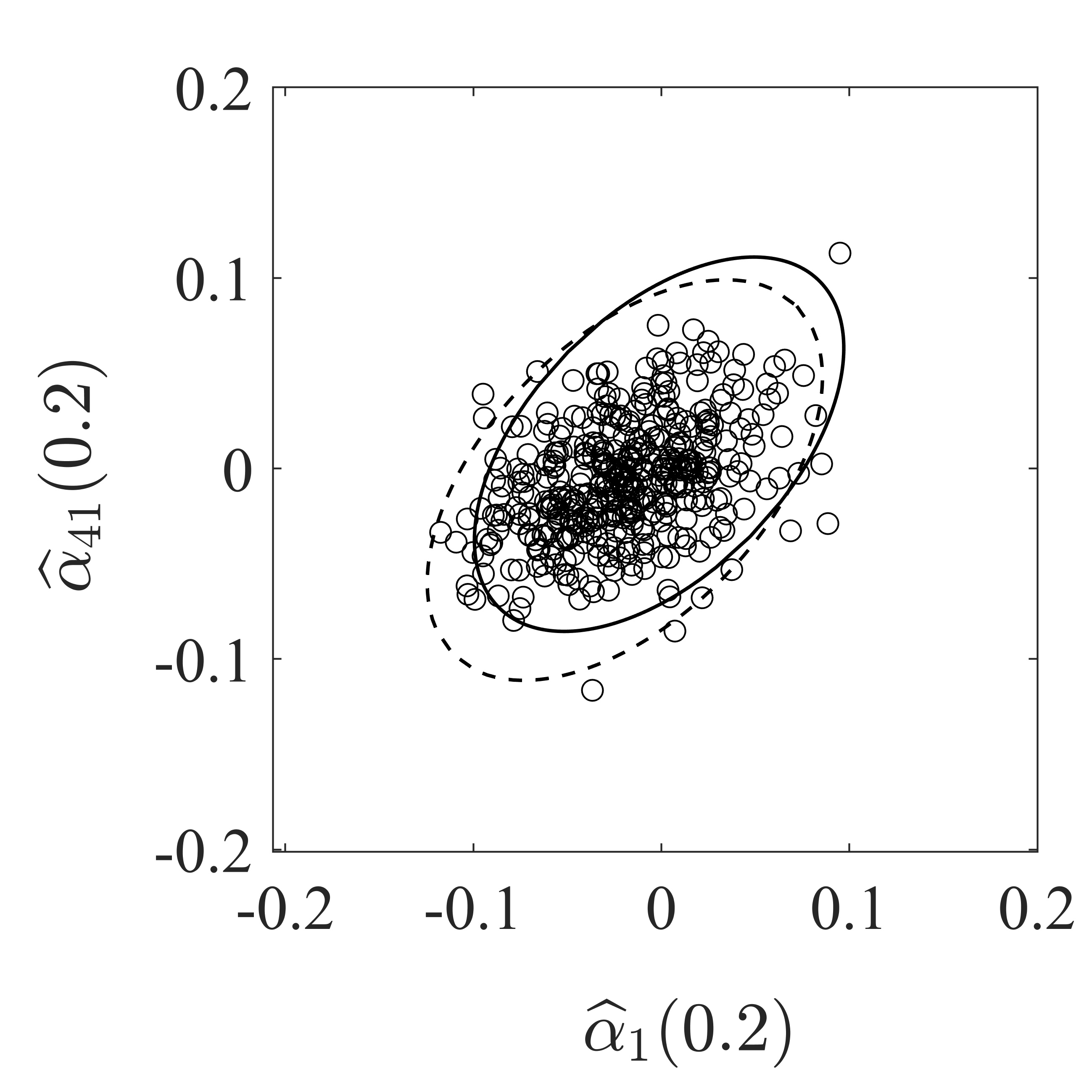}
    \includegraphics[width=0.45\textwidth]{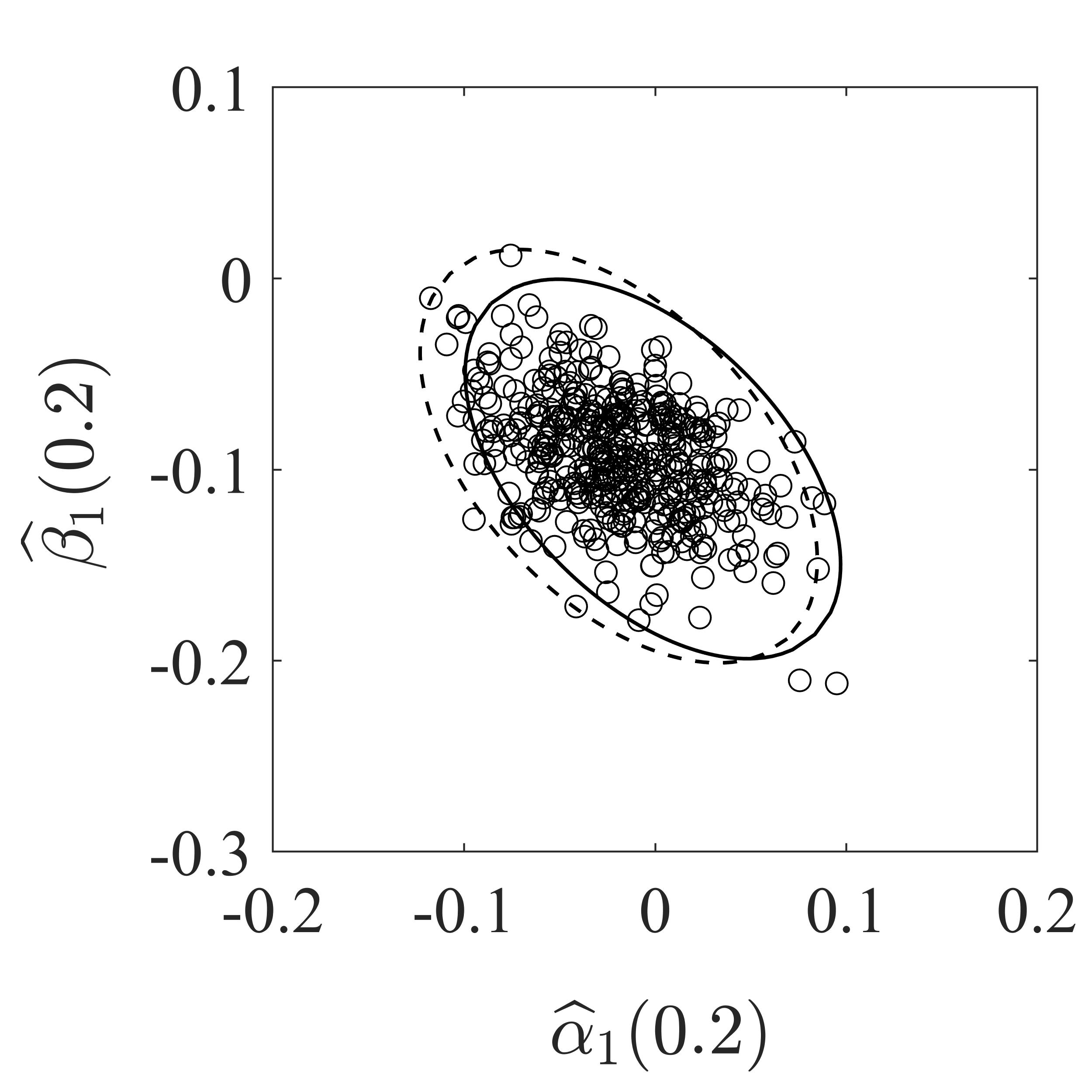}
    \caption{Two-dimensional marginal normality.  Solid ellipse is the contour of $(x-\mu)^\top \Sigma^{-1} (x-\mu)=9$, where $(\mu,\Sigma)$ are the asymptotic center and covariance matrix predicted by Theorem \ref{theorem-2}, respectively.  Dashed ellipse is a similar contour but with $(\mu,\Sigma)$ being the sample mean and covariance matrix, respectively.}
    \label{fig::simulation-4-scatter}
\end{figure}
Figure \ref{fig::simulation-4-QQ} and \ref{fig::simulation-4-scatter}
illustrate the results. The plots suggest that the simulation result very well matches the prediction of our Theorem \ref{theorem-2}.

\section{Data example}
\label{section::data-example}
We analyze an email data set \citet{leskovec2007graph, paranjape2017motifs}, available at \url{http://snap.stanford.edu/data/email-Eu-core-temporal.html}.  It records email exchanges between members of a large European institution.  Each entry $(i,j,t)$ signifies that $i$ emailed $j$ at time $t$, where $t$ is accurate to ``day".  We transcribe the individual entries into network snapshots by aggregating all communications by month -- using a too small time unit will result too sparse network snapshots.  This produces $N=19$ snapshots over $(n_1,n_2,n_3,n_4)=(298, 154, 89, 140)$ nodes in four selected departments in the institute, respectively.  For convenience, we rescale the entire time span to $[0,1]$.
\begin{sloppypar}
We use the cross-validation method in Section \ref{section::simulation-1} to tune the bandwidth $h$ and select
$(h_1,h_2,h_3,h_4)= (0.21,0.15,0.46,0.18)$
for the four departments, respectively.  Using the selected $h$ values, we estimate $\{\hat\alpha_i(t)\}_{i=1}^n$ and $\{\hat\beta_j(t)\}_{j=1}^{n-1}$ (recall $\hat\beta_n(t):=\beta_n(t)\equiv0$).
\end{sloppypar}
However, in real data analysis, we need to be very cautious with the interpretation of estimated $\alpha_i(t)$'s and $\beta_j(t)$'s.  Let us consider a toy example, where $n=2$; over time, node 1 becomes more active as both sender and receiver, whereas node 2 grows more inactive.  One might think this leads to increasing $\alpha_1(t)$ and $\beta_1(t)$ and decreasing $\alpha_2(t)$ and $\beta_2(t)$.  However, due to identifiability concerns, we follow the tradition in directed $\beta$-model analysis and set $\beta_2(t)\equiv 0$.  The increasingly inactive behavior of node 2 as a receiver will ``bounce'' to all other parameters $\alpha_1(t),\alpha_2(t)$ and $\beta_1(t)$.  Consequently, we cannot even be sure that $\alpha_1(t)$ and $\alpha_2(t)$ will be increasing.

Therefore, we should not merely look at the trends in each individual $\alpha_i(t)$ and $\beta_j(t)$ marginally.  Fortunately, the above-mentioned identifiability issue will not affect $\alpha_i(t)-\alpha_j(t)$ or $\beta_i(t)-\beta_j(t)$ for any $i\neq j$.  In view of this, after obtaining the estimation $\hat\alpha_i(t)$'s, we compute a pairwise distance matrix ${\cal D}_\alpha$ as follows
\begin{equation}
\label{equation-12}
    \Big({\cal D}_\alpha\Big)_{i,j}^2
    :=
    \int_a^b
    \big|
        \hat\alpha_i(t)
        -
        \hat\alpha_j(t)
    \big|^2 dt
\end{equation}
Then we perform a K-means clustering using ${\cal D}_\alpha$ and obtain an $\alpha$-cluster label for each node.
We clarify that the clustering here refers to the classification of nodes based on how their $\alpha_i(t)$ curve changes over time and has a very different meaning than the concept of ``varying community memberships'' in dynamic stochastic block models.
Also, performing a multi-dimensional scaling (MDS) on ${\cal D}_\alpha$, projecting into $\mathbb{R}$, produces a one-dimensional latent space position for each node.  Do the same using $\hat\beta_j(t)$'s, we obtain a $\beta$-cluster label and $\beta$-MDS position coordinate on each node.
A natural question is how to select $K$, the number of clusters, for each department's two clustering structures ($\alpha$- and $\beta$-clustering).
For this, we compute the ratio between the average between-cluster distance over the average within-cluster distance, for different choices of $K$.  This suggests that we select $K=4$ for $\alpha$- and $\beta$-clustering for department 1; and $K=5$ for clustering for departments 2--4
when $N=19.$ And we select $K=4$ for $\alpha$- and $\beta$-clustering for department 1 and department 4; and $K=5$ for clustering for departments 2--3 when $N=38.$
And select $K=5$ for $\alpha$- and $\beta$-clustering for department 2; and $K=4$ for clustering for other departments when $N=57.$
All details of this step are reported in Figure \ref{fig::data-example-tune-K}.
\begin{figure}[h!]
    \centering
    \includegraphics[width=0.24\textwidth]{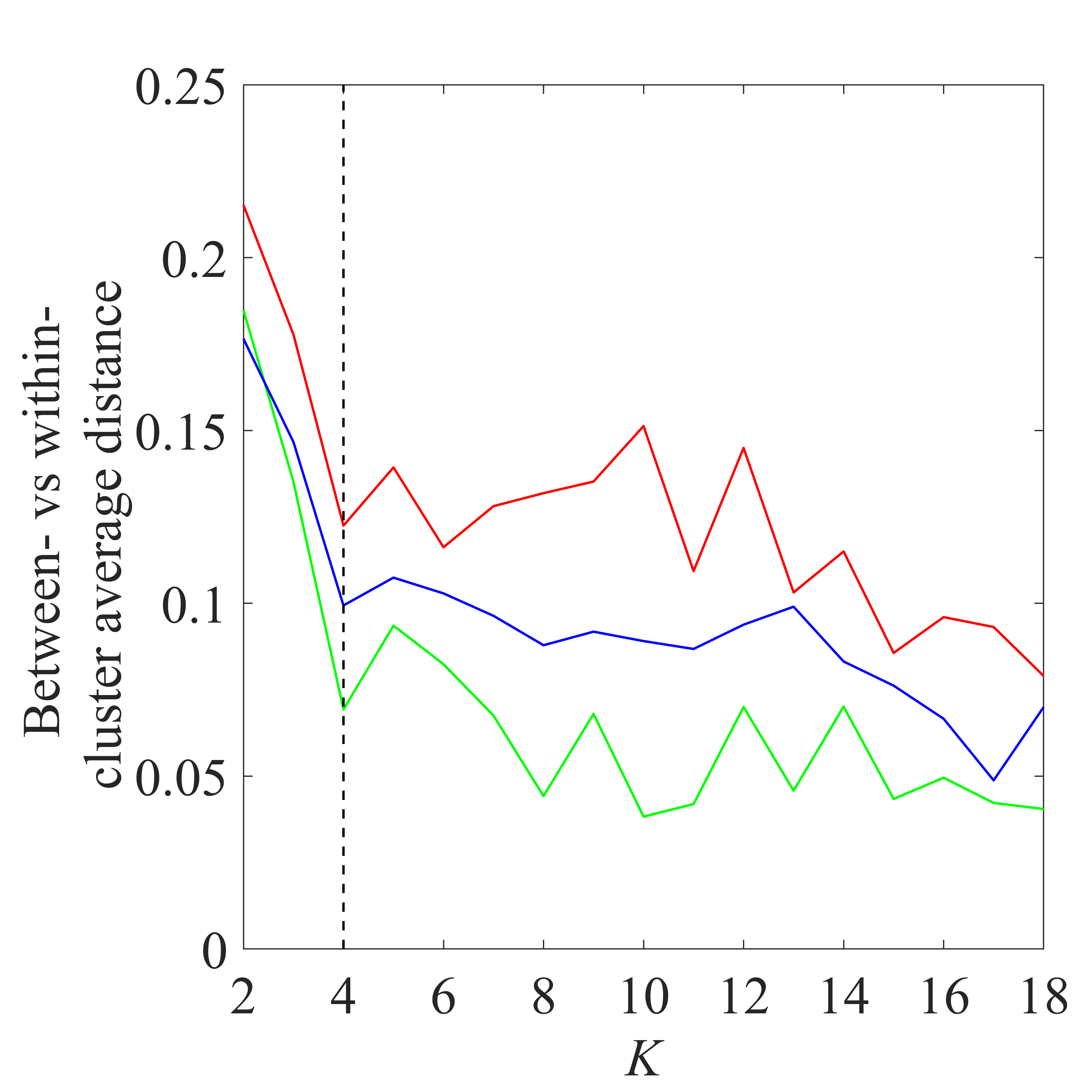}
    \includegraphics[width=0.24\textwidth]{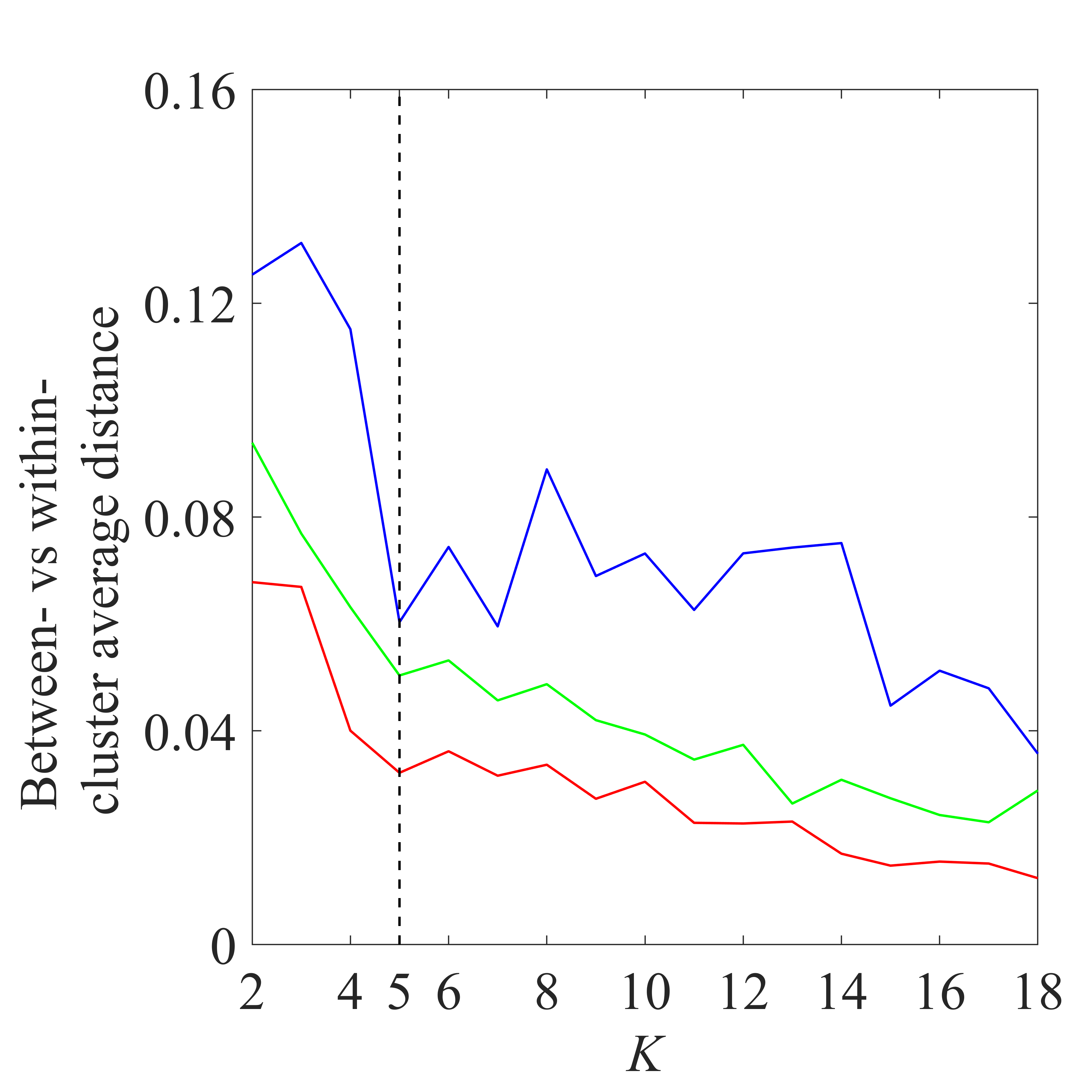}
    \includegraphics[width=0.24\textwidth]{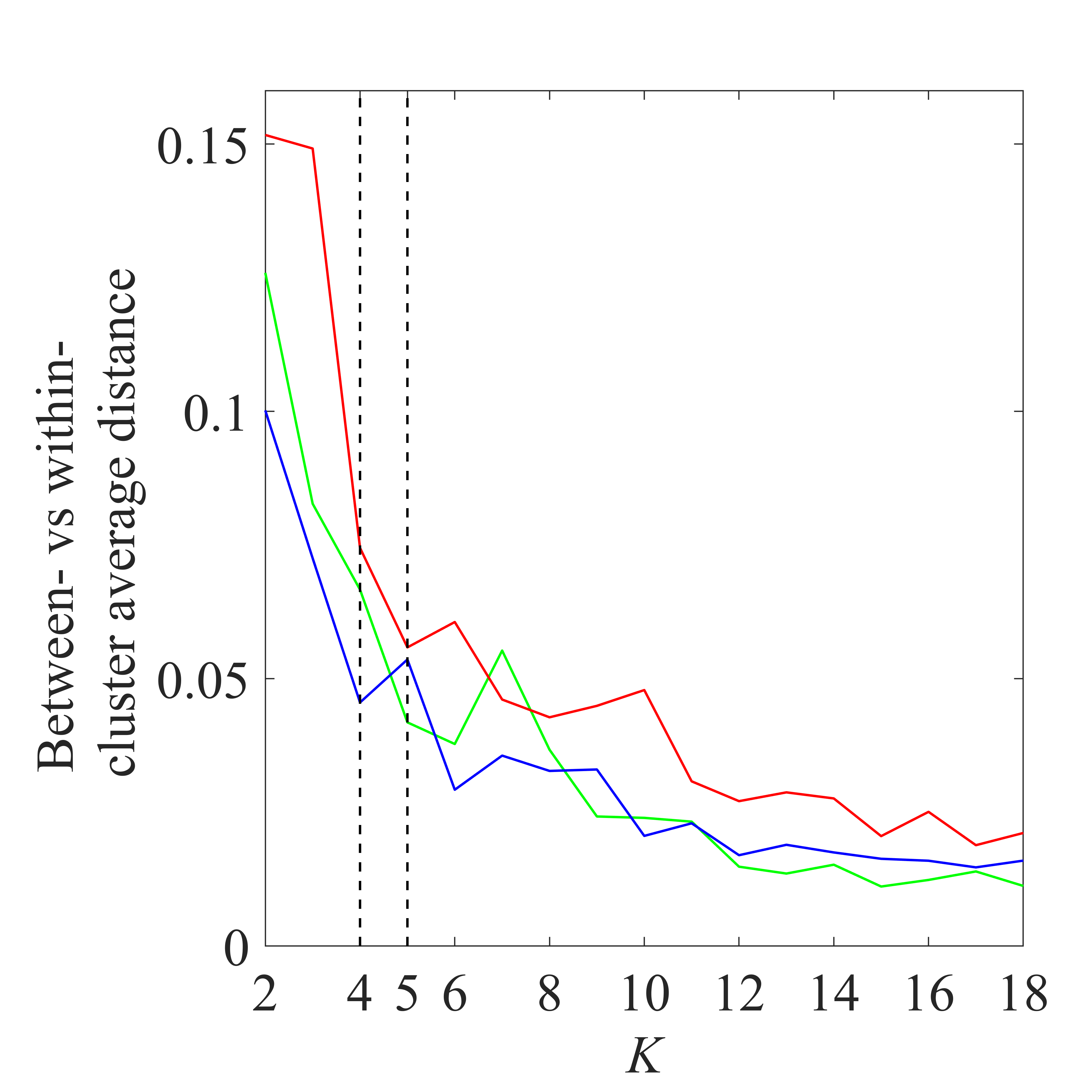}
    \includegraphics[width=0.24\textwidth]{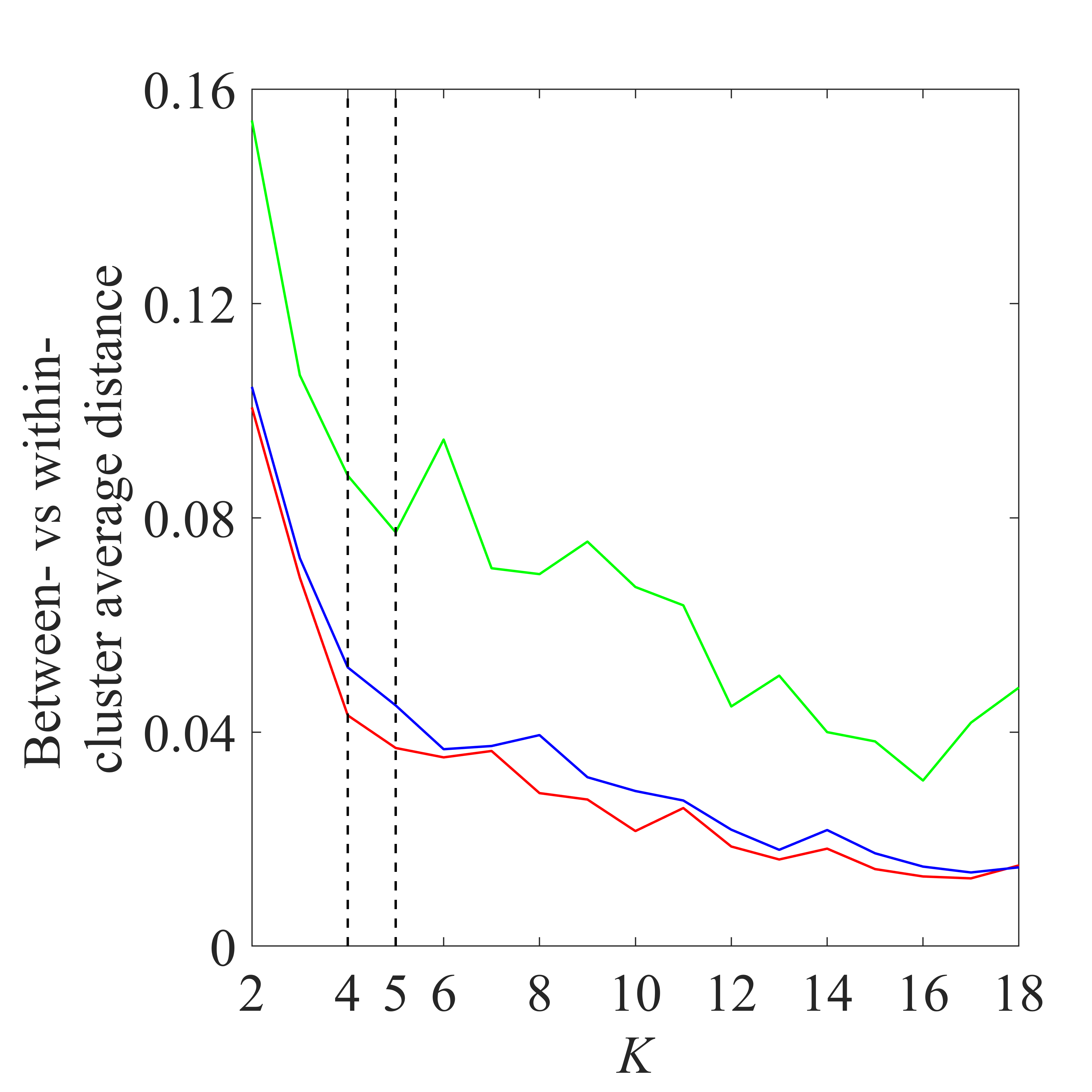}\\
    \includegraphics[width=0.24\textwidth]{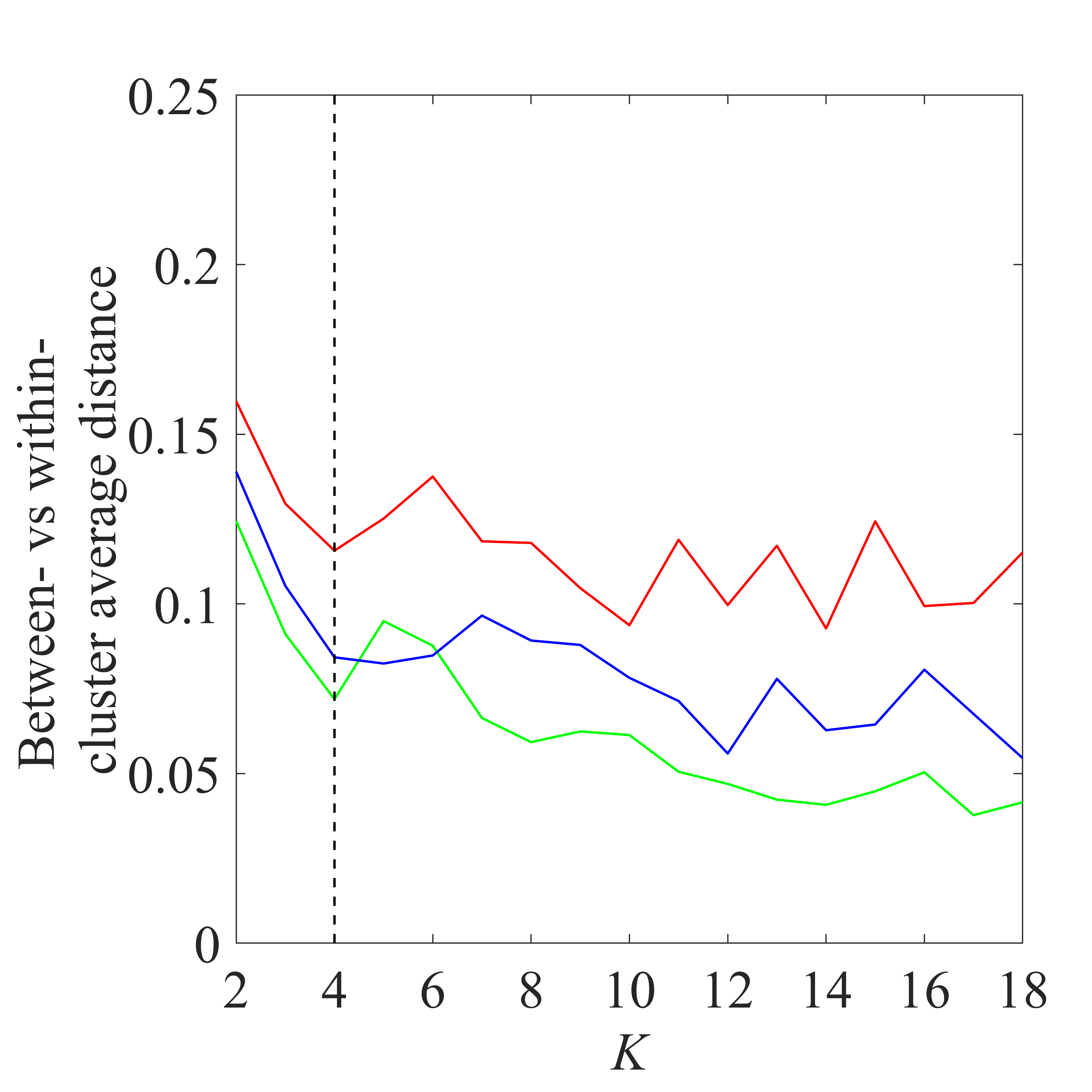}
    \includegraphics[width=0.24\textwidth]{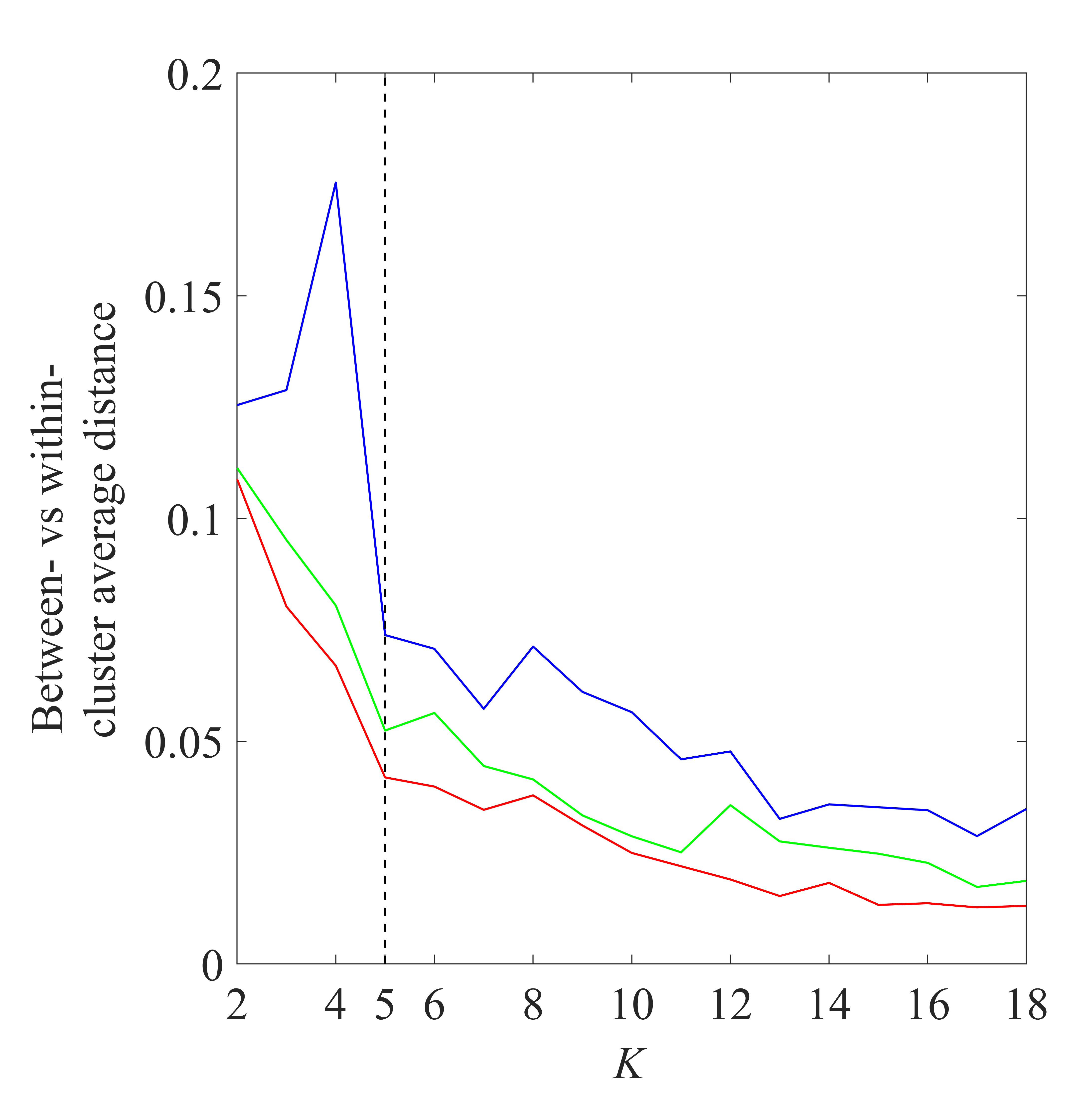}
    \includegraphics[width=0.24\textwidth]{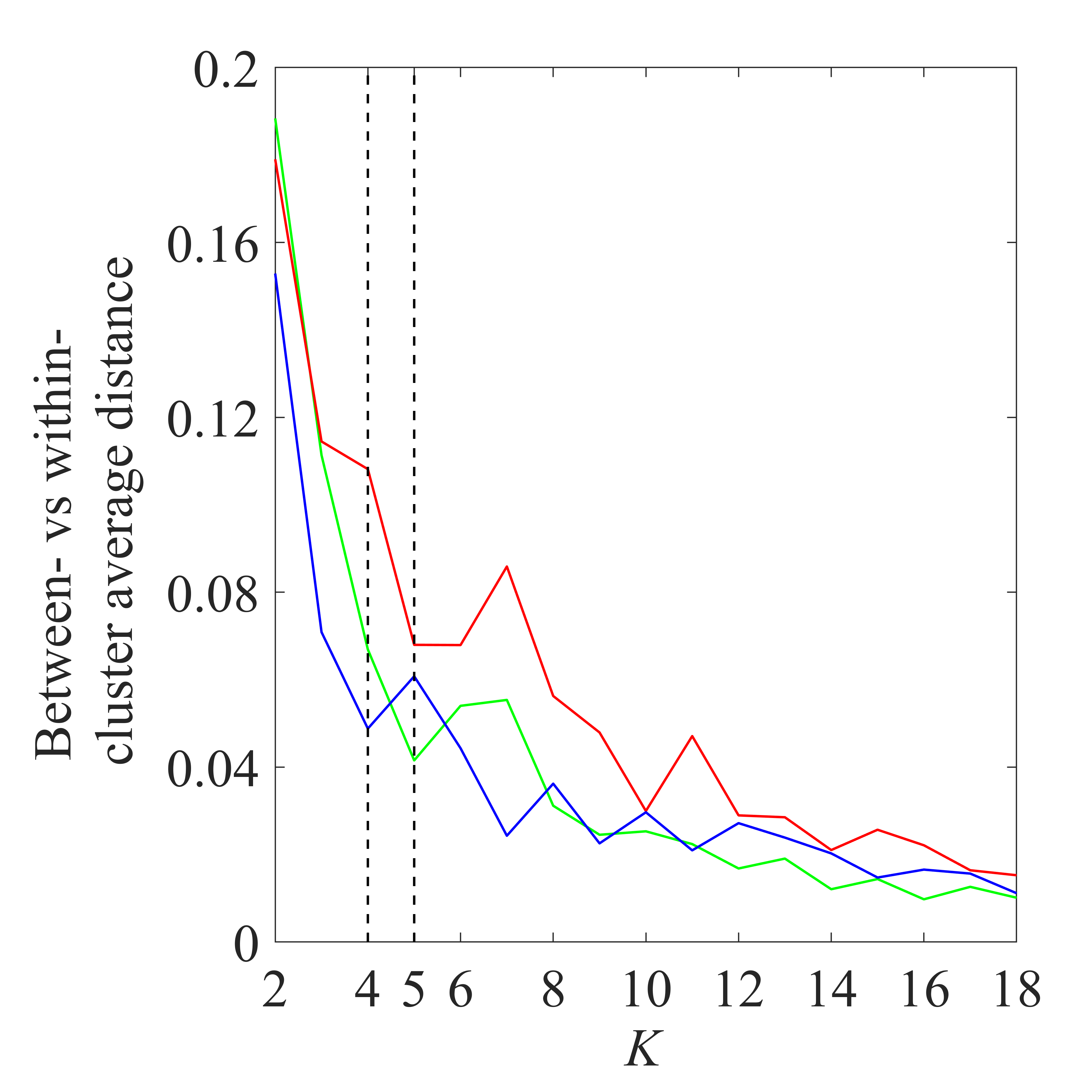}
    \includegraphics[width=0.24\textwidth]{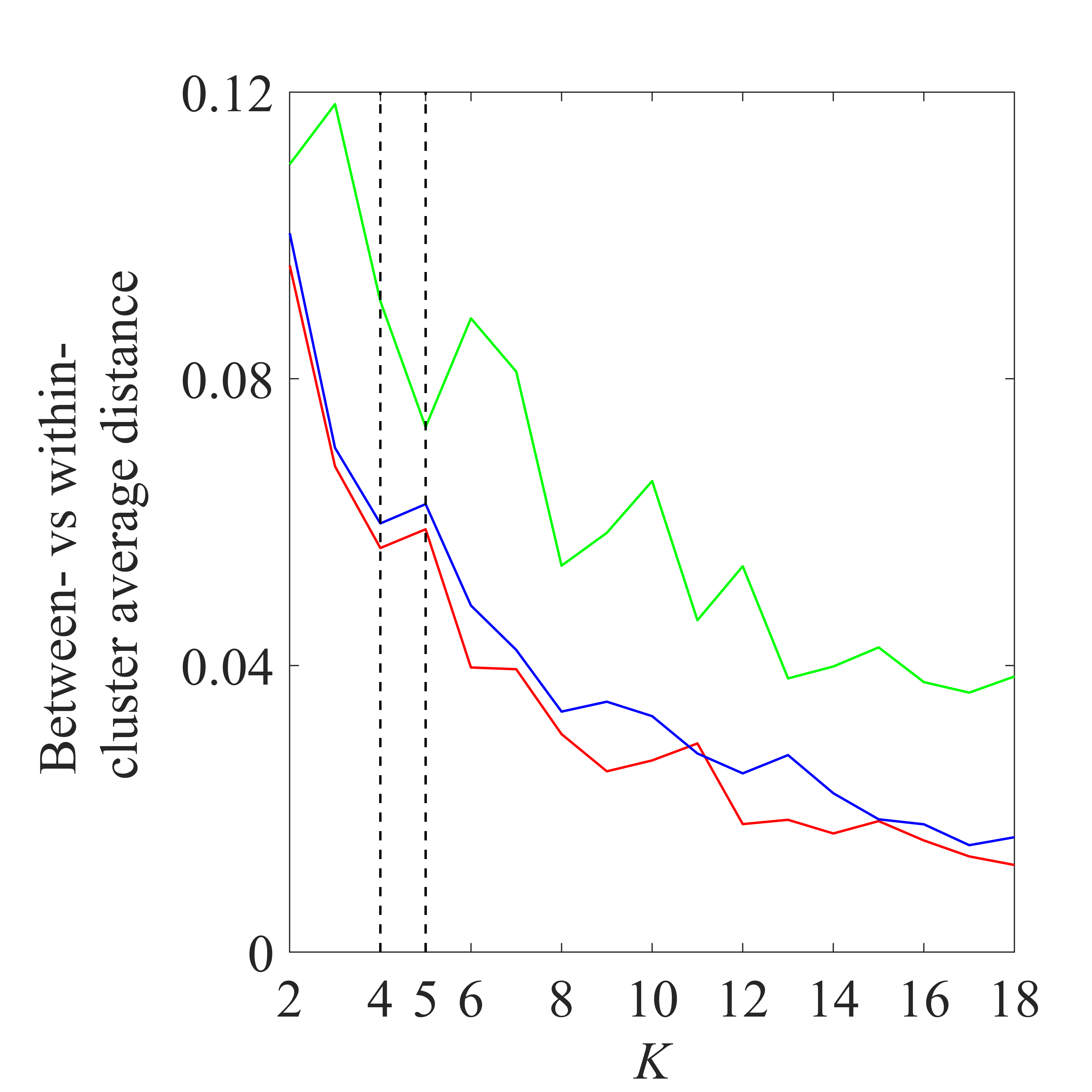}
    \caption{
    Green lines: $N=19;$
    Red lines: $N=38;$
    Blue lines: $N=57.$
    Selecting the numbers of clusters for $\alpha$- and $\beta$-clustering of each department in the email data set.  Our selection balances two considerations: (i). selecting a small $K$; (ii) and minimizing the between- vs within- cluster average distance.}
    \label{fig::data-example-tune-K}
\end{figure}
Next, in Figure \ref{fig::data-example}, we show the clustering result in a colored and marked scatter plot of all nodes when $N=19$ and the
clustering results when
$N=38$ and $N=57$
are in the Supplementary Material.
Each node's $X$- and $Y$-coordinates are its $\alpha$- and $\beta$-MDS coordinates, respectively.
Each node's color indicates its $\alpha$-clustering membership and marker shows its $\beta$-clustering membership.
\begin{figure}[h!]
    \centering
    \includegraphics[width=0.45\textwidth]{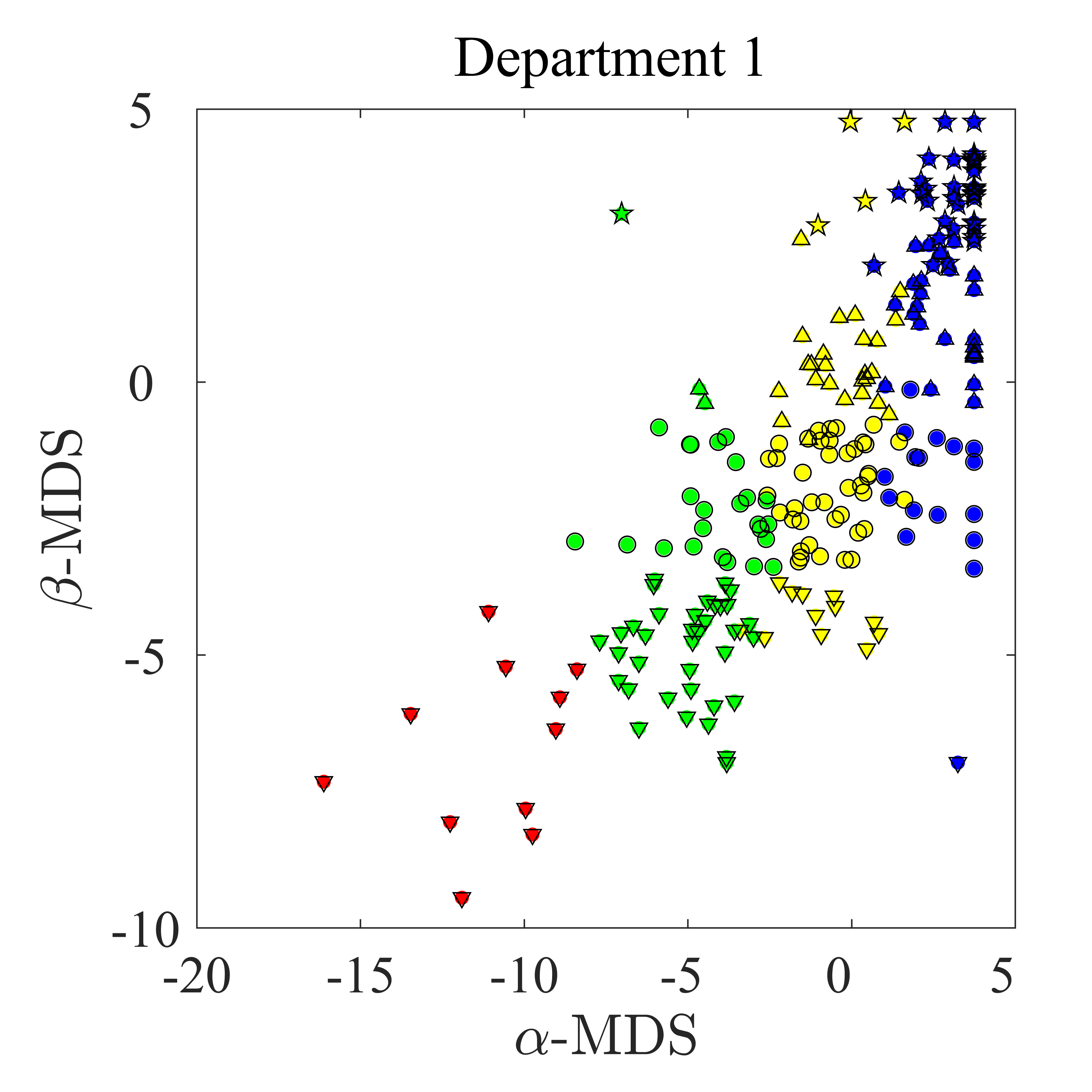}
    \includegraphics[width=0.45\textwidth]{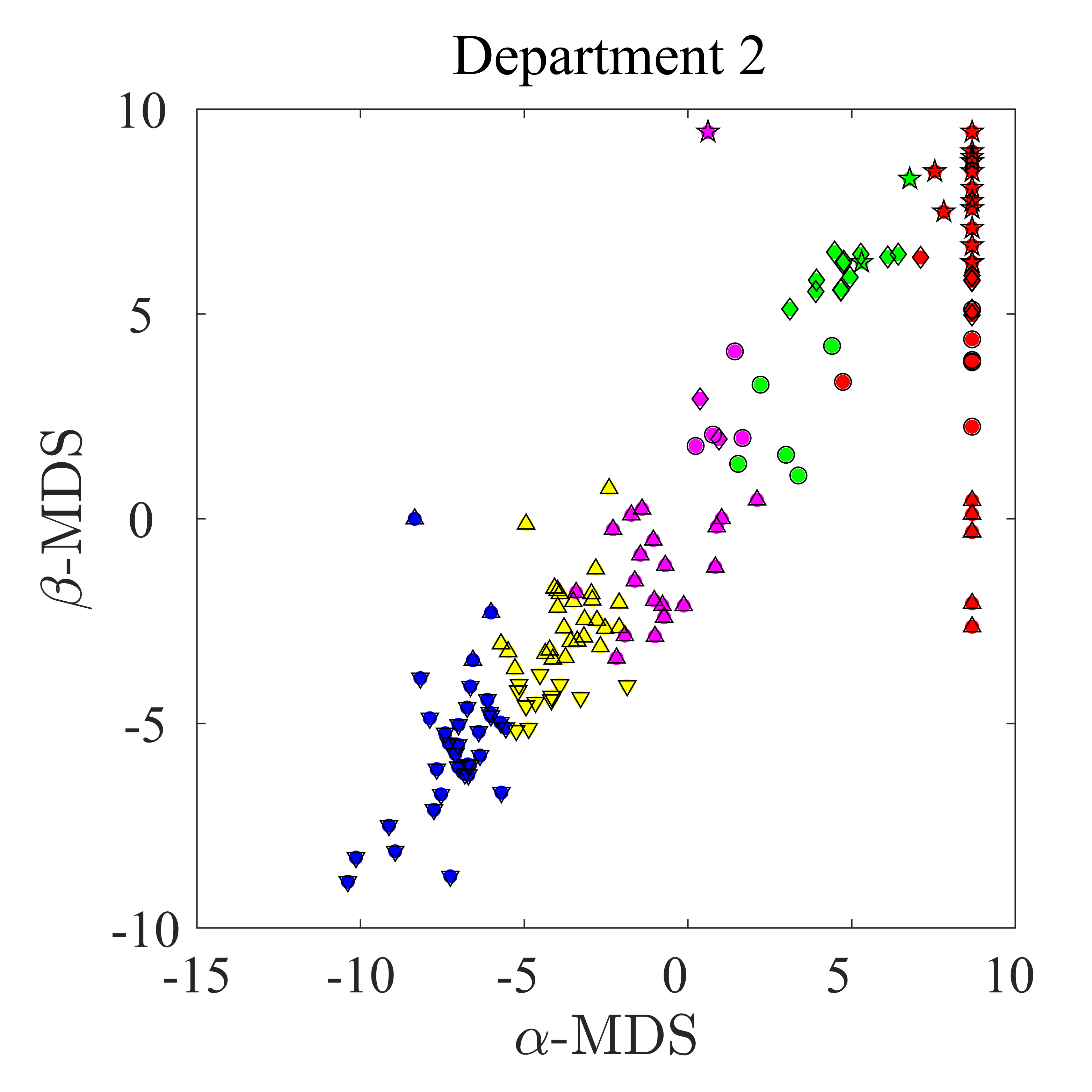}
    \includegraphics[width=0.45\textwidth]{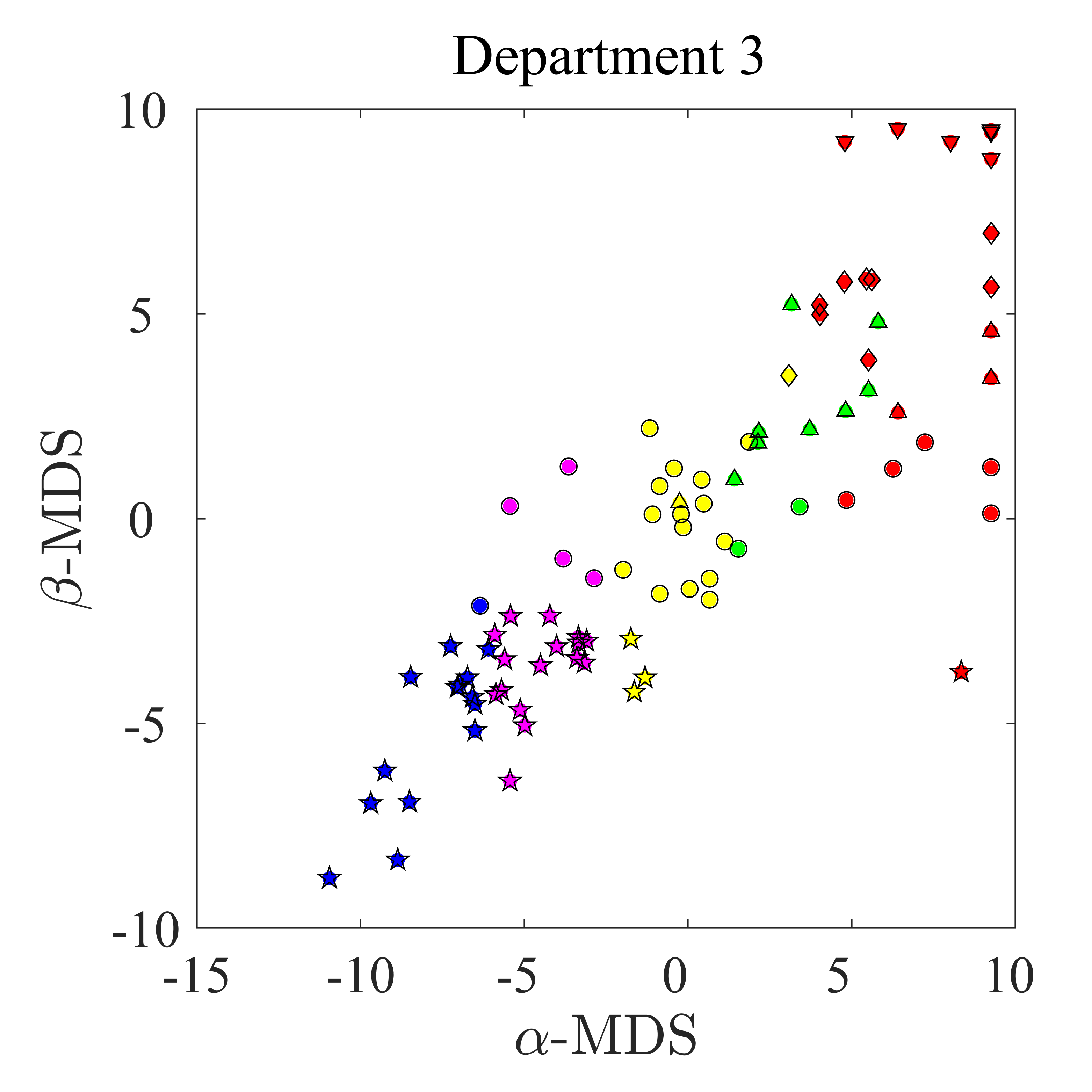}
    \includegraphics[width=0.45\textwidth]{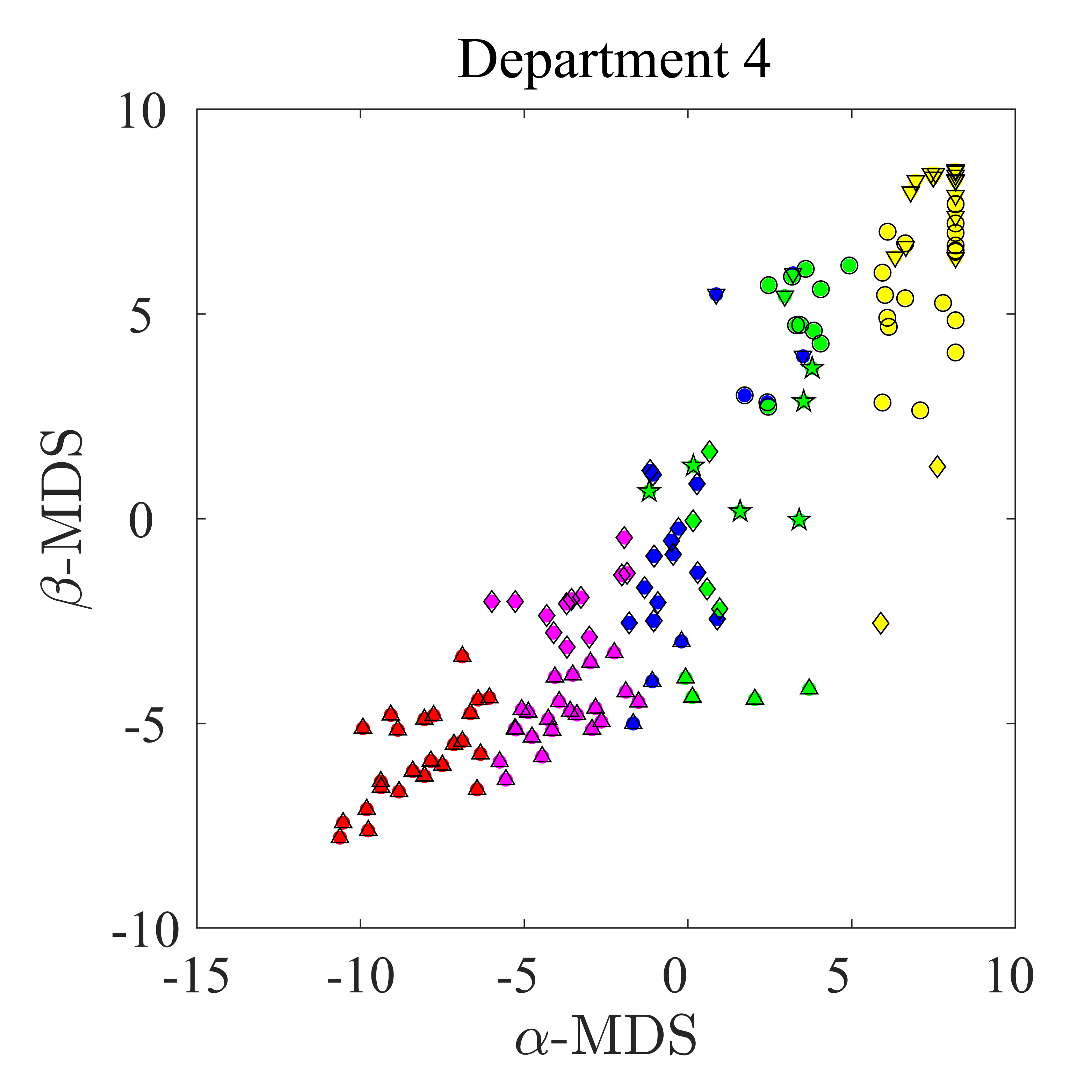}
    \caption{Illustration of $\hat\alpha_i(t)$'s and $\hat\beta_j(t)$'s.  Each point's $X$-coordinate comes from a multi-dimensional scaling of the $\ell_2$ distance matrix between $\hat\alpha_i(t)$'s; similarly obtain its $Y$-coordinate.  Color represents $\alpha$-clustering result, and marker represents $\beta$-clustering result.}
    \label{fig::data-example}
\end{figure}
Now we present some interpretations of the results in Figure \ref{fig::data-example}.
First, we notice that the $\alpha$- and $\beta$-clustering results mostly reflect nodes' $\alpha$- and $\beta$-MDS cooridnates, but they are not entirely aligned -- recall that $\alpha$- and $\beta$-clustering use their corresponding distance matrices, whereas their MDS coordinates can be viewed as the distance matrices nonlinearly projected into one-dimensional subspace.
All four plots in Figure \ref{fig::data-example} suggest that
the trends in which nodes' sender and receiver effects evolve over time are overall concordant with each other, since most points align along diagonal straight lines; but they do not fully align with each other, which is reflected by the difference between $\alpha$- and $\beta$- clustering results.

\section{Discussion}
\label{section::discussion}
% {\color{red}
% [TOPICS:]
% \begin{itemize}
%     \item Change point?
%     \item Network sparsity
%     \item Future work: extension to weighted networks
% \end{itemize}
% }
In this paper, we extend the $\beta$-model for directed binary networks to the dynamic setting by letting the parameters evolve over time.  This work can be viewed as a first step in the exploration of this interesting direction.  Many intriguing questions remain.  Here, we select a few future directions for discussion.

First, our study concentrates on binary edges.  Retrospecting our approach, we see that the key ingredient is to kernel-smooth the partial differentiation formulae of the likelihood function at each time point, using the counterparts at adjacent time points.
For instance, we can extend the edge generation scheme from Bernoulli to Poisson with
$
P(a_{ij} = x) = e^{x(\alpha_i+\beta_j)}/x!\cdot  e^{-e^{\alpha_i+\beta_j}}, \quad x=0,1,2,\ldots.
$
The idea is very similar to how \citet{yan2016asymptoticsSIN} extended the ordinary $\beta$-model from Bernoulli to Poisson and many other general edge distributions.
We leave the detailed formulation and theoretical study to future work.
% For instance, consider a directed network with the following Poisson edge generation scheme:
% $$
%     P(a_{ij} = x) = \dfrac{e^{x(\alpha_i+\beta_j)}}{x!} e^{-e^{\alpha_i+\beta_j}}, \quad x=0,1,2,\ldots
% $$
% This yields a log-likelihood function:
% \begin{align}
%     L(\alpha,\beta) := \sum_{(i,j):i\neq j} \Big\{
%         (\alpha_i+\beta_j)a_{ij} - e^{\alpha_i+\beta_j}
%     \Big\},
% \end{align}
% ignoring constant terms.  Therefore, $\partial L\big/\partial \alpha_i = d_i - e^{\alpha_i}\big(\sum_{j:j\neq i}e^{\beta_j}\big)$ and $\partial L\big/\partial \beta_j = b_j - e^{\beta_j}\big(\sum_{i:i\neq j}e^{\alpha_i}\big)$.  Apply our treatment in this paper to these equations and smooth them over time.  We have the following estimation equations:
% \begin{align}
%     \sum_{\ell\in[N]}
%     K_h(t-T_\ell)
%     \Bigg\{
%     d_i(T_\ell) - e^{\alpha_i(t)}\Big(\sum_{j:j\neq i}e^{\beta_j(t)}\Big)
%     \Bigg\}
%     &=~0\\
%     \sum_{\ell\in[N]}
%     K_h(t-T_\ell)
%     \Bigg\{
%     b_j(T_\ell) - e^{\beta_j(t)}\Big(\sum_{i:i\neq j}e^{\alpha_i(t)}\Big)
%     \Bigg\}
%     &=~0
% \end{align}
% One can anticipate that similar analysis as we show in this paper can possibly establish consistency and asymptotic normality.
% In other words, despite we exclusively focus on binary edges in this paper, our method and analysis are enjoys significantly more generality than they seem and can be easily applied to much richer family of edge distributions \citep{yan2016asymptoticsSIN}.

The second discussion topic is to relax our global smoothness assumption on all $\alpha_i(t)$'s and $\beta_j(t)$'s to accommodate change points \citep{wang2021optimal,Bhamidi2018change,yu2021optimal}.
To tweak our approach to this end, we consider a single change point setting: $\alpha_i(t)$'s and $\beta_j(t)$'s contain a common jump point at $t=t_0$, and are otherwise smooth conformal to Condition \ref{condition-1}.  Suppose $t_0\in [a_1,b_1]\subset [a+C_a, b-C_b]$ for some global constants $C_a
$ and $C_b>0$ (we assume $a_1$ and $b_1$ are known, but they can be set conservatively, close to $a$ and $b,$ respectively).
A simple adaptation of our method for change point detection in this scenario is
\begin{align}
    \hat t_0
    := &~
    \arg\max_{t:t\in[a_1,b_1]}
    \Big\|
        \Big(
            \hat\alpha_+(t) - \hat\alpha_-(t)
            ,
            \hat\beta_+(t) - \hat\beta_-(t)
        \Big)
    \Big\|_2^2,
    \notag\\
    \textrm{where $\alpha_i(t)=\big(\hat\alpha_+(t)\big)_i$ solves:}
    &~
    \sum_{\ell\in[N]: T_\ell\geq t}
    K_h(t-T_\ell)
    \Bigg\{
    d_i(T_\ell) - e^{\alpha_i(t)}\Big(\sum_{j:j\neq i}e^{\beta_j(t)}\Big)
    \Bigg\}
    =0,
    \notag\\
    \textrm{and $\alpha_i(t)=\big(\hat\alpha_-(t)\big)_i$ solves:}
    &~
    \sum_{\ell\in[N]: T_\ell< t}
    K_h(t-T_\ell)
    \Bigg\{
    d_i(T_\ell) - e^{\alpha_i(t)}\Big(\sum_{j:j\neq i}e^{\beta_j(t)}\Big)
    \Bigg\}
    =0,
    \label{discussion::change-point-method-1}
\end{align}
and similarly define $\hat\beta_+(t)$ and $\hat\beta_-(t)$.
Aside from \eqref{discussion::change-point-method-1}, another way to extend the model is to generalize the weight function, taking into account not only the closeness in time (reflected by $|t-T_\ell|$ for each $t$ and $T_\ell$), but also the closeness between their estimated parameters \citep{li2011multiscale}.
Both interesting extensions require further study.

\section*{Acknowledgments}
We thank the Editor Professor Jaakko Peltonen, Associate Editor and two referees for their insightful and constructive comments that led to significant improvements of this paper.
We also thank Professor Chenlei Leng and Professor Subhabrata Sen for helpful discussions.
Qu is supported by the National Natural Science Foundation of China (no: 11771171 and 12001219).
The supplementary material is available by sending emails to qulianq@ccnu.edu.cn.

\setlength{\itemsep}{-1.5pt}
\setlength{\bibsep}{0ex}
\normalem
\bibliography{reference}

\begin{thebibliography}{}

\bibitem[\protect\astroncite{Bhamidi et~al.}{2018}]{Bhamidi2018change}
Bhamidi, S., Jin, J., and Nobel, A. (2018).
\newblock Change point detection in network models: Preferential attachment and
  long range dependence.
\newblock {\em The Annals of Applied Probability}, 28(1):35--78.

\bibitem[\protect\astroncite{Bubeck}{2015}]{MAL-050}
Bubeck, S. (2015).
\newblock Convex optimization: Algorithms and complexity.
\newblock {\em Foundations and Trends in Machine Learning}, 8(3-4):231--357.

\bibitem[\protect\astroncite{Chatterjee et~al.}{2011}]{chatterjee2011random}
Chatterjee, S., Diaconis, P., and Sly, A. (2011).
\newblock Random graphs with a given degree sequence.
\newblock {\em The Annals of Applied Probability}, 21(4):1400--1435.

\bibitem[\protect\astroncite{Chen et~al.}{2021}]{chen2019analysis}
Chen, M., Kato, K., and Leng, C. (2021).
\newblock Analysis of networks via the sparse $\beta$-model.
\newblock {\em Journal of the Royal Statistical Society: Series B (Statistical
  Methodology)}, 83(5):887--910.

\bibitem[\protect\astroncite{Cui and Chen}{2022}]{cui-chen-2023}
Cui, X. and Chen, Y. (2022+).
\newblock Inferring social influence in dynamic networks.
\newblock {\em Statistical Sinica (In press)}, (in press).

\bibitem[\protect\astroncite{Fan et~al.}{2022}]{fan2022asymptotic}
Fan, Y., Jiang, B., Yan, T., and Zhang, Y. (2022).
\newblock Asymptotic theory in bipartite graph models with a growing number of
  parameters.
\newblock {\em Canadian Journal of Statistics (In press)}.

\bibitem[\protect\astroncite{Graham}{2017}]{graham2017econometric}
Graham, B.~S. (2017).
\newblock An econometric model of network formation with degree heterogeneity.
\newblock {\em Econometrica}, 85(4):1033--1063.

\bibitem[\protect\astroncite{Hillar and Wibisono}{2013}]{hillar2013maximum}
Hillar, C. and Wibisono, A. (2013).
\newblock Maximum entropy distributions on graphs.
\newblock {\em arXiv preprint arXiv:1301.3321}.

\bibitem[\protect\astroncite{Leskovec et~al.}{2007}]{leskovec2007graph}
Leskovec, J., Kleinberg, J., and Faloutsos, C. (2007).
\newblock Graph evolution: Densification and shrinking diameters.
\newblock {\em ACM transactions on Knowledge Discovery from Data (TKDD)},
  1(1):2--es.

\bibitem[\protect\astroncite{Li et~al.}{2011}]{li2011multiscale}
Li, Y., Zhu, H., Shen, D., Lin, W., Gilmore, J.~H., and Ibrahim, J.~G. (2011).
\newblock Multiscale adaptive regression models for neuroimaging data.
\newblock {\em Journal of the Royal Statistical Society: Series B (Statistical
  Methodology)}, 73(4):559--578.

\bibitem[\protect\astroncite{Masry}{1996}]{masry1996multivariate}
Masry, E. (1996).
\newblock Multivariate local polynomial regression for time series: uniform
  strong consistency and rates.
\newblock {\em Journal of Time Series Analysis}, 17(6):571--599.

\bibitem[\protect\astroncite{Olhede and Wolfe}{2012}]{olhede2012degree}
Olhede, S.~C. and Wolfe, P.~J. (2012).
\newblock Degree-based network models.
\newblock {\em arXiv preprint arXiv:1211.6537}.

\bibitem[\protect\astroncite{Paranjape et~al.}{2017}]{paranjape2017motifs}
Paranjape, A., Benson, A.~R., and Leskovec, J. (2017).
\newblock Motifs in temporal networks.
\newblock In {\em Proceedings of the tenth ACM international conference on web
  search and data mining}, pages 601--610.

\bibitem[\protect\astroncite{Rinaldo et~al.}{2013}]{rinaldo2013maximum}
Rinaldo, A., Petrovic, S., and Fienberg, S.~E. (2013).
\newblock Maximum likelihood estimation in the $\beta$-model, supplementary
  materials.
\newblock {\em Ann. Stat.}, 41:1085--1110.

\bibitem[\protect\astroncite{Shao et~al.}{2021}]{zhang2021L2}
Shao, M., Zhang, Y., Wang, Q., Zhang, Y., Luo, J., and Yan, T. (2021).
\newblock L-2 regularized maximum likelihood for $\beta $-model in large and
  sparse networks.
\newblock {\em arXiv preprint arXiv:2110.11856}.

\bibitem[\protect\astroncite{Stein and Leng}{2021}]{stein2021sparse}
Stein, S. and Leng, C. (2021).
\newblock A sparse random graph model for sparse directed networks.
\newblock {\em arXiv preprint arXiv:2108.09504}.

\bibitem[\protect\astroncite{Wang et~al.}{2021}]{wang2021optimal}
Wang, D., Yu, Y., and Rinaldo, A. (2021).
\newblock Optimal change point detection and localization in sparse dynamic
  networks.
\newblock {\em The Annals of Statistics}, 49(1):203--232.

\bibitem[\protect\astroncite{Yan et~al.}{2016a}]{yan2016asymptotics}
Yan, T., Leng, C., and Zhu, J. (2016a).
\newblock Asymptotics in directed exponential random graph models with an
  increasing bi-degree sequence.
\newblock {\em The Annals of Statistics}, 44(1):31--57.

\bibitem[\protect\astroncite{Yan et~al.}{2016b}]{yan2016asymptoticsSIN}
Yan, T., Qin, H., and Wang, H. (2016b).
\newblock Asymptotics in undirected random graph models parameterized by the
  strengths of vertices.
\newblock {\em Statistica Sinica}, pages 273--293.

\bibitem[\protect\astroncite{Yan and Xu}{2013}]{yan2013central}
Yan, T. and Xu, J. (2013).
\newblock A central limit theorem in the $\beta$-model for undirected random
  graphs with a diverging number of vertices.
\newblock {\em Biometrika}, 100(2):519--524.

\bibitem[\protect\astroncite{Yu et~al.}{2021}]{yu2021optimal}
Yu, Y., Padilla, O. H.~M., Wang, D., and Rinaldo, A. (2021).
\newblock Optimal network online change point localisation.
\newblock {\em arXiv preprint arXiv:2101.05477}.

\bibitem[\protect\astroncite{Zhang and Xia}{2022}]{zhang2022edgeworth}
Zhang, Y. and Xia, D. (2022).
\newblock Edgeworth expansions for network moments.
\newblock {\em The Annals of Statistics}, 50(2):726--753.

\end{thebibliography}
\bibliographystyle{apa}

\section*{Appendix: Proofs of Theorem \ref{theorem-1}, Theorem \ref{theorem-2} and Proposition \ref{proposition-1}}

\begin{proof}[Proof of Theorem \ref{theorem-1}.]
We show the consistency of the estimators
by checking the conditions of Lemma 1 in the online Supplementary Material.
We first calculate
$$r=\| (F'(t,\theta^{(0)}(t))/(Nn))^{-1}F(t,\theta^{(0)}(t)/(Nn))\|_{\infty}.$$
Note that $-F'(t,\theta(t))/(Nn) \in \mathcal{L}_n(m,M).$
As shown in Lemma 6, we have
$$\frac{1}{Nn}F'_i(t,\theta(t)) \le \epsilon+\frac{1}{Nn}\mathbb{E}F'_i(t,\theta(t)),~i\in[2n-1],$$
where $\epsilon > 0$ is some constant.
In addition, by some arguments similar to the proof of Lemma 6,
we have
\begin{align*}
m&=\frac{1}{NnQ_{Nnh}}\sum_{\ell \in [N]}K_h(t-T_{\ell})\ge
(f(t)+\frac{1}{2}k_{21}f''(t)h^2+o(h^2))/(2nQ_{Nnh}),\\
M&=\frac{1}{4Nn}\sum_{\ell \in [N]}K_h(t-T_{\ell})\le
(f(t)+\frac{1}{2}k_{21}f''(t)h^2+o(h^2))/(2n).
\end{align*}
By Lemmas 2 and 3, we have
\begin{align*}
r=&\|(F'(t,\theta^{(0)}(t))/(Nn))^{-1}F(t,\theta^{(0)}(t)/(Nn)) \|_{\infty}\\
\le& \Bigg(\frac{2c_4(2n-1) M^2\|F(t,\theta^{(0)}(t)/(Nn))\|_{\infty}}{ m^3(n-1)^2}+\max_{i}\frac{|F_i(t,\theta^{(0)}(t))/(Nn)|}{\{ V(t,\theta^{(0)}(t)) \big\}_{i,i}}
+\frac{|F_{2n}(t,\theta^{(0)}(t))/(Nn)|}{\{ V(t,\theta^{(0)}(t)) \big\}_{2n,2n}}\Bigg)\\
\le&\Bigg(\frac{2c_4(2n-1)M^2}{m^3(n-1)^2}+\frac{c_5}{(2n-1) mN}\Bigg)\|F(t,\theta^{(0)}(t))/(Nn)\|_{\infty}\\
\le& O\Bigg(\frac{Q_{Nnh}^3+Q_{Nnh}}{f(t)+k_{21}f''(t)h^2/2+o(h^2)}\Bigg)\Big(\sqrt{C_1\frac{\log (Nnh)}{Nnh}}+C_2C_{Nn}h^2\Big)
\\
\le& O\Bigg(Q_{Nnh}^3\Big(\sqrt{C_1\frac{\log (Nnh)}{Nnh}}+C_2C_{Nn}h^2\Big)\Bigg),
\end{align*}
where $c_4$ and $c_5$ are some positive constants,
and
$$\{V(t,\theta(t))\}_{2n,2n}
=2\sum_{i\in [n]}\{V(t,\theta(t))\}_{i,i}-\sum_{i\in [n]}\sum_{j \in [2n-1]}\{V(t,\theta(t))\}_{i,j}.$$
The last equation holds under Condition \ref{condition-2}.
\noindent By Lemma 1, note that
\begin{align*}
\rho = \frac{c_1(2n-1)M^2H_1}{2m^3n^2}+\frac{H_2}{(n-1)m}\le O\Big(Q_{Nnh}^3+Q_{Nnh}\Big)
\end{align*}
with probability tending to one, where $c_1$ is some positive constant.
Therefore,
\begin{align*}
\rho r & = O\Big(Q_{Nnh}^3+Q_{Nnh}\Big) \times O\Big(Q_{Nnh}^3\big(\sqrt{C_1\frac{\log (Nnh)}{Nnh}}+C_2C_{Nn}h^2\big)\Big)\\
& = O\Big(Q_{Nnh}^6\sqrt{C_1\frac{\log (Nnh)}{Nnh}}+Q_{Nnh}^4\sqrt{C_1\frac{\log (Nnh)}{Nnh}}
+C_2C_{Nn}Q_{Nnh}^6h^2+C_2C_{Nn}Q_{Nnh}^4h^2\Big)\\
&=o(1).
\end{align*}
If $Q_{Nnh}=o\Big(1/(\sqrt{C_1\log (Nnh)/(Nnh)}+C_2C_{Nn}h^2)^{1/6}\Big),$
which is less than $1/2$ as $Nnh\rightarrow \infty.$
Consequently,
$$\sup_{t \in [a,b]}\| \hat \theta(t)-\theta^*(t) \|_{\infty}=O_p\Bigg(Q_{Nnh}^3\Big(\sqrt{C_1\frac{\log (Nnh)}{Nnh}}+C_2C_{Nn}h^2\Big)\Bigg)=o_p(1).$$
It completes the proof.	
\end{proof}
\begin{proof}[Proof of Theorem \ref{theorem-2}.]
By Theorem \ref{theorem-1}, we have
$$\hat b_{Nnh}:=\sup_{t \in [a,b]}\max_{i \in [2n-1]}\Big|\{\hat\theta(t)-\theta^*(t)\}_i\Big|=O_p\Big\{Q_{Nnh}^3(\sqrt{C_1\frac{\log (Nnh)}{Nnh}}+C_2C_{Nn}h^2)\Big\}.$$
Let $\hat \omega_{ij}(t)=\hat \alpha_i(t)+\hat \beta_j(t)-\alpha^*_i(t)-\beta^*_j(t)$.
By Taylor's expansion for $F(t,\hat\theta(t)),$ we have
\begin{align*}
F(t,\hat\theta(t))
=&F(t,\theta^*(t))+F'(t,\theta^*(t))(\hat {\mathbf\theta}(t)-\mathbf \theta^*(t))+\hat{ \gamma}(t,\theta^*(t))\\
=&F(t,\theta^*(t))+\frac{1}{Nn}F'(t,\theta^*(t))(Nn)(\hat {\mathbf\theta}(t)-\mathbf \theta^*(t))+\hat{ \gamma}(t,\theta^*(t))\\
=&F(t,\theta^*(t))-V(t,\theta^*(t))(Nn)(\hat {\mathbf\theta}(t)-\mathbf \theta^*(t))+\hat{ \gamma}(t,\theta^*(t)),
\end{align*}
which implies that
\begin{align*}
&\hat {\mathbf\theta}(t)-\mathbf \theta^*(t)-\frac{1}{Nn}V^{-1}(t,\theta^*(t))\mathbb{E}F(t,\theta^*(t))
\\
=&\frac{1}{Nn}V^{-1}(t,\theta^*(t))\big\{F(t,\theta^*(t))-\mathbb{E}F(t,\theta^*(t))\big\}
+\frac{1}{Nn}V^{-1}(t,\theta^*(t))\hat{\gamma}(t,\theta^*(t)).
\end{align*}
Here
$\hat{\gamma}(t,\theta^*(t))=\Big(\hat{\gamma}_1(t,\theta^*(t)),\hat{\gamma}_2(t,\theta^*(t)),\dots,\hat{\gamma}_{2n-1}(t,\theta^*(t))\Big)^\top,$
where
\begin{align*}
\hat{\gamma}_i(t,\theta^*(t))&=\sum_{\substack{j:1\leq j\leq n\\j\neq i}}\sum_{\ell \in [N]}K_h(t-T_{\ell})\hat{\gamma}_{ij}(t,\theta^*(t)),\quad i \in [n],\\
\hat{\gamma}_{n+j}(t,\theta^*(t))&=\sum_{\substack{i:1\leq i\leq n\\i\neq j}}\sum_{\ell \in [N]}K_h(t-T_{\ell})\hat{\gamma}_{ij}(t,\theta^*(t)),\quad j\in [n-1],
\end{align*}
with
$$\hat\gamma_{ij}(t,\theta^*(t))=-\frac{e^{\alpha^*_i(t)+\beta^*_j(t)+\hat\omega_{ij}(t)\phi_{ij}}(1-e^{\alpha^*_i(t)+\beta^*_j(t)+\phi_{ij}\hat\omega_{ij}(t)})}{2(1+e^{\alpha^*_i(t)+\beta^*_j(t)+\phi_{ij}\hat\omega_{ij}(t)})^3}\hat\omega^2_{ij}(t),$$
and $0 \le \phi_{ij} \le 1.$
Since $\big|e^x(1-e^x)/(1+e^x)^3\big| \le 1$, we have
\begin{align*}
|\hat{\gamma}_{ij}(t,\theta^*(t))| &\le |\hat\omega^2_{ij}(t)/2| \le 2 \hat b_{Nnh}^2,
\\
|\hat{\gamma}_i(t,\theta^*(t))|&=\sum_{\substack{j:1\leq j\leq n\\j\neq i}}\sum_{\ell \in [N]}K_h(t-T_{\ell})|\hat{\gamma}_{ij}(t,\theta^*(t))| \le 2N(n-1)\hat b_{Nnh}^2/h,
\quad i \in [n],
\\
|\hat{\gamma}_{n+j}(t,\theta^*(t))|&=\sum_{\substack{i:1\leq i\leq n\\i\neq j}}\sum_{\ell \in [N]}K_h(t-T_{\ell})|\hat{\gamma}_{ij}(t,\theta^*(t))| \le 2N(n-1)\hat b_{Nnh}^2/h,\quad j \in [n-1].
\end{align*}
By Lemmas 8-10 given below, we have
\begin{align*}
\Big\|\frac{1}{Nn}V^{-1}(t,\theta^*(t))\hat{\gamma}(t,\theta^*(t)\Big\|_{\max}=&o_p\big((Nnh)^{-1/2}\big),\\
\Big\|\frac{1}{Nn}V^{-1}(t,\theta^*(t))\mathbb{E}F(t,\theta^*(t))-k_{21}h^2\mu(t)
\Big\|_{\max}=&o_p\big((Nnh)^{-1/2}\big)
\end{align*}
and
\begin{align*}
&\frac{1}{Nn}V^{-1}(t,\theta^*(t))\big\{F(t,\theta^*(t))-\mathbb{E}F(t,\theta^*(t))\big\}\\
=&\frac{1}{Nn}\bar S(t,\theta^*(t))\big\{F(t,\theta^*(t))-\mathbb{E}F(t,\theta^*(t))\big\}
+o_p\big((Nnh)^{-1/2}\big),
\end{align*}
where $\mu(t)$ is defined in Lemma 9.
That is,
\begin{align*}
&\sqrt{Nnh}\Big[\{\hat {\mathbf\theta}(t)
-\mathbf \theta^*(t)\}_i-k_{21}h^2\mu_i(t)\Big]
\\
=&\sqrt{Nnh}\Big[\big\{\frac{1}{Nn}\bar S(t,\theta^*(t))(F(t,\theta^*(t))-\mathbb{E}F(t,\theta^*(t)))\big\}_i\Big]+o_p(1),~~i \in [2n-1].
\end{align*}
For any $i \in [n],$ it can be shown
\begin{align*}
&\sum_{\substack{j:1\leq j\leq n\\j\neq i}}\sum_{\ell \in [N]}\mathbb{E}\Big|K_h(t-T_{\ell})\big(a_{ij}(T_{\ell})-u_{ij}(t)\big)-\mathbb{E}\big\{K_h(t-T_{\ell})\big(a_{ij}(T_{\ell})-u_{ij}(t)\big)\big\}\Big|^3\\
\le &\sum_{\substack{j:1\leq j\leq n\\j\neq i}}\sum_{\ell \in [N]}
    \mathbb{E}\Big|K_h(t-T_{\ell})\big(a_{ij}(T_{\ell})-u_{ij}(t)\big)-\mathbb{E}\big\{K_h(t-T_{\ell})\big(a_{ij}(T_{\ell})-u_{ij}(t)\big)\big\}\Big|^2/h\\
    =&O\big(Nn/h^2\big),
\end{align*}
and by Lemma 5, we have
$$
\text{Var}\Bigg\{\sum_{\substack{j:1\leq j\leq n\\j\neq i}}\sum_{\ell \in [N]}K_h(t-T_{\ell})\big(a_{ij}(T_{\ell})-u_{ij}(t)\big)\Bigg\}=O(Nn/h).
$$
Therefore,
{\small
\begin{align*}
&\frac
        {\sum_{j\in [n],j\neq i}\sum_{\ell \in [N]}\mathbb{E}\Big|K_h(t-T_{\ell})\big(a_{ij}(T_{\ell})-u_{ij}(t)\big)-\mathbb{E}\big\{K_h(t-T_{\ell})\big(a_{ij}(T_{\ell})-u_{ij}(t)\big)\big\}\Big|^3}
        {\Big(\text{Var}\Big\{\sum_{j\in[n],j\neq i}\sum_{\ell \in [N]}K_h(t-T_{\ell})\big(a_{ij}(T_{\ell})-u_{ij}(t)\big)\Big\}\Big)^{3/2}}\\
    \le& O\Big(\frac{1}{\sqrt{Nnh}}\Big).
\end{align*}}
Similarly, for any $j \in [n-1],$
{\small
\begin{align*}
&\frac{\sum_{i\in[n],i\neq j}\sum_{\ell \in [N]}\mathbb{E}\Big|K_h(t-T_{\ell})\big(a_{ij}(T_{\ell})-u_{ij}(t)\big)-\mathbb{E}\big\{K_h(t-T_{\ell})\big(a_{ij}(T_{\ell})-u_{ij}(t)\big)\big\}\Big|^3}
{\Big(\text{Var}\Big\{\sum_{i\in[n],i\neq j}\sum_{\ell \in [N]}K_h(t-T_{\ell})\big(a_{ij}(T_{\ell})-u_{ij}(t)\big)\Big\}\Big)^{3/2}}\\
\le &O\Big(\frac{1}{\sqrt{Nnh}}\Big).
\end{align*}}
This implies that the condition for the Lyapunov's central limit theorem holds under Conditions \ref{condition-1}-\ref{condition-2},
$C^2_{Nn}Nnh^5Q^2_{Nnh}=o(1),$ $Nh \to \infty$ and
$$Q_{Nnh}=o\Big[1/\big\{(Nnh)^{1/18}(\sqrt{C_1\log (Nnh)/(Nnh)}+C_2C_{Nn}h^2)^2\big\}\Big].$$
\begin{sloppypar}
Thus for any fixed $1 \le p \le 2n-1$ and $t\in[a,b],$
$\sqrt{Nnh}\Big[\{\hat{\theta}_1(t) - \theta_1^*(t),
\ldots,\hat{\theta}_{p}(t) - \theta_{p}^*(t)\}^\top-k_{21}h^2\{\mu_1(t),\dots,\mu_p(t)\}^\top\Big]$
converges in distribution to a $p$-dimensional multivariate normal random vector with mean zero and covariance $k_{02}\Sigma(t,\theta^*(t)).$
\end{sloppypar}
% where for $q \in [p],$
%\begin{align*}
%\mu_q(t)
%=&
%\Bigg\{\frac{\sum_{j\in [n],j\neq q}\mu_{qj}(t)/n}{f(t)\big\{V_0(t,\theta^*(t))\big\}_{q,q}}\Bigg\}^{1-I_{n+1 \le q \le 2n-1}}
%+\Bigg\{\frac{\sum_{j\in [n],j\neq q-n}\mu_{j,q-n}(t)/n}{f(t)\big\{V_0(t,\theta^*(t))\big\}_{q,q}}\Bigg\}^{I_{n+1 \le q \le 2n-1}}\\
%+&(-1)^{n+1 \le q \le 2n-1}\frac{\sum_{j\in[n-1] }\mu_{jn}(t)/n}{f(t)\big\{V_0(t,\theta^*(t))\big\}_{2n,2n}}.
%\end{align*}
A direct calculation yields
\begin{align*}
\text{Cov}\Bigg(\frac{\sqrt{Nnh}}{Nn}\bar S(t,\theta^*(t))\big\{F(t,\theta^*(t))-\mathbb{E}F(t,\theta^*(t))\big\}\Bigg)
=\bar S(t,\theta^*(t))\Omega(t,\theta^*(t))\bar S(t,\theta^*(t))^\top,
\end{align*}
where $\Omega(t,\theta^*(t))=\text{Cov}\big\{\sqrt{h /(Nn)}F(t,\theta^*(t))\big\}.$\\
We next calculate each element of $\Omega(t,\theta^*(t)).$
For $i \in [n],$ by Lemma 5,
the $i$th diagonal element of $\Omega(t,\theta^*(t))$ is
\begin{align*}
\big\{\Omega(t,\theta^*(t))\big\}_{i,i}=&\frac{h}{Nn}\text{Var}\big\{F_i(t,\theta^*(t))\big\}\\
%=&\frac{1}{n}\sum_{j\neq i}^{n}\Big[\Big(u_{ij}(t)(1-u_{ij}(t))f(t)\int\mathcal{K}^2(w)dw\Big)+h\int\mathcal{K}^2(w)wdw\Big(u_{ij}(t)(1-u_{ij}(t))f(t)\Big)^{'}+o(C_{Nn}h)\Big]\\
=&\frac{1}{n}\sum_{\substack{j:1\leq j\leq n\\j\neq i}}\Big[k_{02}f(t)u_{ij}(t)(1-u_{ij}(t))
+k_{12}h\tilde{\psi}_{ij}(t)+o(C_{Nn}h)\Big]
\\
=&\frac{1}{n}\sum_{\substack{j:1\leq j\leq n\\j\neq i}}k_{02}f(t)u_{ij}(t)(1-u_{ij}(t))
+\frac{1}{n}\sum_{\substack{j:1\leq j\leq n\\j\neq i}}k_{12}h\tilde{\psi}_{ij}(t)
+o(C_{Nn}h)
\\
=&k_{02}f(t)\big\{V_0(t,\theta^*(t))\big\}_{i,i}
+k_{12}h\big\{\tilde{V}(t,\theta^*(t))\big\}_{i,i}
+o(C_{Nn}h)
\\
=&k_{02}\big\{\bar V(t,\theta^*(t))\big\}_{i,i}
+k_{12}h\big\{\tilde{V}(t,\theta^*(t))\big\}_{i,i}
+o(C_{Nn}h)+O(Nnh^5C_{Nn}^2),
\end{align*}
where $k_{12}:=\int vK^2(v)dv,$
$\tilde{\psi}_{ij}(t):=
%&\big(u_{ij}(t)(1-u_{ij}(t))f(t)\big)'
u_{ij}(t)(1-u_{ij}(t))f'(t)+f(t)u'_{ij}(t)(1-2u_{ij}(t)),$ and
%$\tilde{V}(t,\theta(t))=(\tilde{\psi}_{ij}(t,\theta(t)))$ is a $(2n-1) \times (2n-1)$ matrix.
$$\big\{\tilde{V}(t,\theta^*(t))\big\}_{i,i}
=\frac{1}{n}\sum_{\substack{j:1\leq j\leq n\\j\neq i}}\tilde{\psi}_{ij}(t).
$$
%\begin{align*}
%\tilde{\psi}_{2n,i}(t) =& \tilde{\psi}_{i,2n}(t) := %\tilde{\psi}_{ii}(t) - \sum_{j=1,j \neq i}^{2n-1} %\tilde{\psi}_{ij}(t),~~i\in[2n-1].
%\end{align*}
For any $j\in[n-1],$
\begin{align*}
\big\{\Omega(t,\theta^*(t))\big\}_{n+j,n+j}=&\frac{h}{Nn}\text{Var}\big\{F_{n+j}(t,\theta^*(t))\big\}
\\
=&k_{02}\big\{\bar V(t,\theta^*(t))\big\}_{n+j,n+j}
+k_{12}h\big\{\tilde{V}(t,\theta^*(t))\big\}_{n+j,n+j}\\
&+o\big(C_{Nn}h\big)+O\big(Nnh^5C^2_{Nn}\big),
\end{align*}
where
$\big\{\tilde{V}(t,\theta^*(t))\big\}_{n+j,n+j}
=\frac{1}{n}\sum_{\substack{i:1\leq i\leq n\\i\neq j}}\tilde{\psi}_{ij}(t).
$
%Then for $i\in [n],i\neq j,$
%\begin{align*}
%Var\Big\{\frac{\sqrt{Nnh} }{Nn}F_i(t,\theta^*(t))\Big\}
%=&k_{02}\big\{\bar V(t,\theta^*(t))\big\}_{i,i}
%+k_{12}h\big\{\tilde{V}(t,\theta^*(t))\big\}_{i,i}
%+o(C_{Nn}h)+O\Big(Nnh^5C^2_{Nn}\Big).
%\end{align*}
In addition,
\begin{align*}
\big\{\Omega(t,\theta^*(t))\big\}_{i,j}=&\frac{h}{Nn}\text{Cov}\Big\{F_{i}(t,\theta^*(t)),F_{j}(t,\theta^*(t))\Big\}
=0,~i \in [n],~j \in [n],~i\neq j,\\
\big\{\Omega(t,\theta^*(t))\big\}_{n+i,n+j}=&\frac{h}{Nn}\text{Cov}\Big\{F_{n+i}(t,\theta^*(t)),
F_{n+j}(t,\theta^*(t))\Big\}
=0,~i,j \in[n-1],~i\neq j,\\
\big\{\Omega(t,\theta^*(t))\big\}_{i,n+i}=&\frac{h}{Nn}\text{Cov}\Big\{F_{i}(t,\theta^*(t)),
F_{n+i}(t,\theta^*(t))\Big\}
=0,~i\in[n-1],\\
\big\{\Omega(t,\theta^*(t))\big\}_{n+i,i}=&\frac{h}{Nn}\text{Cov}\Big\{F_{n+i}(t,\theta^*(t)),F_{i}(t,\theta^*(t))\Big\}
=0,~i\in[n-1].
\end{align*}
For any $i\in[n],j\in[n-1],i\neq j,$
\begin{align*}
\big\{\Omega(t,\theta^*(t))\big\}_{i,n+j}=&\frac{h}{Nn}\text{Cov}\Big\{F_{i}(t,\theta^*(t)),F_{n+j}(t,\theta^*(t))\Big\}
\\
=&k_{02}\big\{\bar V(t,\theta^*(t))\big\}_{i,n+j}
+k_{12}h\big\{\tilde{V}(t,\theta^*(t))\big\}_{i,n+j}
+o\big(C_{Nn}h/n\big)+O\big(Nh^5C^2_{Nn}\big),
\\
\big\{\Omega(t,\theta^*(t))\big\}_{n+j,i}=&\frac{h}{Nn}\text{Cov}\Big\{F_{n+j}(t,\theta^*(t)),
F_{i}(t,\theta^*(t))\Big\}
\\
=&k_{02}\big\{\bar V(t,\theta^*(t))\big\}_{n+j,i}
+k_{12}h\big\{\tilde{V}(t,\theta^*(t))\big\}_{n+j,i}
+o\big(C_{Nn}h/n\big)+O\big(Nh^5C^2_{Nn}\big),
\end{align*}
where
$
\big\{\tilde{V}(t,\theta^*(t))\big\}_{n+j,i}
=\big\{\tilde{V}(t,\theta^*(t))\big\}_{i,n+j}
=\tilde\psi_{ij}(t).
$
Now we establish
\begin{align*}
&\Big\|\bar S(t,\theta^*(t))\Omega(t,\theta^*(t))\bar S(t,\theta^*(t))-k_{02}\bar S(t,\theta^*(t))\bar{V}(t,\theta^*(t))\bar S(t,\theta^*(t))\Big\|_{\max}=o(1).
\end{align*}
For this, define $\tilde{V}(t,\theta^*(t))\in \mathbb{R}^{(2n-1)\times (2n-1)}$
and its $(i,j)$th element is $\big\{\tilde{V}(t,\theta^*(t))\big\}_{i,j}.$
Then, we have
$$
\Omega(t,\theta^*(t))=k_{02}\bar V(t,\theta^*(t))
+k_{12}h\tilde{V}(t,\theta^*(t))
+\Big[o\big(C_{Nn}h\big)+O\big(Nnh^5C^2_{Nn}\big)\Big]\mathbb{D}.
$$
Here $\mathbb{D}=(d_{ij})\in\mathbb{R}^{(2n-1)\times (2n-1)},$
and $d_{ii}=1~(i\in[n]),$
$d_{n+j,n+j}=1~(j\in[n-1]),$
$d_{i,n+j}=d_{n+i,j}=1/n~(i\in [n],~j\in [n-1],~i\neq j),$
and $d_{ij}=0~\text{otherwise}.$
We first show
$$
\Big\|\bar S(t,\theta^*(t))\tilde{V}(t,\theta^*(t))\bar S(t,\theta^*(t))\Big\|_{\max}=O(C_{Nn}Q_{Nnh}^2).
$$
Write
\begin{align*}
\bar S(t,\theta^*(t))\tilde{V}(t,\theta^*(t))=\Big(\sum_{k=1}^{2n-1}\big\{\bar S(t,\theta^*(t))\big\}_{ik}\tilde{\psi}_{kj}(t,\theta^*(t))\Big).
\end{align*}
Specifically, for $i=j \in [n],$ we have
\begin{align*}
\Big|\sum_{k=1}^{2n-1}\big\{\bar S(t,\theta^*(t))\big\}_{i,k}\tilde{\psi}_{kj}(t,\theta^*(t))\Big|
\le& \frac{|\tilde{\psi}_{ii}(t,\theta^*(t))|}{\{\bar V(t,\theta^*(t))\}_{i,i}}
+\frac{|\tilde{\psi}_{ii}(t,\theta^*(t))-\sum_{k=1,k\neq j}^{2n-1}\tilde{\psi}_{kj}(t,\theta^*(t))|}{\{\bar V(t,\theta^*(t))\}_{2n,2n}}
\\
=&\frac{|\tilde{\psi}_{ii}(t,\theta^*(t))|}{\{\bar V(t,\theta^*(t))\}_{i,i}}+\frac{|\tilde{\psi}_{2n,j}(t,\theta^*(t))|}{\{\bar V(t,\theta^*(t))\}_{2n,2n}}
\\
\le &\frac{C_2C_{Nn}Q_{Nnh}}{f(t)}\Big(1+\frac{1}{n}\Big),
\end{align*}
where $\{\bar V(t,\theta^*(t))\}_{2n,2n}
=2\sum_{i=1}^{n}\{\bar V(t,\theta^*(t))\}_{i,i}-\sum_{i=1}^{n}\sum_{j=1}^{2n-1}\{\bar V(t,\theta^*(t))\}_{i,j}$
and \begin{align*}
\tilde{\psi}_{2n,i}(t) =& \tilde{\psi}_{i,2n}(t) := \tilde{\psi}_{ii}(t) - \sum_{j=1,j \neq i}^{2n-1} \tilde{\psi}_{ij}(t),~~i\in[2n-1].
\end{align*}
For $i=j=n+1,\dots,2n-1,$ we have
\begin{align*}
\Big|\sum_{k=1}^{2n-1}\big\{\bar S(t,\theta^*(t))\big\}_{i,k}\tilde{\psi}_{kj}(t,\theta^*(t))\Big|
\le& \frac{|\tilde{\psi}_{ii}(t,\theta^*(t))|}{\{\bar V(t,\theta^*(t))\}_{i,i}}
+\frac{|\tilde{\psi}_{ii}(t,\theta^*(t))-\sum_{k=1,k\neq j}^{2n-1}\tilde{\psi}_{kj}(t,\theta^*(t))|}{\{\bar V(t,\theta^*(t))\}_{2n,2n}}
\\
=&\frac{|\tilde{\psi}_{ii}(t,\theta^*(t))|}{\{\bar V(t,\theta^*(t))\}_{i,i}}
\le  C_2C_{Nn}Q_{Nnh}/f(t).
\end{align*}
For $i\in[n],~j \in[n-1],~i\neq j,$ we have
\begin{align*}
\Big|\sum_{k=1}^{2n-1}\big\{\bar S(t,\theta^*(t))\big\}_{i,k}\tilde{\psi}_{kj}(t,\theta^*(t))\Big|
=& \frac{|\tilde{\psi}_{jj}(t,\theta^*(t))-\sum_{k=1, k\neq j}^{2n-1}\tilde{\psi}_{kj}(t,\theta^*(t))|}{\{\bar V(t,\theta^*(t))\}_{2n,2n}}\\
\le &  C_2C_{Nn}Q_{Nnh}/(nf(t)).
\end{align*}
For $i\in [n],~j=n+1,\dots,2n-1,~j\neq n+i,$ we have
\begin{align*}
\Big|\sum_{k=1}^{2n-1}\big\{\bar S(t,\theta^*(t))\big\}_{i,k}\tilde{\psi}_{kj}(t,\theta^*(t))\Big|
= \frac{|\tilde{\psi}_{ij}(t,\theta^*(t))|}{\{\bar V(t,\theta^*(t))\}_{i,i}}
\le C_2C_{Nn}Q_{Nnh}/(nf(t)).
\end{align*}
For $i=n+1,\dots,2n-1,j\in[n],$ we have
\begin{align*}
\Big|\sum_{k=1}^{2n-1}\big\{\bar S(t,\theta^*(t))\big\}_{i,k}\tilde{\psi}_{kj}(t,\theta^*(t))\Big|
\le & C_2C_{Nn}Q_{Nnh}/(nf(t)).
\end{align*}
%For $j=n+i,i\in [n-1],$ or $i=n+1,j=n+2,$ or $i=n+2,j=n+1,$ or $j=3,i \in [n-1],$
Otherwise, we have
\begin{align*}
\Big|\sum_{k=1}^{2n-1}\big\{\bar S(t,\theta^*(t))\big\}_{i,k}\tilde{\psi}_{kj}(t,\theta^*(t))\Big| =0.
\end{align*}
This implies that
\begin{align*}
\Big\|\bar S(t,\theta^*(t))\tilde{V}(t,\theta^*(t))-\frac{C_2C_{Nn}Q_{Nnh}}{f(t)}\mathbbm{1}\Big\|_{\max}\le 2C_2C_{Nn}Q_{Nnh}/(nf(t)),
\end{align*}
where $\mathbb{1}\in\mathbb{R}^{(2n-1)\times(2n-1)}$ is the identity matrix.
Thus, we have
\begin{align*}
&\Big\|\bar S(t,\theta^*(t))\tilde{V}(t,\theta^*(t))\bar S(t,\theta^*(t))\Big\|_{\max}\\
=&\Big\|\Big(\bar S(t,\theta^*(t))\tilde{V}(t,\theta^*(t))-\frac{C_2C_{Nn}Q_{Nnh}}{f(t)}\mathbbm{1}\Big)\bar S(t,\theta^*(t))
+\frac{C_2C_{Nn}Q_{Nnh}}{f(t)}\bar S(t,\theta^*(t))\Big\|_{\max}\\
\le& \Big\|\Big(\bar S(t,\theta^*(t))\tilde{V}(t,\theta^*(t))
-\frac{C_2C_{Nn}Q_{Nnh}}{f(t)}\mathbbm{1}\Big)\bar S(t,\theta^*(t))\Big\|_{\max}+\frac{C_2C_{Nn}Q_{Nnh}}{f(t)}\|\bar S(t,\theta^*(t))\Big\|_{\max}\\
\le &\Big\|\bar S(t,\theta^*(t))\tilde{V}(t,\theta^*(t))
-\frac{C_2C_{Nn}Q_{Nnh}}{f(t)}\mathbbm{1}\Big\|_{\max} \times \max_{1 \le j \le 2n-1}\sum_{k=1}^{2n-1}\Big|\big\{\bar S(t,\theta^*(t)) \big\}_{k,j}\Big|\\
&+\frac{C_2C_{Nn}Q_{Nnh}}{f(t)}\Big\|\bar S(t,\theta^*(t))\Big\|_{\max}\\
=&O(C_{Nn}Q^2_{Nnh}).
\end{align*}
That is,
\begin{align*}
\Big\|\bar S(t,\theta^*(t))\tilde{V}(t,\theta^*(t))\bar S(t,\theta^*(t))\Big\|_{\max}=O(C_{Nn}Q^2_{Nnh}).
\end{align*}
Similarly, we can show that
$$
\Big[o\big(C_{Nn}h\big)+O\big(Nnh^5C^2_{Nn}\big)\Big]
\Big\|\bar S(t,\theta^*(t))\mathbb{D}\bar S(t,\theta^*(t))\Big\|_{\max}
=O(hC_{Nn}Q^2_{Nnh}).
$$
These facts imply that
\begin{align*}
&\Big\|\bar S(t,\theta^*(t))\Omega(t,\theta^*(t))\bar S(t,\theta^*(t))
-k_{02}\bar S(t,\theta^*(t))\bar{V}(t,\theta^*(t))\bar S(t,\theta^*(t))\Big\|_{\max}\\
\le & \Big\|\bar S(t,\theta^*(t))\tilde{V}(t,\theta^*(t))\bar S(t,\theta^*(t))\Big\|_{\max}
+\Big[o\big(C_{Nn}h\big)+O\big(Nnh^5C^2_{Nn}\big)\Big]\Big\|\bar S(t,\theta^*(t))\mathbb{D}\bar S(t,\theta^*(t))\Big\|_{\max}\\
=&O(hC_{Nn}Q^2_{Nnh}).
\end{align*}
Thus, if $hC_{Nn}Q^2_{Nnh}=o(1),$ then we have
\begin{align*}
&\Big\|\bar S(t,\theta^*(t))\Omega(t,\theta^*(t))\bar S(t,\theta^*(t))
-k_{02}\bar S(t,\theta^*(t))\bar{V}(t,\theta^*(t))\bar S(t,\theta^*(t))\Big\|_{\max}=o(1).
\end{align*}
This completes the proof.
\end{proof}

\begin{proof}[Proof of Proposition \ref{proposition-1}]
%Define $S_0(t,\theta(t))=S(V_0(t,\theta(t))).$
%Since
%\begin{align*}
%V(t,\theta(t)) =\hat f(t) V_0(t,\theta(t))~~\text{and}~~
%\bar V(t,\theta(t))=f(t)V_0(t,\theta(t)),
%\end{align*}
%we can obtain
As in the proof of Lemma 10, we have
\begin{align*}
S(t,\theta(t))=\frac{1}{\hat f(t)}S_0(t,\theta(t))~~\text{and}~~
\bar S(t,\theta(t))=\frac{1}{f(t)}S_0(t,\theta(t)).
\end{align*}
Then, a direct calculation yields
\begin{align*}
&\Big\|\hat\Sigma(t,\hat\theta(t))-\Sigma(t,\theta^*(t))\Big\|_{\max}
\\
=&\Big\|S(t,\hat\theta(t))\hat V(t,\hat\theta(t)) S(t,\hat\theta(t))^\top-k_{02}\bar S(t,\theta^*(t))\bar V(t,\theta^*(t))\bar S(t,\theta^*(t))^\top\Big\|_{\max}
\\
=&\Big\|\frac{1}{\hat{f}^2(t)}S_0(t,\hat\theta(t))\hat V(t,\hat\theta(t)) S_0(t,\hat\theta(t))^\top-\frac{1}{f^2(t)}k_{02}S_0(t,\hat\theta(t))\bar V(t,\theta^*(t)) S_0(t,\hat\theta(t))^\top\Big\|_{\max}
\\
\le&\Big\|\Big(\frac{1}{\hat{f}^2(t)}-\frac{1}{f^2(t)}\Big)S_0(t,\hat\theta(t))\hat V(t,\hat\theta(t)) S_0(t,\hat\theta(t))^\top\Big\|_{\max}
\\
&+\Big\|\frac{1}{f^2(t)}S_0(t,\hat\theta(t))\big(\hat V(t,\hat\theta(t))-k_{02}\bar V(t,\theta^*(t))\big) S_0(t,\hat\theta(t))^\top\Big\|_{\max}.
\end{align*}
We next show that
\begin{align*}
\Big\|\Big(\frac{1}{\hat{f}^2(t)}-\frac{1}{f^2(t)}\Big)S_0(t,\hat\theta(t))\hat V(t,\hat\theta(t))S_0(t,\hat\theta(t))^\top\Big\|_{\max}
=&O_p\Bigg(Q_{Nnh}\Big\{\sqrt{\frac{\log(1/h)}{Nh}}+h^2\Big\}\Bigg),\\
f^{-2}(t)\Big\|S_0(t,\hat\theta(t))\big(\hat V(t,\hat\theta(t))-k_{02}\bar V(t,\theta^*(t))\big) S_0(t,\hat\theta(t))^\top\Big\|_{\max}
=&O_p\Bigg(Q_{Nnh}\Big\{\sqrt{\frac{\log(1/h)}{Nh}}+h^2\Big\}\Bigg).
\end{align*}
Write
\begin{align*}
S_0(t,\hat\theta(t))\hat V(t,\hat\theta(t))
=\Big(\sum_{k\in[2n-1]}\big\{ S_0(t,\hat\theta(t))\big\}_{i,k}\big\{ \hat V(t,\hat\theta(t))\big\}_{k,j}\Big).
\end{align*}
Then, if $i=j \in [n],$ we have
\begin{align*}
&\Big|\sum_{k\in[2n-1]}\big\{ S_0(t,\hat\theta(t))\big\}_{i,k}\big\{ \hat V(t,\hat\theta(t))\big\}_{k,j}\Big|
\\
\le& \frac{|\big\{ \hat V(t,\hat\theta(t))\big\}_{i,i}|}{ \big\{ V_0(t,\hat\theta(t))\big\}_{i,i}}
+\frac{|\big\{ \hat V(t,\hat\theta(t))\big\}_{i,i}-\sum_{k=1,k\neq j}^{2n-1}\big\{ \hat V(t,\hat\theta(t))\big\}_{k,j}|}{ \big\{ V_0(t,\hat\theta(t))\big\}_{2n,2n}}\\
=&\frac{|\big\{ \hat V(t,\hat\theta(t))\big\}_{i,i}|}{ \big\{ V_0(t,\hat\theta(t))\big\}_{i,i}}
+\frac{|\big\{ \hat V(t,\hat\theta(t))\big\}_{2n,j}|}{ \big\{ V_0(t,\hat\theta(t))\big\}_{2n,2n}}.
\end{align*}
It can be shown that with probability tending to 1,
$$|\big\{ \hat V(t,\hat\theta(t))\big\}_{i,i}|
\le 2|\mathbb{E}\big\{ \hat V(t,\theta^*(t))\big\}_{i,i}|.
$$
Then, we have
\begin{align*}
\Big|\sum_{k=1}^{2n-1}\big\{ S_0(t,\hat\theta(t))\big\}_{i,k}\big\{ \hat V(t,\hat\theta(t))\big\}_{k,j}\Big|=O_p(1+1/n).
\end{align*}
Similarly, we can show
\begin{align*}
|\sum_{k=1}^{2n-1}\big\{ S_0(t,\hat\theta(t))\big\}_{i,k}\big\{ \hat V(t,\hat\theta(t))\big\}_{k,j}|
%\le& \frac{|\big\{ \hat V(t,\hat\theta(t))\big\}_{i,i}|}{ \big\{ V_0(t,\hat\theta(t))\big\}_{i,i}}+\frac{|\big\{ \hat V(t,\hat\theta(t))\big\}_{i,i}-\sum_{k=1,k\neq j}^{2n-1}\big\{ \hat V(t,\hat\theta(t))\big\}_{k,j}|}{\big\{ V_0(t,\hat\theta(t))\big\}_{2n,2n}}\\
%=&\frac{|\big\{ \hat V(t,\hat\theta(t))\big\}_{i,i}|}{ \big\{ V_0(t,\hat\theta(t))\big\}_{i,i}}\\
=& O_p(1),~~i=j=n+1,\dots,2n-1,\\
|\sum_{k=1}^{2n-1}\big\{ S_0(t,\hat\theta(t))\big\}_{i,k}\big\{ \hat V(t,\hat\theta(t))\big\}_{k,j}|
%=& \frac{|\big\{ \hat V(t,\hat\theta(t))\big\}_{j,j}-\sum_{k=1,k\neq j}^{2n-1}\big\{ \hat V(t,\hat\theta(t))\big\}_{k,j}|}{\big\{ V_0(t,\hat\theta(t))\big\}_{2n,2n}}\\
%=&\frac{|\big\{ \hat V(t,\hat\theta(t))\big\}_{2n,j}|}{\big\{ V_0(t,\hat\theta(t))\big\}_{2n,2n}}\\
= &O_p(1/n),~~i\in[n],j \in[n-1],i\neq j,\\
|\sum_{k=1}^{2n-1}\big\{ S_0(t,\hat\theta(t))\big\}_{i,k}\big\{ \hat V(t,\hat\theta(t))\big\}_{k,j}|
%= \frac{|\big\{ \hat V(t,\hat\theta(t))\big\}_{i,j}|}{\big\{ V_0(t,\hat\theta(t))\big\}_{i,i}}
= &O_p(1/n),~~i\in [n],j=n+1,\dots,2n-1,j\neq n+i,\\
|\sum_{k=1}^{2n-1}\big\{ S_0(t,\hat\theta(t))\big\}_{i,k}\big\{ \hat V(t,\hat\theta(t))\big\}_{k,j}|
= & O(1/n),~~i=n+1,\dots,2n-1,j\in[n],\\
%if $j=n+i,i\in [n-1],$ or $i=n+1,j=n+2,$ or $i=n+2,j=n+1,$ or $j=3,i \in [n-1],$
|\sum_{k=1}^{2n-1}\big\{ S_0(t,\hat\theta(t))\big\}_{i,k}\big\{ \hat V(t,\hat\theta(t))\big\}_{k,j}|=&0,~~\text{otherwise}.
\end{align*}
These facts imply that
\begin{align*}
\big\| S_0(t,\hat\theta(t))\hat V(t,\hat\theta(t))-C_4\mathbbm{1}\big\|_{\max}=O(1/n).
\end{align*}
with probability tending to 1 for some constant $C_4.$
By Theorem 6 in \cite{masry1996multivariate}, we have
$$\sup_{t \in[a,b]}|\hat f(t) -f(t)|=O\Big(\Big\{\frac{\log(1/h)}{Nh}\Big\}^{1/2}+h^2\Big).$$
This implies that
\begin{align*}
&\Big\|\frac{1}{\hat{f}^2(t)}-\frac{1}{f^2(t)}\Big|\times\Big\|S_0(t,\hat\theta(t))\hat V(t,\hat\theta(t)) S_0(t,\hat\theta(t))^\top\Big\|_{\max}
\\
\le&\Big|\frac{1}{\hat{f}^2(t)}-\frac{1}{f^2(t)}\Big| \times
\Big\|\big( S_0(t,\hat\theta(t))\hat V(t,\hat\theta(t))-C_4\mathbbm{1}\big)S_0(t,\hat\theta(t))
+C_4\mathbbm{1} S_0(t,\hat\theta(t))\Big\|_{\max}\\
\le& \Big|\frac{1}{\hat{f}^2(t)}-\frac{1}{f^2(t)}\Big| \times
\Big\|\big( S_0(t,\hat\theta(t))\hat V(t,\hat\theta(t))
-C_4\mathbbm{1}\big)S_0(t,\hat\theta(t))\Big\|_{\max}\\
&+\Big|\frac{1}{\hat{f}^2(t)}-\frac{1}{f^2(t)}\Big| \times C_4\Big\|S_0(t,\hat\theta(t))\Big\|_{\max}
\\
\le &\Big|\frac{1}{\hat{f}^2(t)}-\frac{1}{f^2(t)}\Big| \times
\Big\| S_0(t,\hat\theta(t))\hat V(t,\hat\theta(t))-C_4\mathbbm{1}\Big\|_{\max} \times \max_{j \in [2n-1]}\sum_{k=1}^{2n-1}\Big|\big\{S_0(t,\hat\theta(t)) \big\}_{k,j}\Big|
\\
&+\Big|\frac{1}{\hat{f}^2(t)}-\frac{1}{f^2(t)}\Big| \times
C_4\|S_0(t,\hat\theta(t))\Big\|_{\max}
\\
=&O_p\Big(Q_{Nnh}\big\{\sqrt{\frac{\log(1/h)}{Nh}}+h^2\big\}\Big).
\end{align*}
By some arguments similar to the above, we can show
\begin{align*}
\Big\|\frac{1}{f^2(t)}S_0(t,\hat\theta(t))\big(\hat V(t,\hat\theta(t))-k_{02}\bar V(t,\theta^*(t))\big) S_0(t,\hat\theta(t))^\top\Big\|_{\max}
=O_p\Big(Q_{Nnh}\bigg\{\sqrt{\frac{\log(1/h)}{Nh}}+h^2\bigg\}\Big).
\end{align*}
Thus, if $Nh \to \infty,$ then
$$\|\hat\Sigma(t,\hat\theta(t))-\Sigma(t,\theta^*(t))\|_{\max}=o_p(1)$$
under the conditions of Theorems \ref{theorem-1} and \ref{theorem-2}.
It completes the proof.
\end{proof}

\end{document}